\documentclass[11pt,letterpaper]{amsart}

\usepackage[left=1in,right=1in,top=1in,bottom=1in]{geometry}
\usepackage[foot]{amsaddr}
\usepackage{amsmath}
\usepackage{amsfonts}
\usepackage{amsthm}
\usepackage{tikz}
\usepackage{subcaption}
\pgfdeclarelayer{bg}    
\pgfsetlayers{bg,main}
\usepackage{thmtools} 
\usepackage{thm-restate}

\usepackage{hyperref}

\usepackage{todonotes}

\newtheorem{theorem}{Theorem}

\newtheorem{lemma}[theorem]{Lemma}

\newcommand{\oa}[1]{\ensuremath{\vec{#1}}}
\newcommand{\iv}[1]{\ensuremath{[\![{#1}]\!]}}
\newcommand{\ham}[1]{\ensuremath{\mathrm{H}[{#1}]}}
\newcommand{\hp}[1]{\ensuremath{\mathrm{P}[{#1}]}}

\newcommand{\td}{\ensuremath{\mathrm{td}}}
\newcommand{\pd}{\ensuremath{\mathrm{pd}}}
\newcommand{\bd}{\ensuremath{\mathrm{bd}}}

\newcommand{\qedwhite}{\hfill \ensuremath{\Box}}
\newcommand{\mineq}{\ensuremath{:=_{\operatorname{min}}}}
\newcommand{\acceq}{\ensuremath{:=_+}}

\newcommand{\cubegraph}{%
	\tikzstyle{node_stylek} = [circle,draw=black,fill=black, inner sep=0pt]
	\tikzstyle{node_stylew} = [circle,draw=black,fill=black, inner sep=0pt]
	
	\tikzstyle{nhidden_style} = [inner sep=0pt, minimum size=0pt]
	\tikzstyle{s} = [draw=black, line width=1, thin]
	\tikzstyle{t} = [draw=black, line width=1.5]
	
	\node[nhidden_style](h11) at (-2,-2) {};
	\node[nhidden_style](h12) at (-2,2) {};
	\node[nhidden_style](h21) at (2,-2) {};
	\node[nhidden_style](h22) at (2,2) {};
	\node[nhidden_style](b11) at (-1,-1) {};
	\node[nhidden_style](b12) at (-1,1) {};
	
	\node[nhidden_style](b21) at (1,-1) {};
	\node[nhidden_style](b22) at (1,1) {};
	
	\begin{pgfonlayer}{bg} 
		
		\draw[s,-]  (h11.center) to [out=100,in=260] (h12.center);
		\draw[s,-]  (h12.center) to [out=10,in=170] (h22.center);
		\draw[s,-]  (h21.center) to [out=80,in=280] (h22.center);
		\draw[s,-]  (h11.center) to [out=350,in=190] (h21.center);
		
		\draw[s,-]  (b21.center) to [out=85,in=275] (b22.center);
		\draw[s,-]  (b12.center) to [out=265,in=95] (b11.center);
		
		\draw[s,-]  (b11.center) to [out=350,in=190] (b21.center);
		\draw[s,-]  (b12.center) to [out=10,in=170] (b22.center);   
		\draw[s,-]  (b11.center) edge (h11.center);
		\draw[s,-]  (h12.center) edge (b12.center);
		\draw[s,-]  (h21.center) edge (b21.center);
		\draw[s,-]  (b22.center) edge (h22.center);
		
	\end{pgfonlayer}
}

\newcommand{\pfaffian}{%
    \tikzstyle{node_stylek} = [rectangle,draw=black,fill=black, inner sep=0pt, minimum size=4pt, line width=.5]
    \tikzstyle{node_stylew} = [circle,draw=black,fill=white, inner sep=0pt, minimum size=4pt, line width=.5]

    \tikzstyle{nhidden_style} = [inner sep=0pt, minimum size=0pt]
    \tikzstyle{s} = [draw=black, line width=1, thin, dashed]
    \tikzstyle{t} = [draw=black, line width=6, thick]
 
    \node[node_stylek](h11) at (-3,-2) {};
    \node[node_stylew](h12) at (-3,2) {};
    \node[node_stylew](h21) at (2,-2) {};
    \node[node_stylek](h22) at (2,2) {};
    \node[node_stylew](b11) at (-2,-1) {};
    \node[node_stylek](b12) at (-2,1) {};
    \node[node_stylek](e11) at (-1,-1.1) {};
    \node[node_stylew](e12) at (-1,1.1) {};
    \node[node_stylew](e21) at (0,-1.1) {};
    \node[node_stylek](e22) at (0,1.1) {};
    \node[node_stylek](b21) at (1,-1) {};
    \node[node_stylew](b22) at (1,1) {};
    
    \begin{pgfonlayer}{bg} 
    
    \draw[->]  (h11) to [out=100,in=260] (h12);
    \draw[->]  (h12) to [out=10,in=170] (h22);
    \draw[->]  (h21) to [out=80,in=280] (h22);
    \draw[->]  (h11) to [out=350,in=190] (h21);
    
    \draw[->]  (b11) edge (e11);
    \draw[->]  (e11) edge (e21);
    \draw[->]  (e21) edge (b21);
    \draw[->]  (b21) to [out=85,in=275] (b22);
  
    \draw[->]  (e22) edge (b22);
    \draw[->]  (e12) edge (e22);
    \draw[->]  (b12) edge (e12);
    \draw[->]  (b12) to [out=265,in=95] (b11);
    
    \draw[->]  (e11) edge (e12);
    \draw[->]  (e22) edge (e21);
    
    \draw[->]  (b11) edge (h11);
    \draw[->]  (h12) edge (b12);
    \draw[->]  (h21) edge (b21);
    \draw[->]  (b22) edge (h22);
    
    \end{pgfonlayer}
}

\newcommand{\coloredone}{%
    \definecolor{orange}{rgb}{1.0,0.64,0}

    \tikzstyle{node_styleb} = [circle,draw=black,fill=blue, inner sep=0pt, minimum size=4pt, line width=.5]
    \tikzstyle{node_styleo} = [circle,draw=black,fill=orange, inner sep=0pt, minimum size=4pt, line width=.5]
    \tikzstyle{node_stylebB} = [rectangle,draw=black,fill=blue, inner sep=0pt, minimum size=4pt, line width=.5]
    \tikzstyle{node_styleoB} = [rectangle,draw=black,fill=orange, inner sep=0pt, minimum size=4pt, line width=.5]

    \tikzstyle{nhidden_style} = [inner sep=0pt, minimum size=0pt]
    \tikzstyle{t} = [draw=black, line width=1.5]
 
    \node[node_styleoB](h11) at (-3,-2) {};
    \node[node_styleb](h12) at (-3,2) {};
    \node[node_styleo](h21) at (2,-2) {};
    \node[node_styleoB](h22) at (2,2) {};
    \node[node_styleb](b11) at (-2,-1) {};
    \node[node_styleoB](b12) at (-2,1) {};
    \node[node_styleoB](e11) at (-1,-1.1) {};
    \node[node_styleb](e12) at (-1,1.1) {};
    \node[node_styleb](e21) at (0,-1.1) {};
    \node[node_styleoB](e22) at (0,1.1) {};
    \node[node_styleoB](b21) at (1,-1) {};
    \node[node_styleb](b22) at (1,1) {};
    
    \begin{pgfonlayer}{bg} 
    
    \draw[t,->]  (h11) to [out=100,in=260] (h12);
    \draw[->]  (h12) to [out=10,in=170] (h22);
    \draw[t,->]  (h22) to [out=280,in=80] (h21);
    \draw[t,->]  (h21) to [out=190,in=350] (h11);
    
    \draw[t,->]  (b11) edge (e11);
    \draw[->]  (e11) edge (e21);
    \draw[t,->]  (e21) edge (b21);
    \draw[t,->]  (b21) to [out=85,in=275] (b22);
  
    \draw[->]  (e22) edge (b22);
    \draw[t,->]  (e12) edge (e22);
    \draw[->]  (b12) edge (e12);
    \draw[t,->]  (b12) to [out=265,in=95] (b11);
    
    \draw[t,->]  (e11) edge (e12);
    \draw[t,->]  (e22) edge (e21);
    
    \draw[->]  (b11) edge (h11);
    \draw[t,->]  (h12) edge (b12);
    \draw[->]  (b21) edge (h21);
    \draw[t,->]  (b22) edge (h22);
    
    \end{pgfonlayer}
}

\newcommand{\coloredtwo}{%
    \definecolor{orange}{rgb}{1.0,0.64,0}
    \tikzstyle{node_styleb} = [circle,draw=black,fill=blue, inner sep=0pt, minimum size=4pt, line width=.5]
    \tikzstyle{node_styleo} = [circle,draw=black,fill=orange, inner sep=0pt, minimum size=4pt, line width=.5]
    \tikzstyle{node_stylebB} = [rectangle,draw=black,fill=blue, inner sep=0pt, minimum size=4pt, line width=.5]
    \tikzstyle{node_styleoB} = [rectangle,draw=black,fill=orange, inner sep=0pt, minimum size=4pt, line width=.5]

    \tikzstyle{nhidden_style} = [inner sep=0pt, minimum size=0pt]
    \tikzstyle{t} = [draw=black, line width=1.5]
 
    \node[node_styleoB](h11) at (-3,-2) {};
    \node[node_styleb](h12) at (-3,2) {};
    \node[node_styleo](h21) at (2,-2) {};
    \node[node_styleoB](h22) at (2,2) {};
    \node[node_styleb](b11) at (-2,-1) {};
    \node[node_styleoB](b12) at (-2,1) {};
    \node[node_styleoB](e11) at (-1,-1.1) {};
    \node[node_styleo](e12) at (-1,1.1) {};
    \node[node_styleb](e21) at (0,-1.1) {};
    \node[node_styleoB](e22) at (0,1.1) {};
    \node[node_styleoB](b21) at (1,-1) {};
    \node[node_styleo](b22) at (1,1) {};
    
    \begin{pgfonlayer}{bg} 
    
    \draw[t,->]  (h11) to [out=100,in=260] (h12);
    \draw[t,->]  (h12) to [out=10,in=170] (h22);
    \draw[->]  (h22) to [out=280,in=80] (h21);
    \draw[t,->]  (h21) to [out=190,in=350] (h11);
    
    \draw[t,->]  (b11) edge (e11);
    \draw[t,->]  (e11) edge (e21);
    \draw[t,->]  (e21) edge (b21);
    \draw[->]  (b22) to [out=275,in=85] (b21);
  
    \draw[t,->]  (b22) edge (e22);
    \draw[t,->]  (e22) edge (e12);
    \draw[t,->]  (e12) edge (b12);
    \draw[t,->]  (b12) to [out=265,in=95] (b11);
    
    \draw[->]  (e12) edge (e11);
    \draw[->]  (e22) edge (e21);
    
    \draw[->]  (b11) edge (h11);
    \draw[->]  (h12) edge (b12);
    \draw[t,->]  (b21) edge (h21);
    \draw[t,->]  (h22) edge (b22);
    
    \end{pgfonlayer}
}

\newcommand{\matching}{%
    \definecolor{orange}{rgb}{1.0,0.64,0}
    \tikzstyle{node_styleb} = [circle,draw=black,fill=blue, inner sep=0pt, minimum size=4pt, line width=.5]
    \tikzstyle{node_styleo} = [circle,draw=black,fill=orange, inner sep=0pt, minimum size=4pt, line width=.5]
    \tikzstyle{node_stylebB} = [rectangle,draw=black,fill=blue, inner sep=0pt, minimum size=4pt, line width=.5]
    \tikzstyle{node_styleoB} = [rectangle,draw=black,fill=orange, inner sep=0pt, minimum size=4pt, line width=.5]

    \tikzstyle{nhidden_style} = [inner sep=0pt, minimum size=0pt]
    \tikzstyle{s} = [draw=black, line width=1, thin, dashed]
    \tikzstyle{t} = [draw=black, line width=1.5]
 
    \node[node_styleoB](h11o) at (-2.85,-2) {};
    \node[node_stylebB](h11b) at (-3.15,-2) {};
    
    \node[node_styleb](h12) at (-3.15,2) {};
    \node[node_styleo](h21) at (2.4,-2) {};
    \node[node_styleoB](h22o) at (2.1,2) {};
     \node[node_stylebB](h22b) at (2.4,2) {};
    \node[node_styleb](b11) at (-2,-1) {};
    \node[node_styleoB](b12o) at (-2.3,1) {};
    \node[node_stylebB](b12b) at (-2,1) {};
    
    \node[node_stylebB](e11b) at (-0.9,-1.1) {};
    \node[node_styleoB](e11o) at (-1.2,-1.1) {};
    
    \node[node_styleb](e12) at (-0.9,1.1) {};
    \node[node_styleb](e21) at (0.2,-1.1) {};
    \node[node_styleoB](e22o) at (-0.1,1.1) {};
    \node[node_stylebB](e22b) at (0.2,1.1) {};
    
    \node[node_styleoB](b21o) at (1.0,-1) {};
    \node[node_stylebB](b21b) at (1.3,-1) {};

    \node[node_styleb](b22) at (1.3,1) {};
    
    \begin{pgfonlayer}{bg} 
    
    \draw[t,->]  (h12) to [out=260,in=100] (h11b);
    \draw[->]  (h12) to [out=10,in=170] (h22o);
    \draw[t,->]  (h21) to [out=80,in=280] (h22b);
    \draw[t,->]  (h21)  to [out=190,in=350] (h11o);
    
    \draw[t,->]  (b11) edge (e11o);
    \draw[->]  (e21) edge (e11b);
    \draw[t,->]  (e21) edge (b21o);
    \draw[t,->]  (b22) to [out=275,in=85] (b21b);
  
    \draw[->]  (b22) edge (e22b);
    \draw[t,->]  (e12) edge (e22o);
    \draw[->]  (e12) edge (b12b);
    \draw[t,->]  (b11) to [out=95,in=265]  (b12b);
    
    \draw[t,->]  (e12) edge (e11b);
    \draw[t,->]  (e21) edge (e22b);
    
    \draw[->]  (b11) edge (h11o);
    \draw[t,->]  (h12) edge (b12o);
    \draw[->]  (h21) edge (b21b);
    \draw[t,->]  (b22) edge (h22o);
    
    \end{pgfonlayer}
}

\newcommand{\overlayerI}{%
    \definecolor{scarlet}{rgb}{1.0,0.14,0}
    \definecolor{indigo}{rgb}{0.57,0,1.0}
    \definecolor{orange}{rgb}{1.0,0.64,0}
    \definecolor{green}{rgb}{0.0,0.5,0.0}
    \tikzstyle{node_styleb} = [circle,draw=black,fill=blue, inner sep=0pt, minimum size=4pt, line width=.5]
    \tikzstyle{node_styleo} = [circle,draw=black,fill=orange, inner sep=0pt, minimum size=4pt, line width=.5]
    \tikzstyle{node_stylebB} = [rectangle,draw=black,fill=blue, inner sep=0pt, minimum size=4pt, line width=.5]
    \tikzstyle{node_styleoB} = [rectangle,draw=black,fill=orange, inner sep=0pt, minimum size=4pt, line width=.5]

    \tikzstyle{nhidden_style} = [inner sep=0pt, minimum size=0pt]
    \tikzstyle{t} = [draw=black, line width=1.5]
    \tikzstyle{s} = [draw=green, line width=1.5]
 
    \node[nhidden_style](b) at (0,-3) {};

    \node[node_styleoB](h11o) at (-2.85,-2) {};
    \node[node_stylebB](h11b) at (-3.15,-2) {};
    
    \node[node_styleb](h12) at (-3.15,2) {};
    \node[node_styleo](h21) at (2.4,-2) {};
    \node[node_styleoB](h22o) at (2.1,2) {};
     \node[node_stylebB](h22b) at (2.4,2) {};
    \node[node_styleb](b11) at (-2,-1) {};
    \node[node_styleoB](b12o) at (-2.3,1) {};
    \node[node_stylebB](b12b) at (-2,1) {};
    
    \node[node_stylebB](e11b) at (-0.9,-1.1) {};
    \node[node_styleoB](e11o) at (-1.2,-1.1) {};
    
    \node[node_styleb](e12) at (-0.9,1.1) {};
    \node[node_styleb](e21) at (0.2,-1.1) {};
    \node[node_styleoB](e22o) at (-0.1,1.1) {};
    \node[node_stylebB](e22b) at (0.2,1.1) {};
    
    \node[node_styleoB](b21o) at (1.0,-1) {};
    \node[node_stylebB](b21b) at (1.3,-1) {};

    \node[node_styleb](b22) at (1.3,1) {};
    
    \begin{pgfonlayer}{bg} 

    \draw[t,->]  (h12) to [out=260,in=100]  (h11b);
    \draw[->]  (h12) to [out=10,in=170]  (h22o);
    \draw[t,->]  (h21) to [out=80,in=280]  (h22b);
    \draw[t,->]  (h21) to [out=190,in=350] (h11o);
    
    \draw[t,->]  (b11) edge (e11o);
    \draw[->]  (e21) edge (e11b);
    \draw[t,->]  (e21) edge (b21o);
    \draw[t,->]  (b22) to [out=275,in=85] (b21b);
  
    \draw[->]  (b22) edge (e22b);
    \draw[t,->]  (e12) edge (e22o);
    \draw[->]  (e12) edge (b12b);
    \draw[t,->]  (b11) to [out=95,in=265] (b12b);
    
    \draw[t,->]  (e12) edge (e11b);
    \draw[t,->]  (e21) edge (e22b);
    
    \draw[->]  (b11) edge (h11o);
    \draw[t,->]  (h12) edge (b12o);
    \draw[->]  (h21) edge (b21b);
    \draw[t,->]  (b22) edge (h22o);
    
    \draw[s,->] (h12) to [out=330, in=100] (b22);
     
    \draw[s,->] (h21) to [out=130, in=350] (b22);
    \draw[s,->] (b22) to [out=190, in=60] (e21);
    \draw[s,->] (e21) to [out=160, in=320] (e12);
    \draw[s,->] (e12) to [out=190, in=60] (b11);
    \draw[s,->] (b11) to [out=270, in=150] (h21);

    \end{pgfonlayer}
}

\newcommand{\overlayerII}{%
    \definecolor{scarlet}{rgb}{1.0,0.14,0}
    \definecolor{indigo}{rgb}{0.57,0,1.0}
    \definecolor{orange}{rgb}{1.0,0.64,0}
    \definecolor{green}{rgb}{0.0,0.5,0.0}
    \tikzstyle{node_styleb} = [circle,draw=black,fill=blue, inner sep=0pt, minimum size=4pt, line width=.5]
    \tikzstyle{node_styleo} = [circle,draw=black,fill=orange, inner sep=0pt, minimum size=4pt, line width=.5]
    \tikzstyle{node_stylebB} = [rectangle,draw=black,fill=blue, inner sep=0pt, minimum size=4pt, line width=.5]
    \tikzstyle{node_styleoB} = [rectangle,draw=black,fill=orange, inner sep=0pt, minimum size=4pt, line width=.5]

    \tikzstyle{nhidden_style} = [inner sep=0pt, minimum size=0pt]
    \tikzstyle{t} = [draw=black, line width=1.5]
    \tikzstyle{s} = [draw=green, line width=1.5]
 
   \node[nhidden_style](b) at (0,-3) {};

    \node[node_styleoB](h11o) at (-2.85,-2) {};
    \node[node_stylebB](h11b) at (-3.15,-2) {};
    
    \node[node_styleb](h12) at (-3.15,2) {};
    \node[node_styleo](h21) at (2.4,-2) {};
    \node[node_styleoB](h22o) at (2.1,2) {};
     \node[node_stylebB](h22b) at (2.4,2) {};
    \node[node_styleb](b11) at (-2,-1) {};
    \node[node_styleoB](b12o) at (-2.3,1) {};
    \node[node_stylebB](b12b) at (-2,1) {};
    
    \node[node_stylebB](e11b) at (-0.9,-1.1) {};
    \node[node_styleoB](e11o) at (-1.2,-1.1) {};
    
    \node[node_styleb](e12) at (-0.9,1.1) {};
    \node[node_styleb](e21) at (0.2,-1.1) {};
    \node[node_styleoB](e22o) at (-0.1,1.1) {};
    \node[node_stylebB](e22b) at (0.2,1.1) {};
    
    \node[node_styleoB](b21o) at (1.0,-1) {};
    \node[node_stylebB](b21b) at (1.3,-1) {};

    \node[node_styleb](b22) at (1.3,1) {};
    
    \begin{pgfonlayer}{bg} 
    
    \draw[t,->]  (h12) to [out=260,in=100]  (h11b);
    \draw[->]  (h12)  to [out=10,in=170]   (h22o);
    \draw[t,->]  (h21) to [out=80,in=280]  (h22b);
    \draw[->]  (h21) to [out=190,in=350]  (h11o);
    
    \draw[t,->]  (b11) edge (e11o);
    \draw[t,->]  (e21) edge (e11b);
    \draw[t,->]  (e21) edge (b21o);
    \draw[->]  (b22) to [out=275,in=85]  (b21b);
  
    \draw[t,->]  (b22) edge (e22b);
    \draw[t,->]  (e12) edge (e22o);
    \draw[t,->]  (e12) edge (b12b);
    \draw[->]  (b11)  to [out=95,in=265] (b12b);
    
    \draw[->]  (e12) edge (e11b);
    \draw[->]  (e21) edge (e22b);
    
    \draw[t,->]  (b11) edge (h11o);
    \draw[t,->]  (h12) edge (b12o);
    \draw[t,->]  (h21) edge (b21b);
    \draw[t,->]  (b22) edge (h22o);
    
    \draw[s,->] (h12) to [out=330, in=100] (b22);
     
    \draw[s,->] (b22) to [out=350, in=130] (h21);
    \draw[s,->] (e21) to [out=60, in=190] (b22);
    \draw[s,->] (e12) to [out=320, in=160] (e21);
    \draw[s,->] (b11) to [out=60, in=190] (e12);
    \draw[s,->] (h21) to [out=150, in=270] (b11);

    \end{pgfonlayer}
}

\newcommand{\solI}{%
    \definecolor{orange}{rgb}{1.0,0.64,0}

    \tikzstyle{node_styleb} = [circle,draw=black,fill=blue, inner sep=0pt, minimum size=4pt, line width=.5]
    \tikzstyle{node_styleo} = [circle,draw=black,fill=orange, inner sep=0pt, minimum size=4pt, line width=.5]
    \tikzstyle{node_stylebB} = [rectangle,draw=black,fill=blue, inner sep=0pt, minimum size=4pt, line width=.5]
    \tikzstyle{node_styleoB} = [rectangle,draw=black,fill=orange, inner sep=0pt, minimum size=4pt, line width=.5]

    \tikzstyle{nhidden_style} = [inner sep=0pt, minimum size=0pt]
    \tikzstyle{t} = [draw=black, line width=1.5]
 
   \node[nhidden_style](b) at (0,-3) {};

    \node[node_styleoB](h11) at (-3,-2) {};
    \node[node_styleb](h12) at (-3,2) {};
    \node[node_styleo](h21) at (2,-2) {};
    \node[node_styleoB](h22) at (2,2) {};
    \node[node_styleb](b11) at (-2,-1) {};
    \node[node_styleoB](b12) at (-2,1) {};
    \node[node_styleoB](e11) at (-1,-1.1) {};
    \node[node_styleb](e12) at (-1,1.1) {};
    \node[node_styleb](e21) at (0,-1.1) {};
    \node[node_styleoB](e22) at (0,1.1) {};
    \node[node_styleoB](b21) at (1,-1) {};
    \node[node_styleb](b22) at (1,1) {};
    
    \begin{pgfonlayer}{bg} 
    
    \draw[t,->]  (h11) to [out=100,in=260] (h12);
    \draw[->]  (h12) to [out=10,in=170]  (h22);
    \draw[t,->]  (h22) to [out=280,in=80] (h21);
    \draw[t,->]  (h21) to [out=190,in=350] (h11);
    
    \draw[t,->]  (b11) edge (e11);
    \draw[->]  (e11) edge (e21);
    \draw[t,->]  (e21) edge (b21);
    \draw[t,->]  (b21) to [out=85,in=275] (b22);
  
    \draw[->]  (e22) edge (b22);
    \draw[t,->]  (e12) edge (e22);
    \draw[->]  (b12) edge (e12);
    \draw[t,->]  (b12)  to [out=265,in=95] (b11);
    
    \draw[t,->]  (e11) edge (e12);
    \draw[t,->]  (e22) edge (e21);
    
    \draw[->]  (b11) edge (h11);
    \draw[t,->]  (h12) edge (b12);
    \draw[->]  (b21) edge (h21);
    \draw[t,->]  (b22) edge (h22);
    
    \end{pgfonlayer}
}

\newcommand{\solII}{%
    \definecolor{orange}{rgb}{1.0,0.64,0}

    \tikzstyle{node_styleb} = [circle,draw=black,fill=blue, inner sep=0pt, minimum size=4pt, line width=.5]
    \tikzstyle{node_styleo} = [circle,draw=black,fill=orange, inner sep=0pt, minimum size=4pt, line width=.5]
    \tikzstyle{node_stylebB} = [rectangle,draw=black,fill=blue, inner sep=0pt, minimum size=4pt, line width=.5]
    \tikzstyle{node_styleoB} = [rectangle,draw=black,fill=orange, inner sep=0pt, minimum size=4pt, line width=.5]

    \tikzstyle{nhidden_style} = [inner sep=0pt, minimum size=0pt]
    \tikzstyle{t} = [draw=black, line width=1.5]
 
    \node[nhidden_style](b) at (0,-3) {};
 
    \node[node_styleoB](h11) at (-3,-2) {};
    \node[node_styleb](h12) at (-3,2) {};
    \node[node_styleo](h21) at (2,-2) {};
    \node[node_styleoB](h22) at (2,2) {};
    \node[node_styleb](b11) at (-2,-1) {};
    \node[node_styleoB](b12) at (-2,1) {};
    \node[node_stylebB](e11) at (-1,-1.1) {};
    \node[node_styleb](e12) at (-1,1.1) {};
    \node[node_styleb](e21) at (0,-1.1) {};
    \node[node_styleoB](e22) at (0,1.1) {};
    \node[node_stylebB](b21) at (1,-1) {};
    \node[node_styleb](b22) at (1,1) {};
    
    \begin{pgfonlayer}{bg} 
    
    \draw[t,->]  (h11) to [out=100,in=260] (h12);
    \draw[->]  (h12) to [out=10,in=170] (h22);
    \draw[t,->]  (h22) to [out=280,in=80]  (h21);
    \draw[->]  (h21)  to [out=190,in=350]  (h11);
    
    \draw[t,<-]  (b11) edge (e11);
    \draw[t,<-]  (e11) edge (e21);
    \draw[t,<-]  (e21) edge (b21);
    \draw[<-]  (b21) to [out=85,in=275] (b22);
  
    \draw[t,->]  (e22) edge (b22);
    \draw[t,->]  (e12) edge (e22);
    \draw[t,->]  (b12) edge (e12);
    \draw[->]  (b12)  to [out=265,in=95]  (b11);
    
    \draw[<-]  (e11) edge (e12);
    \draw[->]  (e22) edge (e21);
    
    \draw[t,->]  (b11) edge (h11);
    \draw[t,->]  (h12) edge (b12);
    \draw[t,<-]  (b21) edge (h21);
    \draw[t,->]  (b22) edge (h22);
    
    \end{pgfonlayer}
}

\newcommand{\avoidone}{%
    \tikzstyle{node_stylek} = [circle,draw=black,fill=black, inner sep=0pt, minimum size=4pt, line width=.5]
    \tikzstyle{node_stylew} = [circle,draw=black,fill=white, inner sep=0pt, minimum size=4pt, line width=.5]

    \tikzstyle{nhidden_style} = [inner sep=0pt, minimum size=0pt]
    \tikzstyle{s} = [.,draw=black, line width=.5]
    \tikzstyle{t} = [draw=black, line width=2]
 
   \node[nhidden_style](h02) at (-3,1) {};

    \node[node_stylek](h11) at (-2,0) {};
    \node[node_stylew](h12) at (-2,1) {};
    \node[node_stylew](h21) at (-1,0) {};
    \node[node_stylek](h22) at (-1,1) {};
    \node[node_stylek](h31) at (0,0) {};
    \node[node_stylew](h32) at (0,1) {};
    \node[node_stylew](h41) at (1,0) {};
    \node[node_stylek](h42) at (1,1) {};
    \node[nhidden_style](h4e1) at (1.7,0) {};
    \node[nhidden_style](h4e2) at (1.7,1) {};
    \node[nhidden_style](h5b1) at (2.3,0) {};
    \node[nhidden_style](h5b2) at (2.3,1) {};
    \node[nhidden_style](hm) at (2.0,.5) {$\dots$};
     
    \node[node_stylek](h51) at (3,0) {};
    \node[node_stylew](h52) at (3,1) {};
    \node[node_stylew](h61) at (4,0) {};
    \node[node_stylek](h62) at (4,1) {};  
    \node[nhidden_style](h72) at (5,1) {};
        
    \begin{pgfonlayer}{bg} 
    
    \draw[t]  (h12) edge (h02.center);
    
    \draw[t]  (h21) edge (h31.center);
    \draw[t]  (h22) edge (h32.center);
    \draw[t]  (h41) edge (h4e1.center);
    \draw[t]  (h42) edge (h4e2.center);
    \draw[t]  (h51) edge (h5b1.center);
    \draw[t]  (h52) edge (h5b2.center);
    
    \draw[t]  (h62) edge (h72.center);
    
    \draw[t] (h11) to [out=270, in=270] (h61.center);

    \draw[s]  (h11) edge (h21.center);
    \draw[s]  (h31) edge (h41.center);
    \draw[s]  (h51) edge (h61.center);
    \draw[s]  (h12) edge (h22.center);
    \draw[s]  (h32) edge (h42.center);
    \draw[s]  (h52) edge (h62.center);
    
    \draw[s]  (h11) edge (h12.center);
    \draw[s]  (h21) edge (h22.center);
    \draw[s]  (h31) edge (h32.center);
    \draw[s]  (h41) edge (h42.center);
    \draw[s]  (h51) edge (h52.center);
    \draw[s]  (h61) edge (h62.center);
       
    \end{pgfonlayer}
    \node[rectangle](R) at (-1.5,1.5) {$c_1$};
    \node[rectangle](R) at (0.5,1.5) {$c_2$};  
    \node[rectangle](R) at (3.5,1.5) {$c_k$}; 
    
    \node[rectangle](R) at (-3.5,1.0) {$a$}; 
    \node[rectangle](R) at (5.5,1.0) {$b$}; 
     
}

\newcommand{\avoidoneex}{%
    \tikzstyle{node_stylek} = [circle,draw=black,fill=black, inner sep=0pt, minimum size=4pt, line width=.5]
    \tikzstyle{node_stylew} = [circle,draw=black,fill=white, inner sep=0pt, minimum size=4pt, line width=.5]

    \tikzstyle{nhidden_style} = [inner sep=0pt, minimum size=0pt]
    \tikzstyle{s} = [.,draw=black, line width=.5]
    \tikzstyle{t} = [draw=black, line width=2]
    \tikzstyle{g} = [draw=gray, line width=2]
 
   \node[nhidden_style](h02) at (-3,1) {};

    \node[node_stylek](h11) at (-2,0) {};
    \node[node_stylew](h12) at (-2,1) {};
    \node[node_stylew](h21) at (-1,0) {};
    \node[node_stylek](h22) at (-1,1) {};
    \node[node_stylek](h31) at (0,0) {};
    \node[node_stylew](h32) at (0,1) {};
    \node[node_stylew](h41) at (1,0) {};
    \node[node_stylek](h42) at (1,1) {};
    \node[nhidden_style](h4e1) at (1.7,0) {};
    \node[nhidden_style](h4e2) at (1.7,1) {};
    \node[nhidden_style](h5b1) at (2.3,0) {};
    \node[nhidden_style](h5b2) at (2.3,1) {};
    \node[nhidden_style](hm) at (2.0,.5) {$\dots$};
     
    \node[node_stylek](h51) at (3,0) {};
    \node[node_stylew](h52) at (3,1) {};
    \node[node_stylew](h61) at (4,0) {};
    \node[node_stylek](h62) at (4,1) {};  
    \node[nhidden_style](h72) at (5,1) {};
        
    \begin{pgfonlayer}{bg} 
    
    \draw[t]  (h12) edge (h02.center);
    
    \draw[t]  (h21) edge (h31.center);
    \draw[t]  (h22) edge (h32.center);
    \draw[t]  (h41) edge (h4e1.center);
    \draw[t]  (h42) edge (h4e2.center);
    \draw[t]  (h51) edge (h5b1.center);
    \draw[t]  (h52) edge (h5b2.center);
    
    \draw[t]  (h62) edge (h72.center);
    
    \draw[t] (h11) to [out=270, in=270] (h61.center);

    \draw[g]  (h11) edge (h21.center);
    \draw[s]  (h31) edge (h41.center);
    \draw[g]  (h51) edge (h61.center);
    \draw[g]  (h12) edge (h22.center);
    \draw[s]  (h32) edge (h42.center);
    \draw[g]  (h52) edge (h62.center);
    
    \draw[s]  (h11) edge (h12.center);
    \draw[s]  (h21) edge (h22.center);
    \draw[g]  (h31) edge (h32.center);
    \draw[g]  (h41) edge (h42.center);
    \draw[s]  (h51) edge (h52.center);
    \draw[s]  (h61) edge (h62.center);
       
    \end{pgfonlayer}
    \node[rectangle](R) at (-1.5,1.5) {$c_1$};
    \node[rectangle](R) at (0.5,1.5) {$c_2$};  
    \node[rectangle](R) at (3.5,1.5) {$c_k$}; 
    
    \node[rectangle](R) at (-3.5,1.0) {$a$}; 
    \node[rectangle](R) at (5.5,1.0) {$b$}; 
     
}

\newcommand{\avoidoneschematic}{%
    \tikzstyle{node_stylek} = [circle,draw=black,fill=black, inner sep=0pt, minimum size=4pt, line width=.5]
    \tikzstyle{node_stylew} = [circle,draw=black,fill=white, inner sep=0pt, minimum size=4pt, line width=.5]

    \tikzstyle{nhidden_style} = [inner sep=0pt, minimum size=0pt]
    \tikzstyle{s} = [dashed,draw=black, line width=.5]
    \tikzstyle{t} = [draw=black, line width=2]
 
    \node[circle, draw=black, inner sep=0pt, minimum size=10pt, line width=1.5](C) at (0,0) {$\dot{\vee}$};   
    \node[nhidden_style](l) at (-2,0) {};
    \node[nhidden_style](r) at (2,0) {};
    
    \node[nhidden_style](h1) at (-1.5,2) {$c_1$};
    \node[nhidden_style](h2) at (-.5,2) {$c_2$};
    \node[nhidden_style](hm) at (.5,2) {$\dots$};

    \node[nhidden_style](hk) at (1.5,2) {$c_k$};

    \node[nhidden_style](hq) at (1.5,-2) {};

    \begin{pgfonlayer}{bg} 
    
    \draw[s]  (C) edge (h1);
    \draw[s]  (C) edge (h2);
    \draw[s]  (C) edge (hk);
    \draw[t]  (C) edge (l);
    \draw[t]  (C) edge (r);

    \end{pgfonlayer}
}

\newcommand{\xor}{%
    \tikzstyle{node_stylek} = [circle,draw=black,fill=black, inner sep=0pt, minimum size=4pt, line width=.5]
    \tikzstyle{node_stylew} = [circle,draw=black,fill=white, inner sep=0pt, minimum size=4pt, line width=.5]

    \tikzstyle{nhidden_style} = [inner sep=0pt, minimum size=0pt]
    \tikzstyle{s} = [.,draw=black, line width=.5]
    \tikzstyle{t} = [draw=black, line width=2]
 
   \node[nhidden_style](h02) at (-3,1) {};

    \node[nhidden_style](h11) at (-3,1) {};
    \node[node_stylew](h12) at (-2,1) {};
    \node[node_stylek](h13) at (-1,1) {};
    \node[node_stylew](h14) at (0,1) {};
    \node[node_stylek](h15) at (1,1) {};
    \node[nhidden_style](h16) at (2,1) {};
    
    \node[nhidden_style](h21) at (-3,-1) {};
    \node[node_stylek](h22) at (-2,-1) {};
    \node[node_stylew](h23) at (-1,-1) {};
    \node[node_stylek](h24) at (0,-1) {};
    \node[node_stylew](h25) at (1,-1) {};
    \node[nhidden_style](h26) at (2,-1) {};
            
    \begin{pgfonlayer}{bg} 
    
    \draw[s]  (h11) edge (h12.center);
    \draw[s]  (h12) edge (h13.center);
    \draw[s]  (h13) edge (h14.center);
    \draw[s]  (h14) edge (h15.center);
    \draw[s]  (h15) edge (h16.center); 
    \draw[t]  (h12) edge (h22.center);
    \draw[t]  (h13) edge (h23.center);
    \draw[t]  (h14) edge (h24.center);
    \draw[t]  (h15) edge (h25.center);
    \draw[s]  (h21) edge (h22.center);
    \draw[s]  (h22) edge (h23.center);
    \draw[s]  (h23) edge (h24.center);
    \draw[s]  (h24) edge (h25.center);
    \draw[s]  (h25) edge (h26.center);

    \end{pgfonlayer}
    \node[rectangle](R) at (-4,-1) {$b_1$};
    \node[rectangle](R) at (3,-1) {$b_2$};  
    \node[rectangle](R) at (-4,1) {$a_1$};  
    \node[rectangle](R) at (3,1) {$a_2$};  
}

\newcommand{\xorschematic}{%
    \tikzstyle{node_stylek} = [circle,draw=black,fill=black, inner sep=0pt, minimum size=4pt, line width=.5]
    \tikzstyle{node_stylew} = [circle,draw=black,fill=white, inner sep=0pt, minimum size=4pt, line width=.5]
    \tikzstyle{node_stylex} = [rectangle,draw=black,fill=black, inner sep=0pt, minimum size=4pt, line width=.5]

    \tikzstyle{nhidden_style} = [inner sep=0pt, minimum size=0pt]
    \tikzstyle{s} = [draw=black, line width=.5]
    \tikzstyle{t} = [draw=black, line width=2]
   
    \node[nhidden_style](h11) at (-3,1) {};
    \node[nhidden_style](h12) at (2,1) {};
    
    \node[nhidden_style](h21) at (-3,-1) {};
    \node[nhidden_style](h22) at (2,-1) {};
    
    \node[node_stylex](h1m) at (-.5,0.8) {};
    \node[node_stylex](h2m) at (-.5,-0.8) {};

    \begin{pgfonlayer}{bg} 
    
    \draw[s]  (h11) edge (h12.center);
    \draw[s]  (h21) edge (h22.center);
    \draw[s]  (h1m) edge (h2m.center);
    
    \node[rectangle](R) at (-4,-1) {$b_1$};
    \node[rectangle](R) at (3,-1) {$b_2$};  
    \node[rectangle](R) at (-4,1) {$a_1$};  
    \node[rectangle](R) at (3,1) {$a_2$};  
    \end{pgfonlayer}
}

\newcommand{\xorcrossing}{%
    \tikzstyle{node_stylek} = [circle,draw=black,fill=black, inner sep=0pt, minimum size=4pt, line width=.5]
    \tikzstyle{node_stylew} = [circle,draw=black,fill=white, inner sep=0pt, minimum size=4pt, line width=.5]
    \tikzstyle{node_stylex} = [rectangle,draw=black,fill=black, inner sep=0pt, minimum size=4pt, line width=.5]

    \tikzstyle{nhidden_style} = [inner sep=0pt, minimum size=0pt]
    \tikzstyle{s} = [draw=black, line width=.5]
    \tikzstyle{t} = [draw=black, line width=2]
   
    \node[nhidden_style](h11) at (-1,2) {};
    \node[nhidden_style](h12) at (1,2) {};
    
    \node[nhidden_style](h21) at (-1,-2) {};
    \node[nhidden_style](h22) at (1,-2) {};
    
    \node[node_stylex](h1m) at (0,1.8) {};
    \node[node_stylex](h2m) at (0,-1.8) {};

    \node[nhidden_style](v11) at (2,-1) {};
    \node[nhidden_style](v12) at (2,1) {};
    
    \node[nhidden_style](v21) at (-2,-1) {};
    \node[nhidden_style](v22) at (-2,1) {};
    
    \node[node_stylex](v1m) at (1.8,0) {};
    \node[node_stylex](v2m) at (-1.8,0) {};
    
    \begin{pgfonlayer}{bg} 
    
    \draw[s]  (h11) edge (h12.center);
    \draw[s]  (h21) edge (h22.center);
    \draw[s]  (h1m) edge (h2m.center);
    
    \draw[s]  (v11) edge (v12.center);
    \draw[s]  (v21) edge (v22.center);
    \draw[s]  (v1m) edge (v2m.center);
   
    \end{pgfonlayer}
}

\newcommand{\xoruncrossing}{%
    \tikzstyle{node_stylek} = [circle,draw=black,fill=black, inner sep=0pt, minimum size=4pt, line width=.5]
    \tikzstyle{node_stylew} = [circle,draw=black,fill=white, inner sep=0pt, minimum size=4pt, line width=.5]
    \tikzstyle{node_stylex} = [rectangle,draw=black,fill=black, inner sep=0pt, minimum size=4pt, line width=.5]

    \tikzstyle{nhidden_style} = [inner sep=0pt, minimum size=0pt]
    \tikzstyle{s} = [draw=black, line width=.5]
    \tikzstyle{t} = [draw=black, line width=2]
   
    \node[nhidden_style](h11) at (-6,2) {};
    \node[nhidden_style](h12) at (6,2) {};
    
    \node[nhidden_style](h21) at (-6,-2) {};
    \node[nhidden_style](h22) at (6,-2) {};
    
    \node[node_stylew](h1m1) at (-4.5,2) {};
    \node[node_stylek](h1c1) at (-4.5,1) {};
    \node[node_stylew](h2c1) at (-4.5,-1) {};
    \node[node_stylek](h2m1) at (-4.5,-2) {};

    \node[node_stylek](h1m2) at (-1.5,2) {};
    \node[node_stylew](h1c2) at (-1.5,1) {};
    \node[node_stylek](h2c2) at (-1.5,-1) {};
    \node[node_stylew](h2m2) at (-1.5,-2) {};

    \node[node_stylew](h1m3) at (1.5,2) {};
    \node[node_stylek](h1c3) at (1.5,1) {};
    \node[node_stylew](h2c3) at (1.5,-1) {};
    \node[node_stylek](h2m3) at (1.5,-2) {};

    \node[node_stylek](h1m4) at (4.5,2) {};
    \node[node_stylew](h1c4) at (4.5,1) {};
    \node[node_stylek](h2c4) at (4.5,-1) {};
    \node[node_stylew](h2m4) at (4.5,-2) {};

    \node[nhidden_style](v11) at (7,-1) {};
    \node[nhidden_style](v12) at (7,1) {};
    
    \node[nhidden_style](v21) at (-7,-1) {};
    \node[nhidden_style](v22) at (-7,1) {};
    
    \node[node_stylex](v1) at (6.8,0) {};
    \node[node_stylex](v2) at (5.05,0) {};
    \node[node_stylex](v3) at (3.95,0) {};
    \node[node_stylex](v4) at (2.05,0) {};
    \node[node_stylex](v5) at (0.95,0) {};
    \node[node_stylex](v6) at (-0.95,0) {};
    \node[node_stylex](v7) at (-2.05,0) {};
    \node[node_stylex](v8) at (-3.95,0) {};
    \node[node_stylex](v9) at (-5.05,0) {};
    \node[node_stylex](v10) at (-6.8,0) {};
    \begin{pgfonlayer}{bg} 
    
    \draw[s]  (h11) edge (h12.center);
    \draw[s]  (h21) edge (h22.center);
    \draw[s]  (h1m1) edge (h1c1.center);
    \draw[s]  (h2m1) edge (h2c1.center);
    \draw[s]  (h1c1) to [out=240, in=120] (h2c1.center);
    \draw[s]  (h2c1) to [out=60, in=300] (h1c1.center);

    \draw[s]  (h1m2) edge (h1c2.center);
    \draw[s]  (h2m2) edge (h2c2.center);
    \draw[s]  (h1c2) to [out=240, in=120] (h2c2.center);
    \draw[s]  (h2c2) to [out=60, in=300] (h1c2.center);
    
    \draw[s]  (h1m3) edge (h1c3.center);
    \draw[s]  (h2m3) edge (h2c3.center);
    \draw[s]  (h1c3) to [out=240, in=120] (h2c3.center);
    \draw[s]  (h2c3) to [out=60, in=300] (h1c3.center);

    \draw[s]  (h1m4) edge (h1c4.center);
    \draw[s]  (h2m4) edge (h2c4.center);
    \draw[s]  (h1c4) to [out=240, in=120] (h2c4.center);
    \draw[s]  (h2c4) to [out=60, in=300] (h1c4.center);

    \draw[s]  (v11) edge (v12.center);
    \draw[s]  (v21) edge (v22.center);
    \draw[s]  (v1) edge (v2.center);
    \draw[s]  (v3) edge (v4.center);
    \draw[s]  (v5) edge (v6.center);
    \draw[s]  (v7) edge (v8.center);
    \draw[s]  (v9) edge (v10.center);
    
    \end{pgfonlayer}
}

\newcommand{\instance}{%
    \tikzstyle{node_stylek} = [circle,draw=black,fill=black, inner sep=0pt, minimum size=4pt, line width=.5]
    \tikzstyle{node_stylew} = [circle,draw=black,fill=white, inner sep=0pt, minimum size=4pt, line width=.5]
    \tikzstyle{node_stylex} = [rectangle,draw=black,fill=black, inner sep=0pt, minimum size=4pt, line width=.5]

    \tikzstyle{nhidden_style} = [inner sep=0pt, minimum size=0pt]
    \tikzstyle{s} = [draw=black, line width=.5]
    \tikzstyle{t} = [draw=black, line width=2]

    \node[circle, draw=black, fill=white, inner sep=0pt, minimum size=10pt, line width=1.5](u1) at (-12.5,0) {$\dot{\vee}$};   
    \node[circle, draw=black, fill=white, inner sep=0pt, minimum size=10pt, line width=1.5](u2) at (-7.5,0) {$\dot{\vee}$};   
    \node[circle, draw=black, fill=white, inner sep=0pt, minimum size=10pt, line width=1.5](u3) at (-2.5,0) {$\dot{\vee}$};   
    \node[circle, draw=black, fill=white, inner sep=0pt, minimum size=10pt, line width=1.5](u4) at (2.5,0) {$\dot{\vee}$};   
    \node[circle, draw=black, fill=white, inner sep=0pt, minimum size=10pt, line width=1.5](u5) at (7.5,0) {$\dot{\vee}$};   
    \node[circle, draw=black, fill=white, inner sep=0pt, minimum size=10pt, line width=1.5](u6) at (12.5,0) {$\dot{\vee}$};   
 
    \node[node_stylek](s11) at (-14,4) {};
    \node[node_stylew](s12) at (-11,4) {};

    \node[node_stylek](s21) at (-9,4) {};
    \node[node_stylew](s22) at (-6,4) {};

    \node[node_stylek](s31) at (-4,4) {};
    \node[node_stylew](s32) at (-1,4) {};
    
    \node[node_stylek](s41) at (1,4) {};
    \node[node_stylew](s42) at (4,4) {};

    \node[node_stylek](s51) at (6,4) {};
    \node[node_stylew](s52) at (9,4) {};

    \node[node_stylek](s61) at (11,4) {};
    \node[node_stylew](s62) at (14,4) {};

    \node[node_stylex](c1) at (-3,3.5) {};
    \node[node_stylex](c2) at (-2.5,3.5) {};
    \node[node_stylex](c3) at (-2.0,3.5) {};
    \node[node_stylex](d1) at (-12.5,0.5) {};
    \node[node_stylex](d2) at (2.25,0.5) {};
    \node[node_stylex](d3) at (7.5,0.5) {};

    \node[node_stylex](e1) at (2,3.5) {};
    \node[node_stylex](e2) at (2.5,3.5) {};
    \node[node_stylex](e3) at (3,3.5) {};
    \node[node_stylex](f1) at (-7.5,0.5) {};
    \node[node_stylex](f2) at (-2.5,0.5) {};
    \node[node_stylex](f3) at (2.75,0.5) {};
  
    \begin{pgfonlayer}{bg} 
    
    \draw[t]  (u1) edge (u2.center);
    \draw[t]  (u2) edge (u3.center);
    \draw[t]  (u3) edge (u4.center);
    \draw[t]  (u4) edge (u5.center);
    \draw[t]  (u5) edge (u6.center);
    
    \draw[t] (u6) to [out=30, in=330] (s62.center);
    \draw[t] (s11) to [out=210, in=150] (u1.center);
    
    \draw[s] (s11) to [out=30, in=150] (s12.center);
    \draw[s] (s11) to [out=-30, in=-150] (s12.center);
    \draw[t] (s12) edge (s21.center);   
    \draw[s] (s21) to [out=30, in=150] (s22.center);
    \draw[s] (s21) to [out=-30, in=-150] (s22.center);
    \draw[t] (s22) edge (s31.center);   
    \draw[s] (s31) to [out=30, in=150] (s32.center);
    \draw[s] (s31) to [out=-30, in=-150] (s32.center);
    \draw[t] (s32) edge (s41.center);   
    \draw[s] (s41) to [out=30, in=150] (s42.center);
    \draw[s] (s41) to [out=-30, in=-150] (s42.center);
    \draw[t] (s42) edge (s51.center);   
    \draw[s] (s51) to [out=30, in=150] (s52.center);
    \draw[s] (s51) to [out=-30, in=-150] (s52.center);
    \draw[t] (s52) edge (s61.center);   
    \draw[s] (s61) to [out=30, in=150] (s62.center);
    \draw[s] (s61) to [out=-30, in=-150] (s62.center);
   
    \draw[s](c1) edge (d1.center);
    \draw[s](c2) edge (d2.center);
    \draw[s](c3) edge (d3.center);

    \draw[s](e1) edge (f1.center);
    \draw[s](e2) edge (f2.center);
    \draw[s](e3) edge (f3.center);

    \node[rectangle](R) at (-12.5,-1) {$u_1$};
    \node[rectangle](R) at (-7.5,-1) {$u_2$};  
    \node[rectangle](R) at (-2.5,-1) {$u_3$};  
    \node[rectangle](R) at (2.5,-1) {$u_4$};  
    \node[rectangle](R) at (7.5,-1) {$u_5$};  
    \node[rectangle](R) at (12.5,-1) {$u_6$};  
   
    \node[rectangle](R) at (-12.5,5) {$S_1$};
    \node[rectangle](R) at (-7.5,5) {$S_2$};  
    \node[rectangle](R) at (-2.5,5) {$S_3$};  
    \node[rectangle](R) at (2.5,5) {$S_4$};  
    \node[rectangle](R) at (7.5,5) {$S_5$};  
    \node[rectangle](R) at (12.5,5) {$S_6$};   
    \end{pgfonlayer}
}

\newcommand{\Hlp}{%
    \definecolor{scarlet}{rgb}{1.0,0.14,0}
    \definecolor{orange}{rgb}{1.0,0.64,0}
    \definecolor{green}{rgb}{0.0,0.5,0.0}
    \tikzstyle{node_styleb} = [circle,draw=black,fill=blue, inner sep=0pt, minimum size=4pt, line width=.5]
    \tikzstyle{node_styleo} = [rectangle,draw=black,fill=orange, inner sep=0pt, minimum size=4pt, line width=.5]
 
    \tikzstyle{nhidden_style} = [inner sep=0pt, minimum size=0pt]
    \tikzstyle{t} = [draw=black, line width=1.5]
    \tikzstyle{s} = [draw=green, line width=1.5]
 
    \node[nhidden_style](b) at (0,-3) {};

    \node[node_styleb](h1) at (-3.5,2.3) {};
    \node[node_styleo](h2) at (-1.2,3) {};
    \node[node_styleb](h3) at (1.2,3) {};
    \node[node_styleo](h4) at (3.5,2.3) {};
  
    \node[node_styleo](h5) at (-2,0.9) {};
    \node[node_styleb](h6) at (-0.7,1.4) {};
    \node[node_styleo](h7) at (0.7,1.4) {};
    \node[node_styleb](h8) at (2,0.9) {};
  
    \node[node_styleb](h9) at (-2,-0.9) {};
    \node[node_styleo](h10) at (-0.7,-1.4) {};
    \node[node_styleb](h11) at (0.7,-1.4) {};
    \node[node_styleo](h12) at (2,-0.9) {};
  
    \node[node_styleo](h13) at (-3.5,-2.3) {};
    \node[node_styleb](h14) at (-1.2,-3) {};
    \node[node_styleo](h15) at (1.2,-3) {};
    \node[node_styleb](h16) at (3.5,-2.3) {};
    
    \begin{pgfonlayer}{bg} 

    \draw[t,->]  (h13) to [out=105,in=255]  (h1);
    \draw[t,->]  (h1) to [out=30,in=190]  (h2);
    \draw[t,->]  (h2) to [out=300,in=100]  (h6);
    \draw[t,->]  (h6) to [out=190,in=20]  (h5);
    \draw[t,->]  (h5) to [out=262,in=98]  (h9);
    \draw[t,->]  (h9) to [out=-20,in=170]  (h10);
    \draw[t,->]  (h10) to [out=0,in=180]  (h11);
    \draw[t,->]  (h11) to [out=10,in=210]  (h12); 
    \draw[t,->]  (h12) to [out=82,in=278]  (h8); 
    \draw[t,->]  (h8) to [out=150,in=-10]  (h7); 
    \draw[t,->]  (h7) to [out=80,in=240]  (h3); 
    \draw[t,->]  (h3) to [out=-10,in=150]  (h4); 
    \draw[t,->]  (h4) to [out=285,in=75]  (h16); 
    \draw[t,->]  (h16) to [out=200,in=10]  (h15); 
    \draw[t,->]  (h15) to [out=190,in=-10]  (h14); 
    \draw[t,->]  (h14) to [out=170,in=-30]  (h13);
    
    \draw[->]  (h1) to [out=-40,in=130]  (h5);
    \draw[->]  (h2) to [out=10,in=170]  (h3);  
    \draw[->]  (h6) to [out=10,in=170]  (h7);
    \draw[->]  (h8) to [out=40,in=230]  (h4);
   
    \draw[->]  (h9) to [out=230,in=40]  (h13);
    \draw[->]  (h10) to [out=260,in=60]  (h14);  
    \draw[->]  (h11) to [out=280,in=120]  (h15);
    \draw[->]  (h12) to [out=-50,in=150]  (h16);     
 
    \end{pgfonlayer}
}

\newcommand{\PlpO}{%
    \definecolor{scarlet}{rgb}{1.0,0.14,0}
    \definecolor{orange}{rgb}{1.0,0.64,0}
    \definecolor{green}{rgb}{0.0,0.5,0.0}
    \definecolor{lgreen}{rgb}{0.2,0.9,0.2}
    \tikzstyle{node_styleb} = [circle,draw=black,fill=blue, inner sep=0pt, minimum size=4pt, line width=.5]
    \tikzstyle{node_styleo} = [rectangle,draw=black,fill=orange, inner sep=0pt, minimum size=4pt, line width=.5]
    \tikzstyle{node_styles} = [circle,draw=black,fill=scarlet, inner sep=0pt, minimum size=4pt, line width=.5]
    
    \tikzstyle{nhidden_style} = [inner sep=0pt, minimum size=0pt]
    \tikzstyle{t} = [draw=black, line width=1.5]
    \tikzstyle{s} = [draw=green, line width=1.5]
    \tikzstyle{sx} = [draw=lgreen, line width=1.5]
      
    \node[nhidden_style](b) at (0,-3) {};

    \node[node_styleb](h1) at (-3.5,2.3) {};
    \node[node_styleo](h2) at (-1.2,3) {};
    \node[node_styleb](h3) at (1.2,3) {};
    \node[node_styleo](h4) at (3.5,2.3) {};
  
    \node[node_styleo](h5) at (-2,0.9) {};
    \node[node_styleb](h6) at (-0.7,1.4) {};
    \node[node_styleo](h7) at (0.7,1.4) {};
    \node[node_styleb](h8) at (2,0.9) {};
  
    \node[node_styleb](h9) at (-2,-0.9) {};
    \node[node_styleo](h10) at (-0.7,-1.4) {};
    \node[node_styleb](h11) at (0.7,-1.4) {};
    \node[node_styleo](h12) at (2,-0.9) {};
  
    \node[node_styleo](h13) at (-3.5,-2.3) {};
    \node[node_styles](h14) at (-1.2,-3) {};
    \node[node_styleo](h15) at (1.2,-3) {};
    \node[node_styleb](h16) at (3.5,-2.3) {};
    
    \begin{pgfonlayer}{bg} 

    \draw[t,->]  (h13) to [out=105,in=255]  (h1);
    \draw[t,->]  (h1) to [out=30,in=190]  (h2);
    \draw[t,->]  (h2) to [out=300,in=100]  (h6);
    \draw[t,->]  (h6) to [out=190,in=20]  (h5);
    \draw[t,->]  (h5) to [out=262,in=98]  (h9);
    \draw[t,->]  (h9) to [out=-20,in=170]  (h10);
    \draw[t,->]  (h10) to [out=0,in=180]  (h11);
    \draw[t,->]  (h11) to [out=10,in=210]  (h12); 
    \draw[t,->]  (h12) to [out=82,in=278]  (h8); 
    \draw[t,->]  (h8) to [out=150,in=-10]  (h7); 
    \draw[t,->]  (h7) to [out=80,in=240]  (h3); 
    \draw[t,->]  (h3) to [out=-10,in=150]  (h4); 
    \draw[t,->]  (h4) to [out=285,in=75]  (h16); 
    \draw[t,->]  (h16) to [out=200,in=10]  (h15); 
    \draw[t,->]  (h15) to [out=190,in=-10]  (h14); 
    \draw[->]  (h14) to [out=170,in=-30]  (h13);
    
    \draw[->]  (h1) to [out=-40,in=130]  (h5);
    \draw[->]  (h2) to [out=10,in=170]  (h3);  
    \draw[->]  (h6) to [out=10,in=170]  (h7);
    \draw[->]  (h8) to [out=40,in=230]  (h4);
   
    \draw[->]  (h9) to [out=230,in=40]  (h13);
    \draw[->]  (h10) to [out=260,in=60]  (h14);  
    \draw[->]  (h11) to [out=280,in=120]  (h15);
    \draw[->]  (h12) to [out=-50,in=150]  (h16);     
    
    \draw[s,->] (h1) to [out=-30, in=180] (h6);
    \draw[s,->] (h6) to [out=-30, in=180] (h8);
    \draw[s,->] (h8) to [out=90, in=-60] (h3);
    \draw[s,->] (h3) to [out=190, in=60] (h6);
    
    \draw[s,->] (h9) to [out=260, in=150] (h14);
    \draw[sx,->] (h14) to [out=60, in=200] (h11);
    \draw[s,->] (h11) to [out=-60, in=170] (h16);
    \draw[s,->] (h16) to [out=120, in=280] (h8);

    \end{pgfonlayer}
}

\newcommand{\PlpI}{%
    \definecolor{scarlet}{rgb}{1.0,0.14,0}
    \definecolor{orange}{rgb}{1.0,0.64,0}
    \definecolor{green}{rgb}{0.0,0.5,0.0}
    \definecolor{lgreen}{rgb}{0.2,0.9,0.2} 
    \tikzstyle{node_styleb} = [circle,draw=black,fill=blue, inner sep=0pt, minimum size=4pt, line width=.5]
    \tikzstyle{node_styleo} = [rectangle,draw=black,fill=orange, inner sep=0pt, minimum size=4pt, line width=.5]
    \tikzstyle{node_styles} = [circle,draw=black,fill=scarlet, inner sep=0pt, minimum size=4pt, line width=.5]
    
    \tikzstyle{nhidden_style} = [inner sep=0pt, minimum size=0pt]
    \tikzstyle{t} = [draw=black, line width=1.5]
    \tikzstyle{s} = [draw=green, line width=1.5]
    \tikzstyle{sx} = [draw=lgreen, line width=1.5]
    
    \node[nhidden_style](b) at (0,-3) {};

    \node[node_styleb](h1) at (-3.5,2.3) {};
    \node[node_styleo](h2) at (-1.2,3) {};
    \node[node_styleb](h3) at (1.2,3) {};
    \node[node_styleo](h4) at (3.5,2.3) {};
  
    \node[node_styleo](h5) at (-2,0.9) {};
    \node[node_styleb](h6) at (-0.7,1.4) {};
    \node[node_styleo](h7) at (0.7,1.4) {};
    \node[node_styleb](h8) at (2,0.9) {};
  
    \node[node_styleb](h9) at (-2,-0.9) {};
    \node[node_styleo](h10) at (-0.7,-1.4) {};
    \node[node_styles](h11) at (0.7,-1.4) {};
    \node[node_styleo](h12) at (2,-0.9) {};
  
    \node[node_styleo](h13) at (-3.5,-2.3) {};
    \node[node_styleb](h14) at (-1.2,-3) {};
    \node[node_styleo](h15) at (1.2,-3) {};
    \node[node_styleb](h16) at (3.5,-2.3) {};
    
    \begin{pgfonlayer}{bg} 

    \draw[t,->]  (h13) to [out=105,in=255]  (h1);
    \draw[t,->]  (h1) to [out=30,in=190]  (h2);
    \draw[t,->]  (h2) to [out=300,in=100]  (h6);
    \draw[t,->]  (h6) to [out=190,in=20]  (h5);
    \draw[t,->]  (h5) to [out=262,in=98]  (h9);
    \draw[t,->]  (h9) to [out=-20,in=170]  (h10);
    \draw[->]  (h10) to [out=0,in=180]  (h11);
    \draw[t,->]  (h11) to [out=10,in=210]  (h12); 
    \draw[t,->]  (h12) to [out=82,in=278]  (h8); 
    \draw[t,->]  (h8) to [out=150,in=-10]  (h7); 
    \draw[t,->]  (h7) to [out=80,in=240]  (h3); 
    \draw[t,->]  (h3) to [out=-10,in=150]  (h4); 
    \draw[t,->]  (h4) to [out=285,in=75]  (h16); 
    \draw[t,->]  (h16) to [out=200,in=10]  (h15); 
    \draw[t,->]  (h15) to [out=190,in=-10]  (h14); 
    \draw[->]  (h14) to [out=170,in=-30]  (h13);
    
    \draw[->]  (h1) to [out=-40,in=130]  (h5);
    \draw[->]  (h2) to [out=10,in=170]  (h3);  
    \draw[->]  (h6) to [out=10,in=170]  (h7);
    \draw[->]  (h8) to [out=40,in=230]  (h4);
   
    \draw[->]  (h9) to [out=230,in=40]  (h13);
    \draw[t,->]  (h10) to [out=260,in=60]  (h14);  
    \draw[->]  (h11) to [out=280,in=120]  (h15);
    \draw[->]  (h12) to [out=-50,in=150]  (h16);     
    
    \draw[s,->] (h1) to [out=-30, in=180] (h6);
    \draw[s,->] (h6) to [out=-30, in=180] (h8);
    \draw[s,->] (h8) to [out=90, in=-60] (h3);
    \draw[s,->] (h3) to [out=190, in=60] (h6);
    
    \draw[s,->] (h9) to [out=260, in=150] (h14);
    \draw[s,->] (h11) to [out=200, in=60] (h14);
    \draw[sx,->] (h11) to [out=-60, in=170] (h16);
    \draw[s,->] (h16) to [out=120, in=280] (h8);

    \end{pgfonlayer}
}

\newcommand{\PlpII}{%
    \definecolor{scarlet}{rgb}{1.0,0.14,0}
    \definecolor{orange}{rgb}{1.0,0.64,0}
    \definecolor{green}{rgb}{0.0,0.5,0.0}
    \definecolor{lgreen}{rgb}{0.2,0.9,0.2}
    \tikzstyle{node_styleb} = [circle,draw=black,fill=blue, inner sep=0pt, minimum size=4pt, line width=.5]
    \tikzstyle{node_styleo} = [rectangle,draw=black,fill=orange, inner sep=0pt, minimum size=4pt, line width=.5]
    \tikzstyle{node_styles} = [circle,draw=black,fill=scarlet, inner sep=0pt, minimum size=4pt, line width=.5]
    
    \tikzstyle{nhidden_style} = [inner sep=0pt, minimum size=0pt]
    \tikzstyle{t} = [draw=black, line width=1.5]
    \tikzstyle{s} = [draw=green, line width=1.5]
    \tikzstyle{sx} = [draw=lgreen, line width=1.5]
    
    \node[nhidden_style](b) at (0,-3) {};

    \node[node_styleb](h1) at (-3.5,2.3) {};
    \node[node_styleo](h2) at (-1.2,3) {};
    \node[node_styleb](h3) at (1.2,3) {};
    \node[node_styleo](h4) at (3.5,2.3) {};
  
    \node[node_styleo](h5) at (-2,0.9) {};
    \node[node_styleb](h6) at (-0.7,1.4) {};
    \node[node_styleo](h7) at (0.7,1.4) {};
    \node[node_styleb](h8) at (2,0.9) {};
  
    \node[node_styleb](h9) at (-2,-0.9) {};
    \node[node_styleo](h10) at (-0.7,-1.4) {};
    \node[node_styleb](h11) at (0.7,-1.4) {};
    \node[node_styleo](h12) at (2,-0.9) {};
  
    \node[node_styleo](h13) at (-3.5,-2.3) {};
    \node[node_styleb](h14) at (-1.2,-3) {};
    \node[node_styleo](h15) at (1.2,-3) {};
    \node[node_styles](h16) at (3.5,-2.3) {};
    
    \begin{pgfonlayer}{bg} 

    \draw[t,->]  (h13) to [out=105,in=255]  (h1);
    \draw[t,->]  (h1) to [out=30,in=190]  (h2);
    \draw[t,->]  (h2) to [out=300,in=100]  (h6);
    \draw[t,->]  (h6) to [out=190,in=20]  (h5);
    \draw[t,->]  (h5) to [out=262,in=98]  (h9);
    \draw[t,->]  (h9) to [out=-20,in=170]  (h10);
    \draw[->]  (h10) to [out=0,in=180]  (h11);
    \draw[t,->]  (h11) to [out=10,in=210]  (h12); 
    \draw[t,->]  (h12) to [out=82,in=278]  (h8); 
    \draw[t,->]  (h8) to [out=150,in=-10]  (h7); 
    \draw[t,->]  (h7) to [out=80,in=240]  (h3); 
    \draw[t,->]  (h3) to [out=-10,in=150]  (h4); 
    \draw[t,->]  (h4) to [out=285,in=75]  (h16); 
    \draw[->]  (h16) to [out=200,in=10]  (h15); 
    \draw[t,->]  (h15) to [out=190,in=-10]  (h14); 
    \draw[->]  (h14) to [out=170,in=-30]  (h13);
    
    \draw[->]  (h1) to [out=-40,in=130]  (h5);
    \draw[->]  (h2) to [out=10,in=170]  (h3);  
    \draw[->]  (h6) to [out=10,in=170]  (h7);
    \draw[->]  (h8) to [out=40,in=230]  (h4);
   
    \draw[->]  (h9) to [out=230,in=40]  (h13);
    \draw[t,->]  (h10) to [out=260,in=60]  (h14);  
    \draw[t,->]  (h11) to [out=280,in=120]  (h15);
    \draw[->]  (h12) to [out=-50,in=150]  (h16);     
    
    \draw[s,->] (h1) to [out=-30, in=180] (h6);
    \draw[s,->] (h6) to [out=-30, in=180] (h8);
    \draw[s,->] (h8) to [out=90, in=-60] (h3);
    \draw[s,->] (h3) to [out=190, in=60] (h6);
    
    \draw[s,->] (h9) to [out=260, in=150] (h14);
    \draw[s,->] (h11) to [out=200, in=60] (h14);
    \draw[s,->] (h16) to [out=170, in=-60] (h11);
    \draw[sx,->] (h16) to [out=120, in=280] (h8);

    \end{pgfonlayer}
}

\newcommand{\PlpIII}{%
    \definecolor{scarlet}{rgb}{1.0,0.14,0}
    \definecolor{orange}{rgb}{1.0,0.64,0}
    \definecolor{green}{rgb}{0.0,0.5,0.0}
    \definecolor{lgreen}{rgb}{0.2,0.9,0.2}
    \tikzstyle{node_styleb} = [circle,draw=black,fill=blue, inner sep=0pt, minimum size=4pt, line width=.5]
    \tikzstyle{node_styleo} = [rectangle,draw=black,fill=orange, inner sep=0pt, minimum size=4pt, line width=.5]
    \tikzstyle{node_styles} = [circle,draw=black,fill=scarlet, inner sep=0pt, minimum size=4pt, line width=.5]
    
    \tikzstyle{nhidden_style} = [inner sep=0pt, minimum size=0pt]
    \tikzstyle{t} = [draw=black, line width=1.5]
    \tikzstyle{s} = [draw=green, line width=1.5]
    \tikzstyle{sx} = [draw=lgreen, line width=1.5]
      
    \node[nhidden_style](b) at (0,-3) {};

    \node[node_styleb](h1) at (-3.5,2.3) {};
    \node[node_styleo](h2) at (-1.2,3) {};
    \node[node_styleb](h3) at (1.2,3) {};
    \node[node_styleo](h4) at (3.5,2.3) {};
  
    \node[node_styleo](h5) at (-2,0.9) {};
    \node[node_styleb](h6) at (-0.7,1.4) {};
    \node[node_styleo](h7) at (0.7,1.4) {};
    \node[node_styles](h8) at (2,0.9) {};
  
    \node[node_styleb](h9) at (-2,-0.9) {};
    \node[node_styleo](h10) at (-0.7,-1.4) {};
    \node[node_styleb](h11) at (0.7,-1.4) {};
    \node[node_styleo](h12) at (2,-0.9) {};
  
    \node[node_styleo](h13) at (-3.5,-2.3) {};
    \node[node_styleb](h14) at (-1.2,-3) {};
    \node[node_styleo](h15) at (1.2,-3) {};
    \node[node_styleb](h16) at (3.5,-2.3) {};
    
    \begin{pgfonlayer}{bg} 

    \draw[t,->]  (h13) to [out=105,in=255]  (h1);
    \draw[t,->]  (h1) to [out=30,in=190]  (h2);
    \draw[t,->]  (h2) to [out=300,in=100]  (h6);
    \draw[t,->]  (h6) to [out=190,in=20]  (h5);
    \draw[t,->]  (h5) to [out=262,in=98]  (h9);
    \draw[t,->]  (h9) to [out=-20,in=170]  (h10);
    \draw[->]  (h10) to [out=0,in=180]  (h11);
    \draw[t,->]  (h11) to [out=10,in=210]  (h12); 
    \draw[->]  (h12) to [out=82,in=278]  (h8); 
    \draw[t,->]  (h8) to [out=150,in=-10]  (h7); 
    \draw[t,->]  (h7) to [out=80,in=240]  (h3); 
    \draw[t,->]  (h3) to [out=-10,in=150]  (h4); 
    \draw[t,->]  (h4) to [out=285,in=75]  (h16); 
    \draw[->]  (h16) to [out=200,in=10]  (h15); 
    \draw[t,->]  (h15) to [out=190,in=-10]  (h14); 
    \draw[->]  (h14) to [out=170,in=-30]  (h13);
    
    \draw[->]  (h1) to [out=-40,in=130]  (h5);
    \draw[->]  (h2) to [out=10,in=170]  (h3);  
    \draw[->]  (h6) to [out=10,in=170]  (h7);
    \draw[->]  (h8) to [out=40,in=230]  (h4);
   
    \draw[->]  (h9) to [out=230,in=40]  (h13);
    \draw[t,->]  (h10) to [out=260,in=60]  (h14);  
    \draw[t,->]  (h11) to [out=280,in=120]  (h15);
    \draw[t,->]  (h12) to [out=-50,in=150]  (h16);     
    
    \draw[s,->] (h1) to [out=-30, in=180] (h6);
    \draw[s,->] (h6) to [out=-30, in=180] (h8);
    \draw[sx,->] (h8) to [out=90, in=-60] (h3);
    \draw[s,->] (h3) to [out=190, in=60] (h6);
    
    \draw[s,->] (h9) to [out=260, in=150] (h14);
    \draw[s,->] (h11) to [out=200, in=60] (h14);
    \draw[s,->] (h16) to [out=170, in=-60] (h11);
    \draw[s,->] (h8) to [out=280, in=120] (h16);

    \end{pgfonlayer}
}

\newcommand{\PlpIV}{%
    \definecolor{scarlet}{rgb}{1.0,0.14,0}
    \definecolor{orange}{rgb}{1.0,0.64,0}
    \definecolor{green}{rgb}{0.0,0.5,0.0}
    \definecolor{lgreen}{rgb}{0.2,0.9,0.2}
    \tikzstyle{node_styleb} = [circle,draw=black,fill=blue, inner sep=0pt, minimum size=4pt, line width=.5]
    \tikzstyle{node_styleo} = [rectangle,draw=black,fill=orange, inner sep=0pt, minimum size=4pt, line width=.5]
    \tikzstyle{node_styles} = [circle,draw=black,fill=scarlet, inner sep=0pt, minimum size=4pt, line width=.5]
    
    \tikzstyle{nhidden_style} = [inner sep=0pt, minimum size=0pt]
    \tikzstyle{t} = [draw=black, line width=1.5]
    \tikzstyle{s} = [draw=green, line width=1.5]
    \tikzstyle{sx} = [draw=lgreen, line width=1.5]
     
    \node[nhidden_style](b) at (0,-3) {};

    \node[node_styleb](h1) at (-3.5,2.3) {};
    \node[node_styleo](h2) at (-1.2,3) {};
    \node[node_styles](h3) at (1.2,3) {};
    \node[node_styleo](h4) at (3.5,2.3) {};
  
    \node[node_styleo](h5) at (-2,0.9) {};
    \node[node_styleb](h6) at (-0.7,1.4) {};
    \node[node_styleo](h7) at (0.7,1.4) {};
    \node[node_styleb](h8) at (2,0.9) {};
  
    \node[node_styleb](h9) at (-2,-0.9) {};
    \node[node_styleo](h10) at (-0.7,-1.4) {};
    \node[node_styleb](h11) at (0.7,-1.4) {};
    \node[node_styleo](h12) at (2,-0.9) {};
  
    \node[node_styleo](h13) at (-3.5,-2.3) {};
    \node[node_styleb](h14) at (-1.2,-3) {};
    \node[node_styleo](h15) at (1.2,-3) {};
    \node[node_styleb](h16) at (3.5,-2.3) {};
    
    \begin{pgfonlayer}{bg} 

    \draw[t,->]  (h13) to [out=105,in=255]  (h1);
    \draw[t,->]  (h1) to [out=30,in=190]  (h2);
    \draw[t,->]  (h2) to [out=300,in=100]  (h6);
    \draw[t,->]  (h6) to [out=190,in=20]  (h5);
    \draw[t,->]  (h5) to [out=262,in=98]  (h9);
    \draw[t,->]  (h9) to [out=-20,in=170]  (h10);
    \draw[->]  (h10) to [out=0,in=180]  (h11);
    \draw[t,->]  (h11) to [out=10,in=210]  (h12); 
    \draw[->]  (h12) to [out=82,in=278]  (h8); 
    \draw[t,->]  (h8) to [out=150,in=-10]  (h7); 
    \draw[t,->]  (h7) to [out=80,in=240]  (h3); 
    \draw[->]  (h3) to [out=-10,in=150]  (h4); 
    \draw[t,->]  (h4) to [out=285,in=75]  (h16); 
    \draw[->]  (h16) to [out=200,in=10]  (h15); 
    \draw[t,->]  (h15) to [out=190,in=-10]  (h14); 
    \draw[->]  (h14) to [out=170,in=-30]  (h13);
    
    \draw[->]  (h1) to [out=-40,in=130]  (h5);
    \draw[->]  (h2) to [out=10,in=170]  (h3);  
    \draw[->]  (h6) to [out=10,in=170]  (h7);
    \draw[t,->]  (h8) to [out=40,in=230]  (h4);
   
    \draw[->]  (h9) to [out=230,in=40]  (h13);
    \draw[t,->]  (h10) to [out=260,in=60]  (h14);  
    \draw[t,->]  (h11) to [out=280,in=120]  (h15);
    \draw[t,->]  (h12) to [out=-50,in=150]  (h16);     
    
    \draw[s,->] (h1) to [out=-30, in=180] (h6);
    \draw[s,->] (h6) to [out=-30, in=180] (h8);
    \draw[s,->] (h3) to [out=-60, in=90] (h8);
    \draw[sx,->] (h3) to [out=190, in=60] (h6);
    
    \draw[s,->] (h9) to [out=260, in=150] (h14);
    \draw[s,->] (h11) to [out=200, in=60] (h14);
    \draw[s,->] (h16) to [out=170, in=-60] (h11);
    \draw[s,->] (h8) to [out=280, in=120] (h16);

    \end{pgfonlayer}
}

\newcommand{\PlpV}{%
    \definecolor{scarlet}{rgb}{1.0,0.14,0}
    \definecolor{orange}{rgb}{1.0,0.64,0}
    \definecolor{green}{rgb}{0.0,0.5,0.0}
    \definecolor{lgreen}{rgb}{0.2,0.9,0.2}
    \tikzstyle{node_styleb} = [circle,draw=black,fill=blue, inner sep=0pt, minimum size=4pt, line width=.5]
    \tikzstyle{node_styleo} = [rectangle,draw=black,fill=orange, inner sep=0pt, minimum size=4pt, line width=.5]
    \tikzstyle{node_styles} = [circle,draw=black,fill=scarlet, inner sep=0pt, minimum size=4pt, line width=.5]
    
    \tikzstyle{nhidden_style} = [inner sep=0pt, minimum size=0pt]
    \tikzstyle{t} = [draw=black, line width=1.5]
    \tikzstyle{s} = [draw=green, line width=1.5]
    \tikzstyle{sx} = [draw=lgreen, line width=1.5]
     
    \node[nhidden_style](b) at (0,-3) {};

    \node[node_styleb](h1) at (-3.5,2.3) {};
    \node[node_styleo](h2) at (-1.2,3) {};
    \node[node_styleb](h3) at (1.2,3) {};
    \node[node_styleo](h4) at (3.5,2.3) {};
  
    \node[node_styleo](h5) at (-2,0.9) {};
    \node[node_styles](h6) at (-0.7,1.4) {};
    \node[node_styleo](h7) at (0.7,1.4) {};
    \node[node_styleb](h8) at (2,0.9) {};
  
    \node[node_styleb](h9) at (-2,-0.9) {};
    \node[node_styleo](h10) at (-0.7,-1.4) {};
    \node[node_styleb](h11) at (0.7,-1.4) {};
    \node[node_styleo](h12) at (2,-0.9) {};
  
    \node[node_styleo](h13) at (-3.5,-2.3) {};
    \node[node_styleb](h14) at (-1.2,-3) {};
    \node[node_styleo](h15) at (1.2,-3) {};
    \node[node_styleb](h16) at (3.5,-2.3) {};
    
    \begin{pgfonlayer}{bg} 

    \draw[t,->]  (h13) to [out=105,in=255]  (h1);
    \draw[t,->]  (h1) to [out=30,in=190]  (h2);
    \draw[->]  (h2) to [out=300,in=100]  (h6);
    \draw[t,->]  (h6) to [out=190,in=20]  (h5);
    \draw[t,->]  (h5) to [out=262,in=98]  (h9);
    \draw[t,->]  (h9) to [out=-20,in=170]  (h10);
    \draw[->]  (h10) to [out=0,in=180]  (h11);
    \draw[t,->]  (h11) to [out=10,in=210]  (h12); 
    \draw[->]  (h12) to [out=82,in=278]  (h8); 
    \draw[t,->]  (h8) to [out=150,in=-10]  (h7); 
    \draw[t,->]  (h7) to [out=80,in=240]  (h3); 
    \draw[->]  (h3) to [out=-10,in=150]  (h4); 
    \draw[t,->]  (h4) to [out=285,in=75]  (h16); 
    \draw[->]  (h16) to [out=200,in=10]  (h15); 
    \draw[t,->]  (h15) to [out=190,in=-10]  (h14); 
    \draw[->]  (h14) to [out=170,in=-30]  (h13);
    
    \draw[->]  (h1) to [out=-40,in=130]  (h5);
    \draw[t,->]  (h2) to [out=10,in=170]  (h3);  
    \draw[->]  (h6) to [out=10,in=170]  (h7);
    \draw[t,->]  (h8) to [out=40,in=230]  (h4);
   
    \draw[->]  (h9) to [out=230,in=40]  (h13);
    \draw[t,->]  (h10) to [out=260,in=60]  (h14);  
    \draw[t,->]  (h11) to [out=280,in=120]  (h15);
    \draw[t,->]  (h12) to [out=-50,in=150]  (h16);     
    
    \draw[s,->] (h1) to [out=-30, in=180] (h6);
    \draw[sx,->] (h6) to [out=-30, in=180] (h8);
    \draw[s,->] (h3) to [out=-60, in=90] (h8);
    \draw[s,->] (h6) to [out=60, in=190] (h3);
    
    \draw[s,->] (h9) to [out=260, in=150] (h14);
    \draw[s,->] (h11) to [out=200, in=60] (h14);
    \draw[s,->] (h16) to [out=170, in=-60] (h11);
    \draw[s,->] (h8) to [out=280, in=120] (h16);

    \end{pgfonlayer}
}

\newcommand{\PlpVI}{%
    \definecolor{scarlet}{rgb}{1.0,0.14,0}
    \definecolor{orange}{rgb}{1.0,0.64,0}
    \definecolor{green}{rgb}{0.0,0.5,0.0}
    \definecolor{lgreen}{rgb}{0.2,0.9,0.2}
    \tikzstyle{node_styleb} = [circle,draw=black,fill=blue, inner sep=0pt, minimum size=4pt, line width=.5]
    \tikzstyle{node_styleo} = [rectangle,draw=black,fill=orange, inner sep=0pt, minimum size=4pt, line width=.5]
    \tikzstyle{node_styles} = [circle,draw=black,fill=scarlet, inner sep=0pt, minimum size=4pt, line width=.5]
    
    \tikzstyle{nhidden_style} = [inner sep=0pt, minimum size=0pt]
    \tikzstyle{t} = [draw=black, line width=1.5]
    \tikzstyle{s} = [draw=green, line width=1.5]
    \tikzstyle{sx} = [draw=lgreen, line width=1.5]
 
    \node[nhidden_style](b) at (0,-3) {};

    \node[node_styleb](h1) at (-3.5,2.3) {};
    \node[node_styleo](h2) at (-1.2,3) {};
    \node[node_styleb](h3) at (1.2,3) {};
    \node[node_styleo](h4) at (3.5,2.3) {};
  
    \node[node_styleo](h5) at (-2,0.9) {};
    \node[node_styleb](h6) at (-0.7,1.4) {};
    \node[node_styleo](h7) at (0.7,1.4) {};
    \node[node_styles](h8) at (2,0.9) {};
  
    \node[node_styleb](h9) at (-2,-0.9) {};
    \node[node_styleo](h10) at (-0.7,-1.4) {};
    \node[node_styleb](h11) at (0.7,-1.4) {};
    \node[node_styleo](h12) at (2,-0.9) {};
  
    \node[node_styleo](h13) at (-3.5,-2.3) {};
    \node[node_styleb](h14) at (-1.2,-3) {};
    \node[node_styleo](h15) at (1.2,-3) {};
    \node[node_styleb](h16) at (3.5,-2.3) {};
    
    \begin{pgfonlayer}{bg} 

    \draw[t,->]  (h13) to [out=105,in=255]  (h1);
    \draw[t,->]  (h1) to [out=30,in=190]  (h2);
    \draw[->]  (h2) to [out=300,in=100]  (h6);
    \draw[t,->]  (h6) to [out=190,in=20]  (h5);
    \draw[t,->]  (h5) to [out=262,in=98]  (h9);
    \draw[t,->]  (h9) to [out=-20,in=170]  (h10);
    \draw[->]  (h10) to [out=0,in=180]  (h11);
    \draw[t,->]  (h11) to [out=10,in=210]  (h12); 
    \draw[->]  (h12) to [out=82,in=278]  (h8); 
    \draw[->]  (h8) to [out=150,in=-10]  (h7); 
    \draw[t,->]  (h7) to [out=80,in=240]  (h3); 
    \draw[->]  (h3) to [out=-10,in=150]  (h4); 
    \draw[t,->]  (h4) to [out=285,in=75]  (h16); 
    \draw[->]  (h16) to [out=200,in=10]  (h15); 
    \draw[t,->]  (h15) to [out=190,in=-10]  (h14); 
    \draw[->]  (h14) to [out=170,in=-30]  (h13);
    
    \draw[->]  (h1) to [out=-40,in=130]  (h5);
    \draw[t,->]  (h2) to [out=10,in=170]  (h3);  
    \draw[t,->]  (h6) to [out=10,in=170]  (h7);
    \draw[t,->]  (h8) to [out=40,in=230]  (h4);
   
    \draw[->]  (h9) to [out=230,in=40]  (h13);
    \draw[t,->]  (h10) to [out=260,in=60]  (h14);  
    \draw[t,->]  (h11) to [out=280,in=120]  (h15);
    \draw[t,->]  (h12) to [out=-50,in=150]  (h16);     
    
    \draw[s,->] (h1) to [out=-30, in=180] (h6);
    \draw[s,->] (h8) to [out=180, in=-30] (h6);
    \draw[s,->] (h3) to [out=-60, in=90] (h8);
    \draw[s,->] (h6) to [out=60, in=190] (h3);
    
    \draw[s,->] (h9) to [out=260, in=150] (h14);
    \draw[s,->] (h11) to [out=200, in=60] (h14);
    \draw[s,->] (h16) to [out=170, in=-60] (h11);
    \draw[sx,->] (h8) to [out=280, in=120] (h16);

    \end{pgfonlayer}
}

\newcommand{\PlpVII}{%
    \definecolor{scarlet}{rgb}{1.0,0.14,0}
    \definecolor{orange}{rgb}{1.0,0.64,0}
    \definecolor{green}{rgb}{0.0,0.5,0.0}
    \definecolor{lgreen}{rgb}{0.2,0.9,0.2}
    \tikzstyle{node_styleb} = [circle,draw=black,fill=blue, inner sep=0pt, minimum size=4pt, line width=.5]
    \tikzstyle{node_styleo} = [rectangle,draw=black,fill=orange, inner sep=0pt, minimum size=4pt, line width=.5]
    \tikzstyle{node_styles} = [circle,draw=black,fill=scarlet, inner sep=0pt, minimum size=4pt, line width=.5]
    
    \tikzstyle{nhidden_style} = [inner sep=0pt, minimum size=0pt]
    \tikzstyle{t} = [draw=black, line width=1.5]
    \tikzstyle{s} = [draw=green, line width=1.5]
    \tikzstyle{sx} = [draw=lgreen, line width=1.5]
    \node[nhidden_style](b) at (0,-3) {};

    \node[node_styleb](h1) at (-3.5,2.3) {};
    \node[node_styleo](h2) at (-1.2,3) {};
    \node[node_styleb](h3) at (1.2,3) {};
    \node[node_styleo](h4) at (3.5,2.3) {};
  
    \node[node_styleo](h5) at (-2,0.9) {};
    \node[node_styleb](h6) at (-0.7,1.4) {};
    \node[node_styleo](h7) at (0.7,1.4) {};
    \node[node_styleb](h8) at (2,0.9) {};
  
    \node[node_styleb](h9) at (-2,-0.9) {};
    \node[node_styleo](h10) at (-0.7,-1.4) {};
    \node[node_styleb](h11) at (0.7,-1.4) {};
    \node[node_styleo](h12) at (2,-0.9) {};
  
    \node[node_styleo](h13) at (-3.5,-2.3) {};
    \node[node_styleb](h14) at (-1.2,-3) {};
    \node[node_styleo](h15) at (1.2,-3) {};
    \node[node_styles](h16) at (3.5,-2.3) {};
    
    \begin{pgfonlayer}{bg} 

    \draw[t,->]  (h13) to [out=105,in=255]  (h1);
    \draw[t,->]  (h1) to [out=30,in=190]  (h2);
    \draw[->]  (h2) to [out=300,in=100]  (h6);
    \draw[t,->]  (h6) to [out=190,in=20]  (h5);
    \draw[t,->]  (h5) to [out=262,in=98]  (h9);
    \draw[t,->]  (h9) to [out=-20,in=170]  (h10);
    \draw[->]  (h10) to [out=0,in=180]  (h11);
    \draw[t,->]  (h11) to [out=10,in=210]  (h12); 
    \draw[t,->]  (h12) to [out=82,in=278]  (h8); 
    \draw[->]  (h8) to [out=150,in=-10]  (h7); 
    \draw[t,->]  (h7) to [out=80,in=240]  (h3); 
    \draw[->]  (h3) to [out=-10,in=150]  (h4); 
    \draw[t,->]  (h4) to [out=285,in=75]  (h16); 
    \draw[->]  (h16) to [out=200,in=10]  (h15); 
    \draw[t,->]  (h15) to [out=190,in=-10]  (h14); 
    \draw[->]  (h14) to [out=170,in=-30]  (h13);
    
    \draw[->]  (h1) to [out=-40,in=130]  (h5);
    \draw[t,->]  (h2) to [out=10,in=170]  (h3);  
    \draw[t,->]  (h6) to [out=10,in=170]  (h7);
    \draw[t,->]  (h8) to [out=40,in=230]  (h4);
   
    \draw[->]  (h9) to [out=230,in=40]  (h13);
    \draw[t,->]  (h10) to [out=260,in=60]  (h14);  
    \draw[t,->]  (h11) to [out=280,in=120]  (h15);
    \draw[->]  (h12) to [out=-50,in=150]  (h16);     
    
    \draw[s,->] (h1) to [out=-30, in=180] (h6);
    \draw[s,->] (h8) to [out=180, in=-30] (h6);
    \draw[s,->] (h3) to [out=-60, in=90] (h8);
    \draw[s,->] (h6) to [out=60, in=190] (h3);
    
    \draw[s,->] (h9) to [out=260, in=150] (h14);
    \draw[s,->] (h11) to [out=200, in=60] (h14);
    \draw[sx,->] (h16) to [out=170, in=-60] (h11);
    \draw[s,->] (h16) to [out=120, in=280] (h8);

    \end{pgfonlayer}
}

\newcommand{\PlpVIII}{%
    \definecolor{scarlet}{rgb}{1.0,0.14,0}
    \definecolor{orange}{rgb}{1.0,0.64,0}
    \definecolor{green}{rgb}{0.0,0.5,0.0}
    \definecolor{lgreen}{rgb}{0.2,0.9,0.2}
    \tikzstyle{node_styleb} = [circle,draw=black,fill=blue, inner sep=0pt, minimum size=4pt, line width=.5]
    \tikzstyle{node_styleo} = [rectangle,draw=black,fill=orange, inner sep=0pt, minimum size=4pt, line width=.5]
    \tikzstyle{node_styles} = [circle,draw=black,fill=scarlet, inner sep=0pt, minimum size=4pt, line width=.5]
    
    \tikzstyle{nhidden_style} = [inner sep=0pt, minimum size=0pt]
    \tikzstyle{t} = [draw=black, line width=1.5]
    \tikzstyle{s} = [draw=green, line width=1.5]
    \tikzstyle{sx} = [draw=lgreen, line width=1.5]
    
    \node[nhidden_style](b) at (0,-3) {};

    \node[node_styleb](h1) at (-3.5,2.3) {};
    \node[node_styleo](h2) at (-1.2,3) {};
    \node[node_styleb](h3) at (1.2,3) {};
    \node[node_styleo](h4) at (3.5,2.3) {};
  
    \node[node_styleo](h5) at (-2,0.9) {};
    \node[node_styleb](h6) at (-0.7,1.4) {};
    \node[node_styleo](h7) at (0.7,1.4) {};
    \node[node_styleb](h8) at (2,0.9) {};
  
    \node[node_styleb](h9) at (-2,-0.9) {};
    \node[node_styleo](h10) at (-0.7,-1.4) {};
    \node[node_styles](h11) at (0.7,-1.4) {};
    \node[node_styleo](h12) at (2,-0.9) {};
  
    \node[node_styleo](h13) at (-3.5,-2.3) {};
    \node[node_styleb](h14) at (-1.2,-3) {};
    \node[node_styleo](h15) at (1.2,-3) {};
    \node[node_styleb](h16) at (3.5,-2.3) {};
    
    \begin{pgfonlayer}{bg} 

    \draw[t,->]  (h13) to [out=105,in=255]  (h1);
    \draw[t,->]  (h1) to [out=30,in=190]  (h2);
    \draw[->]  (h2) to [out=300,in=100]  (h6);
    \draw[t,->]  (h6) to [out=190,in=20]  (h5);
    \draw[t,->]  (h5) to [out=262,in=98]  (h9);
    \draw[t,->]  (h9) to [out=-20,in=170]  (h10);
    \draw[->]  (h10) to [out=0,in=180]  (h11);
    \draw[t,->]  (h11) to [out=10,in=210]  (h12); 
    \draw[t,->]  (h12) to [out=82,in=278]  (h8); 
    \draw[->]  (h8) to [out=150,in=-10]  (h7); 
    \draw[t,->]  (h7) to [out=80,in=240]  (h3); 
    \draw[->]  (h3) to [out=-10,in=150]  (h4); 
    \draw[t,->]  (h4) to [out=285,in=75]  (h16); 
    \draw[t,->]  (h16) to [out=200,in=10]  (h15); 
    \draw[t,->]  (h15) to [out=190,in=-10]  (h14); 
    \draw[->]  (h14) to [out=170,in=-30]  (h13);
    
    \draw[->]  (h1) to [out=-40,in=130]  (h5);
    \draw[t,->]  (h2) to [out=10,in=170]  (h3);  
    \draw[t,->]  (h6) to [out=10,in=170]  (h7);
    \draw[t,->]  (h8) to [out=40,in=230]  (h4);
   
    \draw[->]  (h9) to [out=230,in=40]  (h13);
    \draw[t,->]  (h10) to [out=260,in=60]  (h14);  
    \draw[->]  (h11) to [out=280,in=120]  (h15);
    \draw[->]  (h12) to [out=-50,in=150]  (h16);     
    
    \draw[s,->] (h1) to [out=-30, in=180] (h6);
    \draw[s,->] (h8) to [out=180, in=-30] (h6);
    \draw[s,->] (h3) to [out=-60, in=90] (h8);
    \draw[s,->] (h6) to [out=60, in=190] (h3);
    
    \draw[s,->] (h9) to [out=260, in=150] (h14);
    \draw[sx,->] (h11) to [out=200, in=60] (h14);
    \draw[s,->] (h11) to [out=-60, in=170] (h16);
    \draw[s,->] (h16) to [out=120, in=280] (h8);

    \end{pgfonlayer}
}

\newcommand{\PlpIX}{%
    \definecolor{scarlet}{rgb}{1.0,0.14,0}
    \definecolor{orange}{rgb}{1.0,0.64,0}
    \definecolor{green}{rgb}{0.0,0.5,0.0}
    \tikzstyle{node_styleb} = [circle,draw=black,fill=blue, inner sep=0pt, minimum size=4pt, line width=.5]
    \tikzstyle{node_styleo} = [rectangle,draw=black,fill=orange, inner sep=0pt, minimum size=4pt, line width=.5]
    \tikzstyle{node_styles} = [circle,draw=black,fill=scarlet, inner sep=0pt, minimum size=4pt, line width=.5]
    
    \tikzstyle{nhidden_style} = [inner sep=0pt, minimum size=0pt]
    \tikzstyle{t} = [draw=black, line width=1.5]
    \tikzstyle{s} = [draw=green, line width=1.5]
 
    \node[nhidden_style](b) at (0,-3) {};

    \node[node_styleb](h1) at (-3.5,2.3) {};
    \node[node_styleo](h2) at (-1.2,3) {};
    \node[node_styleb](h3) at (1.2,3) {};
    \node[node_styleo](h4) at (3.5,2.3) {};
  
    \node[node_styleo](h5) at (-2,0.9) {};
    \node[node_styleb](h6) at (-0.7,1.4) {};
    \node[node_styleo](h7) at (0.7,1.4) {};
    \node[node_styleb](h8) at (2,0.9) {};
  
    \node[node_styleb](h9) at (-2,-0.9) {};
    \node[node_styleo](h10) at (-0.7,-1.4) {};
    \node[node_styleb](h11) at (0.7,-1.4) {};
    \node[node_styleo](h12) at (2,-0.9) {};
  
    \node[node_styleo](h13) at (-3.5,-2.3) {};
    \node[node_styles](h14) at (-1.2,-3) {};
    \node[node_styleo](h15) at (1.2,-3) {};
    \node[node_styleb](h16) at (3.5,-2.3) {};
    
    \begin{pgfonlayer}{bg} 

    \draw[t,->]  (h13) to [out=105,in=255]  (h1);
    \draw[t,->]  (h1) to [out=30,in=190]  (h2);
    \draw[->]  (h2) to [out=300,in=100]  (h6);
    \draw[t,->]  (h6) to [out=190,in=20]  (h5);
    \draw[t,->]  (h5) to [out=262,in=98]  (h9);
    \draw[t,->]  (h9) to [out=-20,in=170]  (h10);
    \draw[t,->]  (h10) to [out=0,in=180]  (h11);
    \draw[t,->]  (h11) to [out=10,in=210]  (h12); 
    \draw[t,->]  (h12) to [out=82,in=278]  (h8); 
    \draw[->]  (h8) to [out=150,in=-10]  (h7); 
    \draw[t,->]  (h7) to [out=80,in=240]  (h3); 
    \draw[->]  (h3) to [out=-10,in=150]  (h4); 
    \draw[t,->]  (h4) to [out=285,in=75]  (h16); 
    \draw[t,->]  (h16) to [out=200,in=10]  (h15); 
    \draw[t,->]  (h15) to [out=190,in=-10]  (h14); 
    \draw[->]  (h14) to [out=170,in=-30]  (h13);
    
    \draw[->]  (h1) to [out=-40,in=130]  (h5);
    \draw[t,->]  (h2) to [out=10,in=170]  (h3);  
    \draw[t,->]  (h6) to [out=10,in=170]  (h7);
    \draw[t,->]  (h8) to [out=40,in=230]  (h4);
   
    \draw[->]  (h9) to [out=230,in=40]  (h13);
    \draw[->]  (h10) to [out=260,in=60]  (h14);  
    \draw[->]  (h11) to [out=280,in=120]  (h15);
    \draw[->]  (h12) to [out=-50,in=150]  (h16);     
    
    \draw[s,->] (h1) to [out=-30, in=180] (h6);
    \draw[s,->] (h8) to [out=180, in=-30] (h6);
    \draw[s,->] (h3) to [out=-60, in=90] (h8);
    \draw[s,->] (h6) to [out=60, in=190] (h3);
    
    \draw[s,->] (h9) to [out=260, in=150] (h14);
    \draw[s,->] (h14) to [out=60, in=200] (h11);
    \draw[s,->] (h11) to [out=-60, in=170] (h16);
    \draw[s,->] (h16) to [out=120, in=280] (h8);

    \end{pgfonlayer}
}

\newcommand{\Hlpprime}{%
    \definecolor{scarlet}{rgb}{1.0,0.14,0}
    \definecolor{orange}{rgb}{1.0,0.64,0}
    \definecolor{green}{rgb}{0.0,0.5,0.0}
    \tikzstyle{node_styleb} = [circle,draw=black,fill=blue, inner sep=0pt, minimum size=4pt, line width=.5]
    \tikzstyle{node_styleo} = [rectangle,draw=black,fill=orange, inner sep=0pt, minimum size=4pt, line width=.5]
    
    \tikzstyle{nhidden_style} = [inner sep=0pt, minimum size=0pt]
    \tikzstyle{t} = [draw=black, line width=1.5]
    \tikzstyle{s} = [draw=green, line width=1.5]
 
    \node[nhidden_style](b) at (0,-3) {};

    \node[node_styleb](h1) at (-3.5,2.3) {};
    \node[node_styleo](h2) at (-1.2,3) {};
    \node[node_styleb](h3) at (1.2,3) {};
    \node[node_styleo](h4) at (3.5,2.3) {};
  
    \node[node_styleo](h5) at (-2,0.9) {};
    \node[node_styleb](h6) at (-0.7,1.4) {};
    \node[node_styleo](h7) at (0.7,1.4) {};
    \node[node_styleb](h8) at (2,0.9) {};
  
    \node[node_styleb](h9) at (-2,-0.9) {};
    \node[node_styleo](h10) at (-0.7,-1.4) {};
    \node[node_styleb](h11) at (0.7,-1.4) {};
    \node[node_styleo](h12) at (2,-0.9) {};
  
    \node[node_styleo](h13) at (-3.5,-2.3) {};
    \node[node_styleb](h14) at (-1.2,-3) {};
    \node[node_styleo](h15) at (1.2,-3) {};
    \node[node_styleb](h16) at (3.5,-2.3) {};
    
    \begin{pgfonlayer}{bg} 

    \draw[t,->]  (h13) to [out=105,in=255]  (h1);
    \draw[t,->]  (h1) to [out=30,in=190]  (h2);
    \draw[->]  (h2) to [out=300,in=100]  (h6);
    \draw[t,->]  (h6) to [out=190,in=20]  (h5);
    \draw[t,->]  (h5) to [out=262,in=98]  (h9);
    \draw[t,->]  (h9) to [out=-20,in=170]  (h10);
    \draw[t,->]  (h10) to [out=0,in=180]  (h11);
    \draw[t,->]  (h11) to [out=10,in=210]  (h12); 
    \draw[t,->]  (h12) to [out=82,in=278]  (h8); 
    \draw[->]  (h8) to [out=150,in=-10]  (h7); 
    \draw[t,->]  (h7) to [out=80,in=240]  (h3); 
    \draw[->]  (h3) to [out=-10,in=150]  (h4); 
    \draw[t,->]  (h4) to [out=285,in=75]  (h16); 
    \draw[t,->]  (h16) to [out=200,in=10]  (h15); 
    \draw[t,->]  (h15) to [out=190,in=-10]  (h14); 
    \draw[t,->]  (h14) to [out=170,in=-30]  (h13);
    
    \draw[->]  (h1) to [out=-40,in=130]  (h5);
    \draw[t,->]  (h2) to [out=10,in=170]  (h3);  
    \draw[t,->]  (h6) to [out=10,in=170]  (h7);
    \draw[t,->]  (h8) to [out=40,in=230]  (h4);
   
    \draw[->]  (h9) to [out=230,in=40]  (h13);
    \draw[->]  (h10) to [out=260,in=60]  (h14);  
    \draw[->]  (h11) to [out=280,in=120]  (h15);
    \draw[->]  (h12) to [out=-50,in=150]  (h16);     
   
    \end{pgfonlayer}
}

\title[]{Another Hamiltonian cycle in bipartite Pfaffian graphs}

\date{}

\ifdefined\A
\author{Anonymous}
\else
\author{Andreas Bj\"orklund}
\address[AB]{IT University of Copenhagen, Denmark}
\email{anbjo@itu.dk}
\author{Petteri Kaski}
\address[PK]{Aalto University, Finland}
\email{petteri.kaski@aalto.fi}
\author{Jesper Nederlof}
\address[JN]{Utrecht University, the Netherlands}
\email{j.nederlof@uu.nl}
\fi

\begin{document}

\maketitle

\begin{abstract}
Finding a Hamiltonian cycle in a given graph is computationally challenging, 
and in general remains so even when one is further given one Hamiltonian 
cycle in the graph and asked to find another.
In fact, no significantly faster algorithms are known for finding
another Hamiltonian cycle than for finding a first one even in the setting
where another Hamiltonian cycle is structurally guaranteed to exist, such 
as for odd-degree graphs. 
We identify a graph class---the bipartite Pfaffian graphs of minimum degree 
three---where it is NP-complete to decide whether a given graph in the class
is Hamiltonian, but when presented with a Hamiltonian cycle as part of 
the input, another Hamiltonian cycle can be found efficiently. 

We prove that Thomason's lollipop method~[Ann.~Discrete Math.,~1978], 
a well-known algorithm for finding another Hamiltonian cycle, 
runs in a linear number of steps in cubic bipartite Pfaffian graphs. This was conjectured 
for cubic bipartite planar graphs by Haddadan [MSc~thesis,~Waterloo,~2015];
in contrast, examples are known of both cubic bipartite graphs and 
cubic planar graphs where the lollipop method takes exponential time.

Beyond the lollipop method, we address a slightly more general 
graph class and present two algorithms, one running in linear-time and 
one operating in logarithmic space, that take as input
(i) a bipartite Pfaffian graph $G$ of minimum degree three, 
(ii) a Hamiltonian cycle $H$ in $G$, and 
(iii) an edge $e$ in $H$, and
output at least three other Hamiltonian cycles through the edge $e$ in $G$.

We also present further improved algorithms for finding optimal traveling 
salesperson tours and counting Hamiltonian cycles in bipartite planar graphs 
with running times that are not known to hold in general planar graphs.

We prove our results by a new structural technique that efficiently witnesses 
each Hamiltonian cycle $H$ through an arbitrary fixed anchor edge $e$ 
in a bipartite Pfaffian graph using a two-coloring of the vertices as advice 
that is unique to $H$. Previous techniques---the 
Cut\&Count technique of Cygan, Nederlof, Pilipczuk, Pilipczuk, Van Rooij, 
and Wojtaszczyk [FOCS'11,~TALG'22] in particular---were able to reduce 
the Hamiltonian cycle problem only to essentially \emph{counting} problems; 
our results show that counting can be avoided by leveraging properties of 
bipartite Pfaffian graphs. Our technique also has purely graph-theoretical 
consequences; for example, we show that every cubic bipartite Pfaffian graph 
has either zero or at least six distinct Hamiltonian cycles; the latter case 
is tight for the cube graph.
\end{abstract}

\thispagestyle{empty}

\clearpage 

\setcounter{page}{1}

\section{Introduction}

\noindent
Finding a Hamiltonian cycle in a given undirected graph is 
a well-known, well-researched, and hard problem. 
This paper studies the question whether knowledge of one Hamiltonian cycle 
helps in finding another one. More precisely, 
the \textsc{Another Hamiltonian Cycle} problem asks, given as input 
(i)~a~graph%
\footnote{We tacitly work with undirected simple loopless graphs 
unless mentioned otherwise, as well as assume knowledge of standard 
graph-theoretic terminology~\cite{West1996}.
Our conventions with graphs can be found 
in Section~\ref{sect:conventions}.}{}~$G$, 
(ii)~a~Hamiltonian cycle $H$ in $G$, and 
(iii)~an~edge $e\in E(H)$, 
to find another Hamiltonian cycle $H'\neq H$ in $G$ with $e\in E(H')$.

Our interest is in the class of bipartite {\em Pfaffian}%
\footnote{We postpone a precise definition and motivation of Pfaffian graphs
to Section~\ref{sect:pfaffian}. Planar graphs are Pfaffian.}{}
graphs, a superclass of the bipartite planar graphs. This focus is partly motivated by the fact that both in cubic bipartite graphs and cubic planar graphs, a well-known general algorithm for finding another Hamiltonian cycle, Thomason's lollipop method~\cite{Thomason1978}, requires exponential time in the worst case, c.f. Section~\ref{sect: earlier work}.
Also, the problem of deciding if the graph contains a Hamiltonian cycle at all remains NP-hard in the family of cubic bipartite planar graphs, as proved by 
Akiyama, Nishizeki, and Saito~\cite{AkiyamaNS1980}. 

As our main result, we show that 
\textsc{Another Hamiltonian Cycle} admits both a linear-time algorithm as 
well as a logarithmic-space algorithm
in bipartite Pfaffian graphs of minimum degree three. 
Further restricted to cubic bipartite Pfaffian graphs, we prove that Thomason's lollipop method runs in a linear number of steps and can be implemented to run in linear time.
This is to our knowledge a first example of a nontrivial graph class
where \textsc{Another Hamiltonian Cycle} is efficiently solvable; such 
an example was solicited by Kintali~\cite{Kintali2009}. By trivial we here intend a graph class in which Hamiltonicity detection is NP-hard but is artificially constructed to ensure a simple local rerouting of any Hamiltonian cycle. Rather, in our case the global properties of bipartiteness, Pfaffianity, and the everywhere-local property of minimum-degree three, interplay to enable an efficient algorithm.
Without the minimum-degree constraint, the problem is NP-hard (see Appendix~\ref{sect: NP-hard}).

Our techniques have also purely graph-theoretic consequences. We show
that every cubic bipartite Pfaffian graph has at least three 
other Hamiltonian cycles through 
any edge of a Hamiltonian cycle. All three Hamiltonian cycles can be found in linear time. We also show that such Hamiltonian graphs must have at least six Hamiltonian cycles.
The $8$-vertex cube graph (\begin{tikzpicture}[baseline=-1mm,scale=0.07] \cubegraph \end{tikzpicture}), the canonical example in this class, is an extremal example to both results. It has precisely six distinct Hamiltonian cycles with every graph edge in exactly four of them.

\subsection{Motivation and earlier work}
\label{sect: earlier work}
While a graph need not be Hamiltonian, and a Hamiltonian graph
need not admit another Hamiltonian cycle, there exist graph families 
with Hamiltonian members where another Hamiltonian cycle is always 
known to exist. 
Perhaps the most prominent such family are the odd-degree graphs, which 
via Smith's Theorem (see~\cite{Tutte1946}) have an even number 
of Hamiltonian cycles through any given edge. 
Thomason~\cite{Thomason1978} 
gave a constructive proof by describing an algorithm 
that solves for another Hamiltonian cycle in cubic graphs. The algorithm is often called Thomason's lollipop method as it transforms a Hamiltonian cycle to another one by a sequence of lollipop graphs, see Section~\ref{sect: lollipop} for a precise description of the algorithm. 
Dropping the requirement that the Hamiltonian cycle should go through a 
specific edge, Bos\'ak~\cite{Bosak1967} proved that every cubic 
bipartite graph has an even number of Hamiltonian cycles.
Thomassen~\cite{Thomassen1996, Thomasson1997} showed that no bipartite graph in 
which every vertex in one of the two parts of the bipartition has degree 
at least three, has a unique Hamiltonian cycle.  A famous 
conjecture due to Sheehan~\cite{Sheehan1975} claims that 
no $4$-regular graph can have a unique Hamiltonian cycle;
Thomassen's result proves Sheehan's conjecture for bipartite graphs.

Papadimitriou~\cite{Papadimitriou1994} popularized 
the \textsc{Another Hamiltonian Cycle} problem and Thomason's algorithm 
by introducing the complexity class PPA and showed the containment of 
the problem in odd-degree graphs; completeness for PPA remains open.
It is also open whether the problem can be solved 
in polynomial time---indeed, the drawback of Thomason's 
algorithm is that it may run for a very long time; 
Krawczyk~\cite{Krawczyk1999} and Cameron~\cite{Cameron2001} showed that 
Thomason's algorithm requires exponential time for a family of cubic planar 
graphs. Later this was shown by Zhong~\cite{Zhong2018} also for cubic 
bipartite graphs. The best bound to date is the recent result 
of Bria\'{n}ski and Szady~\cite{Brianski2022}, which shows that there 
are cubic $3$-connected planar graphs on $n$ vertices, in which 
Thomason's lollipop algorithm runs in $\Omega(1.18^n)$ time.

The above papers reason about Thomason's algorithm specifically, but it 
may of course be other algorithms that solve the problem more efficiently. 
Some progress in this 
direction was provided by Bazgan, Santha, and Tuza~\cite{Bazgan1999}
that showed that one given a cubic graph on $n$ vertices and one of its 
Hamiltonian cycles can find another cycle of length $(1-\epsilon)n$ for 
any fixed constant $\epsilon>0$ in polynomial time. 
Deligkas, Mertzios, Spirakis, and Zamaraev~\cite{Deligkas2020} derived an exponential-time
polynomial-space deterministic algorithm that given a cubic graph along 
with one of its Hamiltonian cycles finds another Hamiltonian cycle;
the algorithm is shown to be faster than the fastest known 
exponential-time polynomial-space deterministic algorithm for finding 
a Hamiltonian cycle in cubic graphs.

\subsection{Main results for bipartite Pfaffian graphs}

\label{sect:main}

Let us now review our main results for bipartite Pfaffian graphs
and the underlying techniques in more detail.
Our main theorems are as follows. 

\begin{restatable}[Main; Linear--time \textsc{Another Hamiltonian Cycle} in minimum degree three]{theorem}{thmahc}
\label{thm:ahc}
There exists a deterministic linear-time algorithm that, given as input
(i) a bipartite Pfaffian graph $G$ with minimum degree three, 
(ii) a Hamiltonian cycle $H$ in $G$, and (iii) an edge $e\in E(H)$, 
outputs a Hamiltonian cycle $H'\neq H$ in $G$ with $e\in E(H')$. 
\end{restatable}
\begin{restatable}[Main; Logarithmic--space \textsc{Another Hamiltonian Cycle} in minimum degree three]{theorem}{thmahclog}
\label{thm:ahc-log}
There exists a deterministic logarithmic-space algorithm that, given as input
(i) a bipartite Pfaffian graph $G$ with minimum degree three, 
(ii) a Hamiltonian cycle $H$ in $G$, and (iii) an edge $e\in E(H)$, 
outputs a Hamiltonian cycle $H'\neq H$ in $G$ with $e\in E(H')$. 
\end{restatable}

One consequence of the above algorithm is that it indirectly proves that there are efficient parallel circuits for generating another Hamiltonian cycle in this graph class, due to the known complexity class containment $\mbox{L}\subseteq \mbox{NC}^2$ by exponentiation of the state transition matrix by iterated squaring.

The framework underlying our main theorems
can also be used to prove an upper bound on the number of steps needed for Thomason's lollipop method to terminate.

\begin{restatable}[Thomason's lollipop method in cubic bipartite Pfaffian graphs]{theorem}{thmtho}
\label{thm:tho}
Thomason's lollipop method starting from any Hamiltonian cycle $H$ and 
any edge $e\in E(H)$ in an $n$-vertex cubic bipartite Pfaffian graph $G$,
terminates after at most $n$ steps.
\end{restatable} 

It was conjectured in a master thesis at Waterloo by Haddadan~\cite{Haddadan2015} that Thomason's lollipop method runs in a linear number of steps in cubic bipartite planar graphs. It was there also proven to hold for the subfamily of such graphs that does not have the wheel graph on six vertices as a minor. However, as the author himself points out, finding a first Hamiltonian cycle in this limited graph family does not seem intractable. We are not aware of any other papers providing a polynomial time bound on Thomason's lollipop method in any graph class.

Our framework also can be used to prove the following structural results for Hamiltonian cycles.

\begin{restatable}[Non-uniqueness in minimum degree three]{corollary}{corahc}
\label{cor:ahc}
For every bipartite Pfaffian graph $G$ of minimum degree three 
and for every edge $e\in E(G)$, it holds that $G$ has either zero or 
at least four distinct Hamiltonian cycles $H$ with $e\in E(H)$.
\end{restatable} 

The cube graph is a cubic bipartite Pfaffian 
graph with six distinct Hamiltonian cycles. We show that no Hamiltonian
graph in this class can have fewer Hamiltonian cycles.

\begin{restatable}[Cubic tight lower bound]{corollary}{corchc}
\label{cor:chc}
Every cubic bipartite Pfaffian graph has either zero or at least six
distinct Hamiltonian cycles.
\end{restatable} 

Chia and Ong~\cite{ChiaOA2007} (in the paragraph after Theorem 10) asked whether there exists cubic bipartite planar graphs with exactly four Hamiltonian cycles. The above Corollary rules out that possibility.

\medskip
\noindent
{\em Remarks.}
One consequence of Theorem~\ref{thm:ahc} is that if 
the \textsc{Another Hamiltonian Cycle} problem in general odd degree graphs 
is PPA-complete as hypothesized by 
Papadimitriou~\cite[Open Problem~(4)]{Papadimitriou1994}, then any proof 
cannot carry over to cubic bipartite planar graphs unless also 
$\mathrm{PPA}=\mathrm{FP}$.
Our result also seems related to another well-known conjecture, 
namely Barnette's conjecture (cf.~Tutte~\cite[Unsolved Problem V]{Tutte1969}),
which states that every cubic $3$-connected bipartite planar graph 
({\em Barnette graph}) has a Hamiltonian cycle.
Gorsky, Steiner, and Wiederrecht~\cite{GorskySW2023} recently extended the conjecture by showing that if Barnette's conjecture is true,
it also holds that every cubic bipartite $3$-connected Pfaffian graph has a Hamiltonian cycle.
It is known that if Barnette's conjecture is true, then there is a Hamiltonian cycle 
through every edge in every such graph, see Kelmans~\cite{Kelmans1994}. 
Moreover, it was indirectly shown by Holton, Manvel, and McKay~\cite{Holton1985} 
that any Barnette graph larger than the smallest such graph---namely
the cube graph---can be reduced to a smaller 
Barnette graph in such a way that if the smaller graph has a Hamiltonian 
cycle through every edge, 
then the larger one must also have a Hamiltonian cycle. This means that 
if it was possible given any single Hamiltonian cycle in a Barnette graph 
to generate a Hamiltonian cycle through any specific edge \emph{not} on 
the initial cycle, then Barnette's conjecture would be constructively true. We remark 
that our algorithm is not known to be able to do this, not even indirectly 
by applying it several times in a chain of Hamiltonian cycle transformations.
We do note however, that our algorithms in Theorems~\ref{thm:ahc} and~\ref{thm:ahc-log} not only make sure the edge $e=\{s,t\}$ is part of both Hamiltonian cycles, they
also retain the other edge incident to $s$ on $H$; this can be observed by Lemma~\ref{lem:s-avoiding-cycle-in-d}, that is, no edge incident to $s$ is changed by the algorithms since it is not on the alternating cycle we use. In particular this makes it possible given $H$ and any given edge $f\in E(G)$ \emph{not} on $H$, to find another Hamiltonian cycle $H'$ such that $f$ is also not part of $H'$.

\subsection{Overview of techniques} 

At the heart of our algorithms and structural results is what we believe
to be a new framework for efficiently witnessing a Hamiltonian cycle $H$
through an arbitrary {\em anchor} edge $e$
in a bipartite Pfaffian graph $G$. We associate a (not necessarily proper) two-coloring 
$\chi_H:V(G)\rightarrow\{0,1\}$ to $H$ that is unique to $H$ (but dependent on
a fixed but arbitrary Pfaffian orientation of $G$ as well as $e$) and 
that defines a unique {\em acyclic} Hamiltonian%
\footnote{In precise terms, the orientation is a directed acyclic graph that
contains as a directed subgraph a directed Hamiltonian {\em path} from 
one end of $e$ to the other.}{}
orientation of $G\setminus e$. We will refer to such an $\chi_H$ as a \emph{good coloring}. This acyclicity in particular enables the
unique recovery of $H$ in linear time by standard topological sorting
when given $\chi_H$ as advice. The reader may want to consult 
Figure~\ref{fig:good-colorings} (on page~\pageref{fig:good-colorings})
for an advance illustration at this point; the framework itself is developed
in Section~\ref{sect:hc}. 

We also show that it suffices to know $\chi_H$ in only one of the parts
of a bipartition of $G$ to efficiently extend to a Hamiltonian cycle,
which is not necessarily equal to $H$ however. More precisely, we 
show that a coloring $\lambda$ of one of the parts leads to an auxiliary 
bipartite graph $F_\lambda$ whose perfect matchings correspond to 
the good colorings $\chi_H$ extending $\lambda$, which in turn each define 
a unique Hamiltonian cycle $H$. We refer to Figure~\ref{fig:f-lambda} 
(on page~\pageref{fig:f-lambda}) for an advance illustration of this setting.

Finally, when $G$ has minimum degree three, we observe that we can
use $F_\lambda$ to efficiently switch from one Hamiltonian cycle $H$ in $G$
(described by a perfect matching $M_H$ in $F_\lambda$)
to another Hamiltonian cycle $H'\neq H$ in $G$ by switching along 
an alternating cycle in $F_\lambda$ which can be discovered through 
a directed cycle in an auxiliary directed graph $D_{\lambda,H}$. 
Moreover, we show that this construction and discovery can be executed 
in deterministic linear time. We refer to Figure~\ref{fig:cycle-switching} 
(on page~\pageref{fig:cycle-switching}) for an advance illustration of 
this setting; the switching construction itself is developed
in Section~\ref{sect:ahc}.

\subsection{Further results}

\label{sect:further}
There are further algorithmic consequences of our framework 
on Pfaffian graphs to the problems of deterministically 
finding and counting Hamiltonian cycles
as well as the Traveling Salesperson Problem (TSP) in the bipartite 
Pfaffian/planar setting.

Our framework for witnessing Hamiltonicity should be contrasted with 
the Cut\&Count approach for detecting Hamiltonian cycles by 
Cygan, Nederlof, Pilipczuk, Pilipczuk, Van Rooij, and Wojtaszczyk~\cite{CyganNPPRWO2022}, which reduces the Hamiltonian 
cycle problem to a local problem by showing that a cycle cover of the input
graph is Hamiltonian if and only if the number of the exponentially many 
vertex partitions that are consistent (defined in a certain local way) 
with it is \emph{odd}. This approach therefore necessarily reduces the 
original decision problem to a parity counting problem, which has several 
disadvantages that seem inherent to the approach, including the need for 
randomization and a running time factor that is pseudo-polynomial in 
the integer weights for edge-weighted problem variants. Our framework shows 
that for bipartite Pfaffian graphs there is a more natural way to witness 
that a cycle cover is Hamiltonian using only a \emph{single} vertex 
partition; that is, $\chi_H$.

\medskip
{\em Earlier work.}
Before stating our results, let us set the stage by reviewing pertinent 
earlier work. 
To start with, let us recall that in terms of complexity lower bounds, 
an algorithm for detecting Hamiltonicity in a given $n$-vertex 
planar graph with worst-case $\operatorname{exp}(o(\sqrt{n}))$ running 
time would violate the Exponential Time Hypothesis, as seen by combining the 
Sparsification Lemma of Impagliazzo, Paturi, and Zane~\cite{Impagliazzo2001} 
with the reduction in Garey, Johnson, and Tarjan~\cite{GareyJT1976}, see De\u{\i}neko, Klinz, and Woeginger~\cite{DeinekoKW06}.

Algorithms for detecting and finding a Hamiltonian cycle in a given $n$-vertex planar
graph running in worst-case $\operatorname{exp}(O(\sqrt{n}))$ 
time have also been known since the work of De\u{\i}neko, Klinz, and Woeginger~\cite{DeinekoKW06}.
Further exploiting the properties of planar graphs, Dorn, Penninkx, Bodlaender, and Fomin~\cite{DornPBF10}
improved the base in the exponential running time significantly. 
Subsequently several powerful algorithmic techniques have been developed 
to get faster algorithms for connectivity problems in general and 
Hamiltonicity problems in particular, including the Cut\&Count technique of 
Cygan, Nederlof, Pilipczuk, Pilipczuk, Van Rooij, and Wojtaszczyk~\cite{CyganNPPRWO2022}, the rank-based method 
with representative sets and reductions via fast Gaussian elimination 
in Bodlaender, Cygan, Kratsch, and Nederlof~\cite{Bodlaender2015}, and especially the 
technique of bases of perfect matchings in Cygan, Kratsch, and Nederlof~\cite{CKN18}.
Pino, Bodlaender, and Van Rooij~\cite{Pino2016,PinoMaster} use these techniques 
over branch width to bring the base of the exponential running time 
down as much as possible in the planar case. 

As far as we know though, there have been no studies showing that the 
input graph being bipartite would help getting even faster algorithms 
in planar graphs. One reason to expect it would is that the fastest known 
algorithm for Hamiltonicity detection in a bipartite graph is much faster 
than the fastest for a general graph, see~Bj\"orklund~\cite{Bjorklund2014}. 
Another one is the TSP algorithm restricted to bipartite graphs by 
Nederlof~\cite{Nederlof2020} that has a $O(c^n)$ running time for 
some $c<2$ if matrix multiplication is in quadratic time, as opposed to 
what is known for general graphs.

\medskip
{\em Our results.} 
By effectively testing all possible coloring-advice bits
in our new Hamiltonicity framework, we get improved 
deterministic algorithms parameterized by various graph decompositions' 
width measures in the class of bipartite Pfaffian (or planar) graphs. 
Earlier parameterized algorithms of this type apply to all graphs, and 
it is not obvious how to further exploit a bipartite structure and modify 
these algorithms for improved running time. Our takeaway message is that 
it nevertheless appears to be easier to count and find weighted Hamiltonian 
cycles in bipartite planar graphs than in general planar ones. 
Let us now state each of our results and highlight a comparison 
with pertinent earlier algorithms. 

\begin{restatable}[Bipartite Pfaffian TSP parameterized by path width]{theorem}{thmtsppw}
\label{thm:TSP-pw}
Given an edge-weighted bipartite Pfaffian graph $G$ on $n$ vertices along 
with a path decomposition of width $\operatorname{pw}(G)$, we can compute 
the minimum weight Hamiltonian cycle in $G$ in 
$4^{\operatorname{pw}(G)}\operatorname{poly}(n)$ time.
\end{restatable}
Bodlaender, Cygan, Kratsch, and Nederlof~\cite{Bodlaender2015} present an algorithm 
running in $(2+2^{\omega/2})^{\operatorname{pw}(G)}\operatorname{poly}(n)$ 
time with $\omega<2.373$ the square matrix multiplication exponent that 
works for \emph{any} graph.
Our algorithm is only asymptotically better if $\omega>2$. 
Our technique also generalizes to other graph decompositions, 
in particular the branch width.

\begin{restatable}[Bipartite Pfaffian TSP parameterized by branch width]{theorem}{thmtspbw}
\label{thm:TSP-bw}
 Given an edge-weighted bipartite Pfaffian graph $G$ on $n$ vertices 
along with a branch decomposition of width $\operatorname{bw}(G)$, 
we can compute the minimum weight Hamiltonian cycle 
in $G$ in $8^{\operatorname{bw}(G)}\operatorname{poly}(n)$ time.
\end{restatable} 
The most striking application of the above result is perhaps that it 
gives an improved running time for \textsc{TSP} in bipartite planar graphs 

\begin{restatable}[Bipartite planar TSP]{corollary}{cortspplanar}
\label{cor:TSP-planar}
Given an edge-weighted bipartite planar graph $G$ on $n$ vertices, we 
can compute the minimum weight Hamiltonian cycle in $G$ 
in $O(2^{6.366\sqrt{n}})$ time.
\end{restatable}

This should be contrasted with the best known result for general planar 
graphs which is the $O(2^{6.570\sqrt{n}})$ time bound by 
Pino, Bodlaender, and Van Rooij~\cite{Pino2016} (see also~\cite{PinoMaster} for the 
full argument). Our result matches the running time of the 
randomized \emph{graphic} planar \textsc{TSP} algorithm in 
Cygan, Nederlof, Pilipczuk, Pilipczuk, Van Rooij, and Wojtaszczyk~\cite{CyganNPPRWO2022}.

Our technique also makes it possible to count the Hamiltonian cycles. 
\begin{restatable}[Bipartite Pfaffian counting Hamiltonian cycles parameterized by path width]{theorem}{thmchcpw}
\label{thm:chc-pw}
Given a bipartite Pfaffian graph $G$ on $n$ vertices along with 
a path decomposition of width $\operatorname{pw}(G)$, we can count 
the Hamiltonian cycles in $G$ in 
$4^{\operatorname{pw}(G)}\operatorname{poly}(n)$ time.
\end{restatable}
The best known result for general graphs is the 
$6^{\operatorname{pw}(G)}\operatorname{poly}(n)$ time algorithm by 
Bodlaender, Cygan, Kratsch, and Nederlof~\cite{Bodlaender2015}. It is also known that 
improving this general graph result 
to $(6-\epsilon)^{\operatorname{pw}(G)}\operatorname{poly}(n)$ 
for any $\epsilon>0$ is impossible unless the 
Strong Exponential Time Hypothesis is false, see 
Curticapean, Lindzey, and Nederlof~\cite{Curticapean2018}. There is also a 
general graph $(2^\omega+2)^{\operatorname{tw}(G)}\operatorname{poly}(n)$ 
time algorithm parameterized in the tree width by 
W\l odarczyk~\cite{Wlodarczyk2019}.

\begin{restatable}[Bipartite Pfaffian counting Hamiltonian cycles parameterized by branch width]{theorem}{thmchcbw}
\label{thm:chc-bw}
Given a bipartite Pfaffian graph $G$ on $n$ vertices along with 
a branch decomposition of width $\operatorname{bw}(G)$, we can count 
the Hamiltonian cycles in $G$ in 
$2^{\omega\operatorname{bw}(G)}\operatorname{poly}(n)$ time, 
where $\omega<2.373$.
\end{restatable}
For this problem we get a much faster algorithm in bipartite planar graphs.

\begin{restatable}[Bipartite planar counting of Hamiltonian cycles]{corollary}{corchcplanar}
\label{cor:chc-planar}
Given a bipartite planar graph $G$ on $n$ vertices, we can compute the 
number of Hamiltonian cycles in $G$ in $O(2^{5.049\sqrt{n}})$ time.
\end{restatable}

The best known bound for general planar graphs as far as we can tell 
is to implicitly build on a result of 
Bodlaender, Cygan, Kratsch, and Nederlof~\cite{Bodlaender2015} to get 
a $O(2^{6.508\sqrt{n}})$ time algorithm.

\subsection{Pfaffian graphs}
\label{sect:pfaffian}

Let us now define and motivate Pfaffian graphs in more detail. 
An {\em orientation} of a graph
$G$ replaces every edge $\{u,v\}\in E(G)$ with 
either the directed arc $(u,v)$ or the directed arc $(v,u)$, thereby 
obtaining a directed graph~$\oa G$. A cycle $C$ in $G$
is {\em central} if the graph $G\setminus V(C)$ admits a perfect matching. 
We say that an orientation of a cycle is {\em consistent} if it is 
strongly connected. 
An orientation $\oa G$ of $G$ is {\em Pfaffian} if for every central cycle $C$
in $G$ it holds that both consistent orientations of $C$ have an odd number
of arcs in common with $\oa G$. A graph is {\em Pfaffian} if it admits
a Pfaffian orientation. 

The bipartite Pfaffian graphs are most famous as the graph class in 
which P\'{o}lya's permanent problem has a solution, the bipartite graphs 
in which one can compute the number of perfect matchings efficiently by
reduction to a matrix determinant, see 
e.g.~Robertson, Seymour, and Thomas~\cite{Robertson1999} and 
McCuaig~\cite{McCuaig2004}. 
This brings us to one of our motivations to study the complexity of the detection (and counting) of Hamiltonian cycles on this graph class: Previous algorithms for Hamiltonian cycles (such as the one by Bj\"orklund~\cite{Bjorklund2014}) use determinant-based methods previously designed for counting matchings (modulo $2$) in polynomial time; given this close connection between the two problems it is natural to ask whether Pfaffianity can be exploited for detecting and counting Hamiltonian cycles, similarly as for counting perfect matchings.

The bipartite Pfaffian graphs were characterized by Little~\cite{Little1975} 
as those graphs $G$ that do not have a vertex set $U\subseteq V(G)$ 
such that $G\setminus U$ has a perfect matching and the induced subgraph
$G[U]$ admits an even subdivision of $K_{3,3}$ as a subgraph. 
McCuaig~\cite{McCuaig2004} and Robertson, Seymour, and 
Thomas~\cite{Robertson1999} gave a structural characterization of bipartite 
Pfaffian graphs and the latter also outlined an $O(n^3)$ time algorithm 
for their recognition; this algorithm can also produce a Pfaffian orientation 
when one exists.

A general Pfaffian graph, as opposed to a bipartite one, can be very dense 
as observed by Norine~\cite{Norine2005}: there is an infinite family 
of $n$-vertex Pfaffian graphs with $\Omega(n^2)$ edges. This construction in 
particular poses obstacles to find characterizations as the ones mentioned 
above for bipartite graphs. Indeed, it is not known how to efficiently 
recognize a general Pfaffian graph.

The most famous Pfaffian graphs are the {\em planar} ones, graphs 
whose vertices can be embedded in the plane with straight lines connecting 
the vertices of every edge without any two lines crossing each other except 
at endpoints. That these graphs are Pfaffian was discovered 
by Kasteleyn~\cite{Kasteleyn1967}; furthermore, there is a linear-time 
algorithm that finds a Pfaffian orientation given a planar graph 
by Little~\cite{Little1974}.

\subsection{Conventions and organization}

\label{sect:conventions}

We assume knowledge of standard graph-theoretic terminology; 
see e.g.~West~\cite{West1996}. 
Graphs in this paper are undirected unless otherwise mentioned; this 
in particular also applies to subgraphs such as paths, cycles, and Hamiltonian
cycles. No graph or directed graph in this paper has loops or multiple edges.
For a graph or directed graph~$G$, we write $V(G)$ for the vertex 
set of $G$ and $E(G)$ for the edge set of $G$. We identify the {\em edges} of a 
graph with two-subsets $\{u,w\}$ where $u$ and $w$ are distinct vertices. 
We call the edges of a directed graph {\em arcs} in what follows, and 
identify each arc with a two-tuple $(u,w)$ where $u$ and $w$ are distinct 
vertices. We recall our conventions with orientations and Pfaffian graphs
from Section~\ref{sect:pfaffian}.

We work with Iverson's bracket notation---for a logical proposition~$P$, 
we define
\[
\iv{P}=\begin{cases}
1 & \text{if $P$ is true}\,;\\
0 & \text{if $P$ is false}\,.
\end{cases}
\]

The rest of this paper is organized as follows.
Section~\ref{sect:hc} presents our novel witnessing technique for 
Hamiltonian cycles in bipartite Pfaffian graphs. 
We prove our main theorems and their indirect structural corollaries in Section~\ref{sect:ahc}.
Our further results on finding and counting Hamiltonian cycles as well
as on TSP are proved in Section~\ref{sect:fhc}.
Finally, in Appendix~\ref{sect: NP-hard} we prove the NP-hardness of \textsc{Another Hamiltonian Cycle} in bipartite Pfaffian graphs without constrained vertex degrees.

\section{Hamiltonian cycles in bipartite Pfaffian graphs}
\label{sect:hc}

This section presents what we believe to be a novel technique to efficiently
witness Hamiltonian cycles in bipartite Pfaffian graphs via
(not necessarily proper) two-colorings of the vertices. 
We also show how to efficiently construct a Pfaffian orientation 
from a known Hamiltonian cycle in a bipartite Pfaffian graph, 
as well as show how to efficiently find a witness by extending 
a given partial witness defined on only one of the parts of a bipartition.

Throughout this section $G$ is an $n$-vertex bipartite Pfaffian graph and 
$\oa G$ is a fixed but otherwise arbitrary Pfaffian orientation of $G$.
Since we are interested in whether $G$ is Hamiltonian, without loss of 
generality we may assume that $n$ is even and $n\geq 4$ in what follows.

Select an arbitrary edge $e\in E(G)$ and call it the {\em anchor} edge. 

\subsection{Preliminaries: The structure of Pfaffian orientations}
\label{sect:structure-pfaffian}

We start by recalling the known structure of Pfaffian orientations of $G$.
Namely, de Carvalho, Lucchesi, and Murty~\cite{deCarvalho2005} observed that 
all Pfaffian orientations of $G$ are obtainable from each other 
by reversals of arcs across vertex cuts. More precisely, for 
any Pfaffian orientation $\oa G$ and any vertex $u\in V(G)$, 
it holds that reversing the arcs incident to $u$ in $\oa G$ results in another 
Pfaffian orientation; furthermore, every Pfaffian orientation of $G$ can be
obtained by starting with an arbitrary Pfaffian orientation of $G$ and 
repeating such operations for different vertices~\cite{deCarvalho2005}.

\subsection{The two-coloring defined by an anchored Hamiltonian cycle}
\label{sect:chiH}

We are interested in characterising 
each Hamiltonian cycle $H$ in $G$ that traverses the selected anchor 
edge $e$---we say that such an $H$ is {\em anchored}---using a function 
$\chi_H:V(G)\rightarrow \{0,1\}$ that is unique%
\footnote{Unique but not canonical; as we will see, the function $\chi_H$ 
will depend not only on the Hamiltonian cycle $H$ but also on the choice of 
our assumed fixed but arbitrary Pfaffian orientation $\oa G$ of $G$.}{}
to $H$ and from which we will (in the next subsection) see $H$ can be efficiently constructed.

Towards this end, let us study the Pfaffian orientation $\oa G$ at 
the anchor $e$. Let $(s,t)\in E(\oa G)$ be the arc in $\oa G$ whose 
underlying edge in $G$ is the anchor edge $e=\{s,t\}$. Construct from 
the Pfaffian orientation $\oa G$ a new orientation $\oa G_e$ of $G$ that 
is otherwise identical to $\oa G$ except that the arc $(s,t)$ has been 
replaced with the arc $(t,s)$. That is, by definition we have
$(t,s)\in E(\oa G_e)$. 

Now consider an arbitrary anchored Hamiltonian cycle $H$ in $G$.
Since $e\in E(H)$, there is a unique consistent 
orientation $\oa H$ of $H$ such that $(t,s)\in E(\oa H)$.
Let us write $v_0,v_1,\ldots,v_{n-1}$ for the vertices of~$G$ indexed 
in the directed $\oa H$-path order from $s$ to $t$; that is, 
\begin{equation}
\label{eq:v-seq}
v_0=s\,,\quad 
v_{n-1}=t\,,\quad\text{and}\quad 
(v_i,v_{(i+1)\bmod n})\in E(\oa H)\quad\text{for all $i=0,1,\ldots,n-1$}\,.
\end{equation}
Associate with $H$ the (not necessarily proper) vertex-coloring function 
$\chi_H:V(G)\rightarrow\{0,1\}$ defined by setting
\begin{equation}
\label{eq:chi-Hst}
\begin{split}
\chi_H(v_0)&=0\quad\text{and}\\
\chi_H(v_{i+1})&\equiv \chi_H(v_i)+\iv{(v_i,v_{i+1})\in E(\oa G_e)}\pmod 2
\quad\text{for all $i=0,1,\ldots,n-2$}\,.
\end{split}
\end{equation}
Because $\oa G$ is a Pfaffian orientation 
and $H$ is a central cycle of $G$, we have 
\begin{equation}
\label{eq:H-odd}
|E(\oa H)\cap E(\oa G)|
=
\sum_{i=0}^{n-1}\iv{(v_i,v_{(i+1)\bmod n})\in E(\oa G)}\equiv 1\pmod 2\,.
\end{equation}
Since $\oa G$ and $\oa G_e$ differ only in the orientation of $e\in E(H)$, 
from \eqref{eq:H-odd} we immediately have
\begin{equation}
\label{eq:H-even}
|E(\oa H)\cap E(\oa G_e)|
=
\sum_{i=0}^{n-1}\iv{(v_i,v_{(i+1)\bmod n})\in E(\oa G_e)}\equiv 0\pmod 2\,.
\end{equation}
We thus conclude
\begin{equation}
\label{eq:chi-Hts}
\begin{split}
\chi_H(v_0)
\stackrel{\eqref{eq:chi-Hst}}{=}
0
\stackrel{\eqref{eq:H-even}}{\equiv}
|E(\oa H)\cap E(\oa G_e)|
&=
\sum_{i=0}^{n-2}\iv{(v_i,v_{i+1})\in E(\oa G_e)}
+\iv{(v_{n-1},v_0)\in E(\oa G_e)}\\
&\stackrel{\eqref{eq:chi-Hst}}{\equiv}
\chi_H(v_{n-1})
+\iv{(v_{n-1},v_{0})\in E(\oa G_e)}\pmod 2\,.
\end{split}
\end{equation}
Since by definition of $\oa G_e$ we have
$(t,s)\in E(\oa G_e)$, from \eqref{eq:v-seq} and~\eqref{eq:chi-Hts} 
we conclude that $\chi_H(t)=1$.
Furthermore, from \eqref{eq:chi-Hst} and \eqref{eq:chi-Hts} we have 
for all $(u,w)\in E(\oa H)$ that
\begin{equation}
\label{eq:chi-H}
\chi_H(w)\equiv \chi_H(u)+\iv{(u,w)\in E(\oa G_e)}\pmod 2\,.
\end{equation}
That is, each arc $(u,w)\in E(\oa H)$ is $\chi_H$-monochromatic (i.e. both endpoints are assigned the same value by 
$\chi_H$) if and only if $(w,u)\in E(\oa G_e)$. 

\subsection{The orientation induced by a good coloring}
\label{sect:induced-orientation}

Suppose now that we do not know anything about the (anchored) 
Hamiltonian cycles of $G$, if any, and have access only to the Pfaffian
orientation~$\oa G$ and the orientation $\oa G_e$; the latter is
easily obtainable from $\oa G$, cf.~Section~\ref{sect:chiH}.

Consider an arbitrary vertex coloring 
$\chi:V(G)\rightarrow\{0,1\}$ with $\chi(s)=0$ and $\chi(t)=1$.
The last observation in Section~\ref{sect:chiH} 
suggests that we should explore reversing exactly 
the $\chi$-monochromatic arcs in $\oa G_e$. Let us make this formal as follows.
Let the orientation $\oa G_e^\chi$ {\em induced} by the 
coloring $\chi$ be the unique orientation of $G$ that for each
edge $\{u,w\}\in E(G)$ satisfies
\begin{equation}
\label{eq:oa-induced}
(u,w)\in E(\oa G_e^\chi)
\quad\text{if and only if}\quad
\chi(w)\equiv\chi(u)+\iv{(u,w)\in E(\oa G_e)}\pmod 2\,.
\end{equation}

To witness the serendipity of \eqref{eq:oa-induced}, suppose that $G$ admits 
an anchored Hamiltonian cycle $H$; it follows immediately 
from \eqref{eq:chi-H} and~\eqref{eq:oa-induced} 
that $E(\oa H)\subseteq E(\oa G_e^{\chi_H})$. Thus, 
if we know only the coloring $\chi_H$ but not $H$, we can search for 
$E(\oa H)$ in $E(\oa G_e^{\chi_H})$; let us next analyse this situation
in more detail from the standpoint of our arbitrary $\chi$. 

Call the coloring $\chi$ {\em good} if there exists an
anchored Hamiltonian cycle $H$ in $G$ with $\chi=\chi_H$; otherwise 
call $\chi$ {\em bad}. The following lemma shows that 
the orientation $\oa G_e^\chi$ for a good $\chi$ enables {\em linear-time} 
and unique algorithmic recovery of $H$ by standard longest-path search
in a directed acyclic graph (DAG); in fact, mere topological sorting suffices,
as is apparent from the proof. See also Figure~\ref{fig:good-colorings} for
an illustration of the concepts involved in a Hamiltonian planar graph.

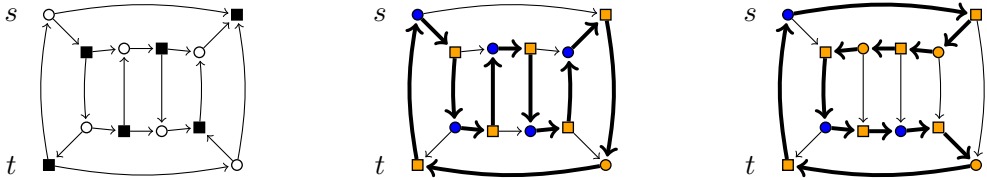
\begin{figure}[!ht]
\[
\begin{tikzpicture}[scale=0.5,shorten >=1pt, auto, node distance=1cm, ultra thick]
\node[rectangle](R) at (-4,2) {$s$};\node[rectangle](R) at (-4,-2) {$t$};\pfaffian
\end{tikzpicture}
\qquad \qquad
\begin{tikzpicture}[scale=0.5,shorten >=1pt, auto, node distance=1cm, ultra thick]
\node[rectangle](R) at (-4,2) {$s$};\node[rectangle](R) at (-4,-2) {$t$};\coloredone
\end{tikzpicture}
\qquad \qquad
\begin{tikzpicture}[scale=0.5,shorten >=1pt, auto, node distance=1cm, ultra thick]
\node[rectangle](R) at (-4,2) {$s$};\node[rectangle](R) at (-4,-2) {$t$};\coloredtwo
\end{tikzpicture}
\]
 \caption{Illustration of orientations induced by good colorings. Left: an undirected bipartite planar graph $G$ drawn in one of its orientations $\oa G_e$ with $e=\{s,t\}$, one arc reversal away from a Pfaffian orientation $\oa G$. Middle and right: two vertex colorings $\chi_H$ and coloring-induced orientations $\oa G_e^{\chi_H}$ for two different Hamiltonian cycles $H$, with the arcs of $\oa H$ drawn in bold in each case. Observe that every monochromatic arc reverses its orientation with respect to $\oa G_e$, whereas bichromatic arcs keep their orientation. Observe also that the removal of the arc $(t,s)$ from $\oa G_e^{\chi_H}$ leaves an acyclic Hamiltonian directed graph, whereby the directed Hamiltonian path and 
hence $H$ can be found, for example, by topological sorting; cf.~Lemma~\ref{lem:acyclic-hamiltonicity}.} 
  \label{fig:good-colorings}
\end{figure}

\begin{lemma}[Acyclic Hamiltonicity of good-coloring-induced orientations]
\label{lem:acyclic-hamiltonicity}
Let $\chi$ be good. Then, the orientation 
$\oa G_e^\chi\setminus (t,s)$ 
of $G\setminus e$ is acyclic with the unique source vertex $s$ and 
the unique sink vertex~$t$. Moreover, the longest directed path 
in $\oa G_e^\chi\setminus (t,s)$ is unique and a directed Hamiltonian path.
\end{lemma}

\begin{proof}
Since $\chi$ is good, there exists an anchored Hamiltonian cycle $H$
with $\chi=\chi_H$. Furthermore, we can follow the notational conventions
in Section~\ref{sect:chiH} with respect to this $H$, including the 
vertex-indexing $v_0,v_1,\ldots,v_{n-1}$ for $G$ and \eqref{eq:v-seq} 
in particular. From $E(\oa H)\subseteq E(\oa G_e^\chi)$ we thus
conclude that the sequence $v_0,v_1,\ldots,v_{n-1}$ defines a longest
directed path (which is also a directed Hamiltonian path from $s$ to $t$) 
in the directed graph 
$\oa G_e^\chi\setminus (t,s)$. It follows immediately that $s$ is the 
only possible source vertex and $t$ is the only possible sink vertex
in $\oa G_e^\chi\setminus (t,s)$. 

Let us next show that $\oa G_e^\chi\setminus (t,s)$ is acyclic as a directed
graph. To reach a contradiction, suppose that $\oa D$ is a directed cycle in 
$\oa G_e^\chi\setminus (t,s)$. Since $(t,s)=(v_{n-1},v_0)\notin E(\oa D)$
and $\oa D$ is a directed cycle with 
$V(\oa D)\subseteq\{v_0,v_1,\ldots,v_{n-1}\}$, there must 
exist $0\leq i<j\leq n-1$ with at least two proper inequalities among the
three such that $(v_j,v_i)\in E(\oa D)$ and thus $(v_j,v_i)\in E(\oa G_e^\chi)$.
Let $\oa C$ be the directed cycle 
with $V(\oa C)=\{v_i,v_{i+1},\ldots,v_j\}$ and
$E(\oa C)=\{(v_i,v_{i+1}),(v_{i+1},v_{i+2}),\ldots,(v_{j-1},v_j),(v_j,v_i)\}$.
In particular, $\oa C\neq \oa H$ since $(t,s)\notin E(\oa C)$. 
Let $C$ be the underlying undirected cycle of $\oa C$, and observe that 
$C$ is a cycle of $G$. Since $G$ is bipartite, $C$ is even and has at least
four vertices. Thus, the edges of $H\setminus V(C)$ contain
a perfect matching of $G\setminus V(C)$, implying that $C$ is central.
Since $C$ avoids $e$ and $\oa C$ is a consistent orientation of $C$, 
we conclude by Pfaffianity that 
\[
|E(\oa C)\cap E(\oa G_e)|=|E(\oa C)\cap E(\oa G)|\equiv 1\pmod 2\,.
\]
But this is a contradiction since for all $(u,w)\in E(\oa C)$ we have
$(u,w)\in E(\oa G_e^\chi)$, and thus by \eqref{eq:oa-induced} it holds
that $\chi(w)\equiv\chi(u)+\iv{(u,w)\in E(\oa G_e)}\pmod 2$; take the
sum of these congruences over all arcs $(u,w)\in E(\oa C)$ to conclude that
$|E(\oa C)\cap E(\oa G_e)|\equiv 0\pmod 2$, a contradiction. Thus,
$\oa G_e^\chi\setminus (t,s)$ is acyclic as a directed graph.

From acyclicity it also immediately follows that $s$ is a source vertex 
and $t$ is a sink vertex of $\oa G_e^\chi\setminus (t,s)$; indeed, 
any arc into $s$ or any arc out of $t$ would complete a directed cycle 
together with an appropriate proper segment of the directed Hamiltonian 
path $s=v_0,v_1,\ldots,v_{n-1}=t$. 
This longest path (of $n$ vertices) is also seen to be unique in 
$\oa G_e^\chi\setminus (t,s)$; indeed, the existence of any other such path
would again imply an arc that would complete a directed cycle together
with an appropriate proper segment of $v_0,v_1,\ldots,v_{n-1}$.
\end{proof} 

An immediate corollary of the proof Lemma~\ref{lem:acyclic-hamiltonicity} is that there
are at most $2^{n-2}$ Hamiltonian cycles through any fixed edge in 
an $n$-vertex bipartite Pfaffian graph. For comparison, there are planar 
graphs with at least $2.08^n$ Hamiltonian cycles, see~\cite{Buchin2007}.

\subsection{Constructing a Pfaffian orientation from a Hamiltonian cycle}

Next we address the task of constructing a Pfaffian orientation if we know
one Hamiltonian cycle, with the intent of constructing possible further
Hamiltonian cycles with the help of the Pfaffian orientation obtained. 

\begin{lemma}[Constructing a Pfaffian orientation from a Hamiltonian cycle]
\label{lem:pfaffian-orientation-from-hamiltonian-cycle}
There exists a linear-time algorithm that, given as input a bipartite
Pfaffian graph $G$ and a Hamiltonian cycle $H$ in $G$, outputs a 
Pfaffian orientation $\oa G$ of $G$.
\end{lemma}

\begin{proof}
Let $G$ and $H$ be given as input, and select an arbitrary edge $e\in E(H)$ 
as an anchor edge in the sense of Section~\ref{sect:chiH}. 
Let $\oa H'$ be an arbitrary consistent orientation of $H$;
let $s,t\in V(G)$ such that $(t,s)\in E(\oa H')$ and $e=\{s,t\}$. 

Since $G$ is Pfaffian, it has a Pfaffian orientation. 
Let $\oa G$ be an arbitrary Pfaffian orientation of $G$.
Without loss of generality we can assume that $(s,t)\in E(\oa G)$; 
indeed, reverse all arcs if $(s,t)\notin E(\oa G)$.
Let $\oa G_e$, $\oa H$, and $\chi_H$ be constructed from $\oa G$, $e$, and $H$
as in Section~\ref{sect:chiH}. In particular, we have $\oa H'=\oa H$.

Now observe that $\oa A=\oa G_e^{\chi_H}\setminus(t,s)$ is acyclic 
and Hamiltonian by Lemma~\ref{lem:acyclic-hamiltonicity}. 
Furthermore, since $\oa A$ is acyclic and has the directed Hamiltonian path
$\oa H'\setminus(t,s)$, we can in linear time construct $\oa A$ from $G$, 
$\oa H'$, and $e$ by orienting $G\setminus e$ so that 
(i) the edges of $H\setminus e$ are oriented as in $\oa H'\setminus(t,s)$, 
and (ii) all the other edges of $G\setminus e$ are oriented into arcs in 
directed $\oa H'\setminus(t,s)$-order; indeed, otherwise a directed 
cycle would result.

Next, let $\chi:V(G)\rightarrow\{0,1\}$ be the proper two-coloring of the
vertices of $G$ with $\chi(s)=0$ and $\chi(t)=1$; such a $\chi$ exists and 
is unique because $G$ is bipartite Hamiltonian. Since no edge of $G$ 
is $\chi$-monochromatic, we have $(\oa G')_e=(\oa G')_e^\chi$ for all Pfaffian
orientations $\oa G'$ of $G$. We will construct a sequence  
$\chi_0,\chi_1,\ldots,\chi_n:V(G)\rightarrow\{0,1\}$ of colorings and a
sequence $\oa G_0,\oa G_1,\ldots,\oa G_n$ of Pfaffian orientations of $G$ 
such that 
\begin{equation}
\label{eq:a-chain}
\oa A\cup(t,s)
=
\oa G_e^{\chi_H}
=
(\oa G_0)_e^{\chi_0}
=
(\oa G_1)_e^{\chi_1}
=
\cdots
=
(\oa G_n)_e^{\chi_n}
=
(\oa G_n)_e^{\chi}
=
(\oa G_n)_e\,.
\end{equation}
This concludes that $\oa A\cup (s,t)$ is a Pfaffian orientation of $G$.
Moreover, $\oa A\cup (s,t)$ is constructible in linear time from 
the given input $G$ and $H$.

It remains to construct the sequences $\chi_0,\chi_1,\ldots,\chi_n$ and
$\oa G_0,\oa G_1,\ldots,\oa G_n$ as well as conclude \eqref{eq:a-chain}.
The first two identities in \eqref{eq:a-chain} are immediate when 
we set $\chi_0=\chi_H$ and $\oa G_0=\oa G$. Let $v_1,v_2,\ldots,v_n$ be
an arbitrary enumeration of the $n$ vertices of $G$. For $k=1,2,\ldots,n$, 
define $\chi_k:V(G)\rightarrow\{0,1\}$ for all $j=1,2,\ldots,n$ by the
rule 
\[
\chi_k(v_j)=\begin{cases}
\chi(v_j)       & \text{if $j\leq k$}\,;\\
\chi_{k-1}(v_j) & \text{if $j>k$}.
\end{cases}
\]
It is immediate that $\chi_n=\chi$, which establishes the last two 
identities in \eqref{eq:a-chain}. Furthermore, $\chi_{k-1}$ and $\chi_k$ 
are identical expect possibly at $v_k$. Also observe that 
$\chi_k(s)=\chi_H(s)=\chi(s)=0$ and $\chi_k(t)=\chi_H(t)=\chi(t)=1$
for all $k=0,1,\ldots,n$. 
To define the sequence $\oa G_k$ for $k=1,2,\ldots,n$, 
split into cases as follows: 
when $\chi_k(v_k)=\chi_{k-1}(v_k)$, set
$\oa G_k=\oa G_{k-1}$;
when $\chi_k(v_k)\neq \chi_{k-1}(v_k)$,
set $\oa G_k$ to be otherwise identical to $\oa G_{k-1}$ except reverse all
arcs incident to $v_k$. In both of these cases we observe 
by \eqref{eq:oa-induced} that 
$(\oa G_{k-1})_e^{\chi_{k-1}}=(\oa G_k)_e^{\chi_k}$, which establishes
all the remaining equalities in \eqref{eq:a-chain}. 
Furthermore, $\oa G_k$ is a Pfaffian orientation
since $\oa G_{k-1}$ is a Pfaffian orientation; 
indeed, recall Section~\ref{sect:structure-pfaffian} and that $\oa G_0=\oa G$ 
is a Pfaffian orientation by assumption.
\end{proof}

\subsection{Finding a good coloring}
\label{sect:finding-a-good-coloring}

Let us now study the task of finding a good coloring $\chi$ 
given the bipartite Pfaffian graph $G$, 
the Pfaffian orientation $\oa G$, and the anchor edge $e$ as input;
also recall the conventions and further notation---in particular, 
the vertices $s$ and $t$---from Sections~\ref{sect:chiH} 
and~\ref{sect:induced-orientation}.
In this setting, a natural question to ask 
is how much one needs to reveal from a good coloring $\chi$ to enable 
efficient completion to a good coloring. We now show that it suffices to 
reveal $\chi$ in one of the parts of the bipartition of $G$ by a reduction
to bipartite perfect matching in an auxiliary bipartite graph.

More precisely, let the sets $L$ (``left'') and $R$ (``right'') form 
a partition of the vertices of $G$ such that $s\in L$, $t\in R$, and 
every edge of $G$ has one end in $L$ and the other end in $R$.%
\footnote{This bipartition $(L,R)$ of $G$ is in fact unique unless $G$
is not Hamiltonian. Moreover, $(L,R)$ is computable in linear time from the
given input.}{}
Let $\lambda:L\rightarrow\{0,1\}$ with $\lambda(s)=0$ be a given further input. 
Our task is to find whether there exists a good coloring 
$\chi:V(G)\rightarrow\{0,1\}$ with $\chi(\ell)=\lambda(\ell)$ 
for all $\ell\in L\subseteq V(G)$; 
that is, whether there exists a good coloring that
extends the partial coloring $\lambda$. 

Construct an auxiliary bipartite graph $F_\lambda$ as follows.
Let the vertex set $V(F_\lambda)=V(G)\times\{0,1\}$. 
To avoid notational confusion between arcs and vertices of $F_\lambda$, 
we will use bracketed notation $[u,k]$ for vertices 
of $F_\lambda$ with $u\in V(G)$ and $k\in\{0,1\}$. 
The edge set $E(F_\lambda)$ is defined by the following rule.
For all $\ell\in L$, $r\in R$, $p\in\{0,1\}$, and $\rho\in\{0,1\}$ 
with $\{\ell,r\}\in E(G)$, we have 
\begin{equation}
\label{eq:f-edge}
\{[\ell,p],[r,\rho]\}\in E(F_\lambda)
\end{equation}
if and only if both
\begin{equation}
\label{eq:f-parity}
\rho\equiv\lambda(\ell)+\iv{(\ell,r)\in \oa G_e}\pmod 2
\end{equation}
and
\begin{equation}
\label{eq:f-s}
\ell\neq s
\quad\text{or}\quad 
p\neq 0
\quad\text{or}\quad 
r=t
\,.
\end{equation}
For an edge $\{[\ell,p],[r,\rho]\}\in E(F_\lambda)$, we say that 
the edge $\{\ell,r\}\in E(G)$ is the \emph{projection} of the edge (to $G$) and 
call $p$ the {\em port} at $\ell$ and $\rho$ the {\em parity} at $r$, 
stressing that port and parity have asymmetric roles in our construction 
even though both range in $\{0,1\}$.

Let us now start analysing the structure of $F_\lambda$ in more detail.
First, the parts $L\times\{0,1\}$ and $R\times\{0,1\}$ witness 
by \eqref{eq:f-edge} that $F_\lambda$ is bipartite. In particular, 
$F_\lambda$ has $2n$ vertices with $|L|=|R|=n/2$, where we recall 
that $n\geq 4$ is the number of vertices in $G$ with $n$ even.
Second, recalling that $s\in L$ and $t\in R$, the constraint \eqref{eq:f-s} 
effectively states that $[s,0]$ is adjacent only to $[t,0]$ in $F_\lambda$; 
indeed, recalling that $\lambda(s)=0$ and $(s,t)\notin\oa G_e$, 
from~\eqref{eq:f-parity} we have that $[s,0]$ is not adjacent to $[t,1]$.
Third, for all $\ell\in L\setminus\{s\}$, we observe from~\eqref{eq:f-parity}
and \eqref{eq:f-s}, the latter being trivially true, that the vertices
$[\ell,0]$ and $[\ell,1]$ have identical vertex neighborhoods in $F_\lambda$. 

We are now ready for our first key lemma. Recall the coloring $\chi_H$ 
associated to an anchored Hamiltonian cycle $H$ of $G$ 
from Section~\ref{sect:chiH}. The first lemma shows that every perfect 
matching in $F_\lambda$ gives rise to an anchored Hamiltonian cycle; 
different perfect matchings may give rise to the same 
anchored Hamiltonian cycle however. 
This structure is illustrated in Figure~\ref{fig:f-lambda}.

\begin{figure}[!ht]
\[
\begin{tikzpicture}[scale=0.5,shorten >=1pt, auto, node distance=1cm, ultra thick]  
\node[rectangle](R) at (-4,2) {$s$};\node[rectangle](R) at (-4,-2) {$t$};\pfaffian
\end{tikzpicture}
\qquad \qquad \qquad
\begin{tikzpicture}[scale=0.5,shorten >=1pt, auto, node distance=1cm, ultra thick]
\node[rectangle](R) at (-4,2) {$s$};\node[rectangle](R) at (-4,-2) {$t$};\matching
\end{tikzpicture}
\]
 \caption{Illustration of a perfect matching in the graph $F_\lambda$.
Left: The graph~$G$ drawn in one of its orientations $\oa G_e$ with $e=\{s,t\}$
and the bipartition $(L,R)$ with vertices in $L$ drawn as white circles and
the vertices in $R$ drawn as black boxes.
Right: The graph $F_\lambda$ and the coloring $\lambda$ drawn as an oriented 
overlay of $G$. 
Observe that each vertex $r\in R$ has two copies $[r,\rho]$ in $F_\lambda$, 
one for each parity $\rho\in\{0,1\}$, with a blue square indicating 
parity $0$ and an orange square indicating parity $1$. Although each vertex
$\ell\in L$ has two copies $[\ell,p]$ in $F_\lambda$, one copy for each
port $p\in\{0,1\}$, we contract these two copies into one vertex (circle) 
in the drawing, and display for each vertex its color $\lambda(\ell)\in\{0,1\}$
(blue or orange) instead. Each arc in the drawing is oriented from $L$ 
to $R$ and represents two edges of $F_\lambda$ with opposite ports.
A perfect matching $M$ in $F_\lambda$ is represented by the bold arcs. 
In particular, observe that each circle is incident to two bold arcs; 
these two bold arcs represent two edges in $M$ with opposite ports. 
These opposite ports are otherwise arbitrary except for the edge of $M$ 
that projects to $\{s,t\}$, which must have port~$0$. Also observe that
from the drawn $M$ it is visually intuitive how to obtain a Hamiltonian
cycle in $G$ corresponding to $M$ by turning the bold arcs into the edges
of a Hamiltonian cycle in $G$; this intuition is made rigorous 
in Lemma~\ref{lem:matching-to-extension} by the Hamiltonian cycle $\ham M$
of~$G$ obtained from $M$.}
\label{fig:f-lambda}
\end{figure}
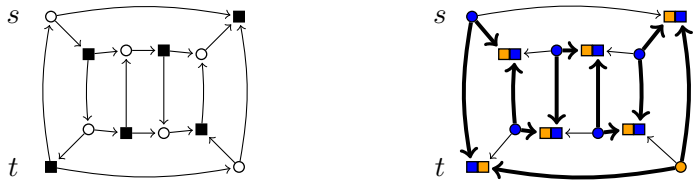

\begin{lemma}[Perfect matchings witness good extensions]
\label{lem:matching-to-extension}
For every perfect matching $M$ in $F_\lambda$, there exists a
Hamiltonian cycle $\ham M$ in $G$ with $\chi_{\ham M}(\ell)=\lambda(\ell)$ 
for all $\ell\in L$.
\end{lemma}

\begin{proof}
Let $M$ be an arbitrary perfect matching in $F_\lambda$. 
For $[\ell,p]\in L\times\{0,1\}$ and 
$[r,\rho]\in L\times\{0,1\}$, let us use functional notation
$M([\ell,p])=[r,\rho]$ or
$M([r,\rho])=[\ell,p]$ to signal that the vertices 
$[\ell,p]$ and $[r,\rho]$ are matched by $M$ in $F_\lambda$.
We construct the anchored Hamiltonian cycle $\ham M$ 
as well as the coloring $\chi=\chi_{\ham M}$
with $\chi(\ell)=\lambda(\ell)$ for all $\ell\in L$
by traversing all the vertices of $F_\lambda$ in an order determined 
by $M$ to yield the Hamiltonian cycle $\ham M$. In particular, we 
will define $\ham M$ in steps by introducing, one vertex at a time, 
a vertex order $v_0,v_1,\ldots,v_{n-1}$ for the vertices of $G$
with $v_0=s$, $v_{n-1}=t$, and 
$(v_i,v_{(i+1)\bmod n})\in E(\overrightarrow{\ham M})\subseteq E(\oa G_e^{\chi})$ for all $i=0,1,\ldots,n-1$. 

Our traversal starts from the vertex $[\ell_0,p_0]$ of $F_\lambda$ defined by 
$\ell_0=s$ and $p_0=1$. The traversal then follows edges in the perfect
matching $M$, changing parity (at vertices in $R\times\{0,1\}$) and port 
(at vertices in $L\times\{0,1\}$) to arrive at subsequent edges; 
for these changes, for $z\in\{0,1\}$ it is convenient to 
write $\overline{z}=(z+1)\bmod 2$ for notational brevity; that is, 
$\overline 0=1$ and $\overline 1=0$. 

In precise terms, the traversal is as follows. 
Assuming we have defined $\ell_j\in L$ and $p_j\in\{0,1\}$
for all $j\in \{0,1,\ldots,i\}$ with $i\geq 0$, we proceed to define 
$\ell_{i+1}\in L$ and $p_{i+1}\in\{0,1\}$ as follows.
Set $v_{2i}=\ell_i$ and $\chi(\ell_i)=\lambda(\ell_i)$.
Define $r_i\in R$ and $\rho_i\in\{0,1\}$ by $M([\ell_i,p_i])=[r_i,\rho_i]$.
Set $v_{2i+1}=r_i$ and $\chi(r_i)=\rho_i$.
Define $\ell_i'\in L$ and $p_i'\in\{0,1\}$ by 
$M([r_i,\overline{\rho_i}])=[\ell_i',p_i']$.
Set $\ell_{i+1}=\ell_i'$ and $p_{i+1}=\overline{p_i'}$
as well as $\chi(\ell_{i+1})=\lambda(\ell_{i+1})$ and $v_{2(i+1)}=\ell_{i+1}$.
We continue this process for $i=0,1,\ldots$ and claim that eventually 
$\ell_{i+1}=\ell_0$ and $p_{i+1}=p_0$ with $i+1=n/2$,
at which point $\{v_0,v_1,\ldots,v_{n-1}\}=V(G)$ and $\ham M$ is 
a Hamiltonian cycle in $G$ 
with $\chi_{\ham M}=\chi$ with $\chi(\ell)=\lambda(\ell)$ for all $\ell\in L$.

Let us now analyse the traversal process in more detail. First, we observe that
$\ell_{i+1}\neq\ell_i$; indeed, suppose that $\ell_{i+1}=\ell_i$ and 
observe that the traversal step from $\ell_i$ to $\ell_{i+1}$ changes parity 
from $\rho_i$ to $\bar\rho_i$ at $r_i$, yet from~\eqref{eq:f-parity} we 
observe that every edge of $F_\lambda$ that projects to the 
edge $\{\ell_i,r_i\}=\{\ell_{i+1},r_i\}$ has the same parity at $r_i$, 
a contradiction. Next, let us observe that 
$(v_{2i},v_{2i+1})=(\ell_i,r_i)\in E(\oa G_e^\chi)$. 
Indeed, the identity is immediate, and membership holds 
by~\eqref{eq:oa-induced} and
\[
\chi(r_i)
=
\rho_i
\stackrel{\eqref{eq:f-parity}}{\equiv} 
\lambda(\ell_i)+\iv{(\ell_i,r_i)\in\oa G_e}
=\chi(\ell_i)+\iv{(\ell_i,r_i)\in\oa G_e}
\pmod 2\,.
\]
Let us then observe that 
$(v_{2i+1},v_{2i+2})=(r_i,\ell_{i+1})\in E(\oa G_e^\chi)$.
Again the identity is immediate, 
and membership holds 
by~\eqref{eq:oa-induced}, the fact that $\oa G_e^\chi$ orients 
$\{\ell_{i+1},r_i\}\in E(G)$ in one of two possible orientations, and
\[
\chi(r_i)=
\rho_i
\neq
\bar\rho_i
\stackrel{\eqref{eq:f-parity}}{\equiv} 
\lambda(\ell_{i+1})+\iv{(\ell_{i+1},r_i)\in\oa G_e}
=\chi(\ell_{i+1})+\iv{(\ell_{i+1},r_i)\in\oa G_e}
\pmod 2\,.
\]

Next let us show that all the vertices $\ell_i$ and $r_i$ traversed by the
process are distinct, until $\ell_{i+1}=\ell_0$ for some $i\geq 1$, noting
that the case $i=0$ has already been excluded earlier.  
Suppose $\ell_1,\ell_2,\ldots,\ell_i$ are distinct; since $M$ contains exactly
two edges (of opposite parities) that project to edges incident to any 
fixed $r\in R$, we observe that these two edges of $M$ have been each 
traversed once by the process for each $r_0,r_1,\ldots,r_i$ since 
$\ell_0,\ell_1,\ldots,\ell_i$ are distinct, implying that 
$r_0,r_1,\ldots,r_i$ are distinct, and thus that $v_0,v_1,\ldots,v_{2i+1}$
are distinct. 
So suppose that $\ell_{i+1}=\ell_j$ for
some $0\leq j\leq i$; also note that this must happen for some $i<|L|=n/2$.
If $j\geq 1$, we have a contradiction since $M$ contains exactly
two edges (of opposite ports) that project to edges incident to any
fixed $\ell\in L$, and for $\ell=\ell_j$ these two edges 
(projecting to $\{\ell_j,r_{j-1}\}$ and $\{\ell_j,r_j\}$) have already 
been traversed; so there is no edge in $M$ that projects to 
$\{\ell_j,r_i\}=\{\ell_{i+1},r_i\}$, a contradiction. So we must have $j=0$.
This implies in particular that 
$(v_k,v_{(k+1)\bmod (2i+2)})\in E(\oa G_e^\chi)$ for all $k=0,1,\ldots,2i+1$.
Furthermore, since the edge of $M$ that is incident to $[\ell_0,p_0]=[s,1]$ 
has already been traversed, we have that the edge 
$\{[\ell_{i+1},\overline{p_{i+1}}],[r_i,\bar\rho_i]\}$ in $M$ must 
be the edge $\{[s,0],[t,0]\}$ (recall our analysis earlier that $[s,0]$ is 
adjacent only to $[t,0]$ in $F_\lambda$); thus we conclude that 
$\chi(t)=\chi(r_i)=\rho_i\neq\bar\rho_i=0$; that is, $\chi(t)=1$. 

Let us next show that $i=n/2-1$. So to reach a contradiction, suppose 
that $i<n/2-1$. In particular, the edges of $G$ underlying the arcs 
$(v_k,v_{(k+1)\bmod (2i+2)})\in E(\oa G_e^\chi)$ for $k=0,1,\ldots,2i+1$
trace a cycle of even length $2i+2<n$ in $G$. This leaves some of the
vertices in $G$, and thus all corresponding vertices of $F_\lambda$ 
regardless of port/parity, unvisited by the traversal process. By starting
the traversal process again from an arbitrary unvisited vertex in 
$L\times\{0,1\}$, we end up tracing a further even-length cycle in $G$, 
and repeating the process until all vertices of $G$ are visited, 
we obtain a vertex-disjoint union of even-length 
cycles that together cover the vertices of $G$, as well as a coloring $\chi$
such that all the cycles (in their consistently oriented form as they were 
traversed) occur as directed subgraphs of $\oa G_e^\chi$. Since $2i+2<n$, 
this cycle cover thus contains a cycle $C$ that does not contain the anchor 
edge $e$ and whose consistent orientation $\oa C$ is a subgraph of 
$\oa G_e^\chi$; observing that $C$ is central---indeed, use every other
edge from each even cycle other than $C$ in the cover to witness a perfect
matching in $G\setminus V(C)$---this leads to a contradiction via 
Pfaffianity by the same argument as was used in the proof of 
Lemma~\ref{lem:acyclic-hamiltonicity}; thus, $i=n/2-1$.

Since $i=n/2-1$, it follows from $(v_k,v_{(k+1)\bmod n})\in E(\oa G_e^{\chi})$ 
for all $k=0,1,\ldots,n-1$ and from 
Lemma~\ref{lem:acyclic-hamiltonicity} we conclude that 
$\chi=\chi_{\ham M}$ for the anchored Hamiltonian cycle $\ham M$ in $G$ 
defined by $V(\ham M)=\{v_0,v_1,\ldots,v_{n-1}\}$ and 
$E(\ham M)=\{\{v_k,v_{(k+1)\bmod n}\}:k=0,1,\ldots,n-1\}$.
\end{proof}

Conversely, we show that every good extension of $\lambda$ is witnessed 
by a perfect matching in $F_\lambda$. 

\begin{lemma}[Good extensions witness perfect matchings]
\label{lem:extension-to-matching}
For every anchored Hamiltonian cycle $H$ in $G$ 
with $\chi_H(\ell)=\lambda(\ell)$ for all $\ell\in L$, 
there exists a perfect matching $M_H$ in $F_\lambda$ with 
$\ham{M_H}=H$.
\end{lemma}

\begin{proof}
Let $H$ be an anchored Hamiltonian cycle in $G$ with 
$\chi_H(\ell)=\lambda(\ell)$ for all $\ell\in L$. 
Let $\oa H$ be the consistent orientation of $H$ with 
distinct vertices $v_0,v_1,\ldots,v_{n-1}$ so that $v_0=s$, $v_{n-1}=t$, and
$(v_k,v_{(k+1)\bmod n})\in E(\oa H)\subseteq E(\oa G_e^{\chi_H})$ 
for all $k=0,1,\ldots,n-1$.
Set $\ell_i=v_{2i}$, $r_i=v_{2i+1}$, and 
$\rho_i=(\chi_H(\ell_i)+\iv{(\ell_i,r_i)\in \oa G_e})\bmod 2$
for all $i=0,1,\ldots,n/2-1$.
Observe also that by \eqref{eq:oa-induced} and 
$(\ell_i,r_i)\in E(\oa G_e^{\chi_H})$, we have $\chi_H(r_i)=\rho_i$
for all $i=0,1,\ldots,n/2-1$.

Start with an empty matching $M_H$.
It is immediate from \eqref{eq:f-parity} and \eqref{eq:f-s}
that $\{[\ell_i,1],[r_i,\rho_i]\}\in E(F_\lambda)$ for $i=0,1,\ldots,n/2-1$.
Take each of these $n/2$ vertex-disjoint edges into $M_H$. 
Next observe that 
we have $\{[\ell_{(i+1)\bmod (n/2)},0],[r_i,\bar\rho_i]\}\in E(F_\lambda)$ 
for $i=0,1,\ldots,n/2-1$; indeed, from 
$(r_i,\ell_{(i+1)\bmod (n/2)})\in E(\oa G_e^{\chi_H})$ we conclude 
by~\eqref{eq:oa-induced} that 
\[
\rho_i=\chi_H(r_i)\equiv\chi_H(\ell_{(i+1)\bmod (n/2)})+\iv{(r_i,\ell_{(i+1)\bmod (n/2)})\in E(\oa G_e)}\pmod 2\,; 
\]
that is,
\[
\bar\rho_i\equiv\chi_H(\ell_{(i+1)\bmod (n/2)})+\iv{(\ell_{(i+1)\bmod (n/2)},r_i)\in E(\oa G_e)}\pmod 2\,, 
\]
so \eqref{eq:f-parity} holds. Furthermore, \eqref{eq:f-s} holds expect 
possibly when $\ell_{(i+1)\bmod (n/2)}=s$; but then $i=n/2-1$ and 
thus $r_i=t$, so \eqref{eq:f-s} holds also in this case.
Take each of these $n/2$ vertex-disjoint edges into $M_H$ and observe 
that the $n$ edges now in $M_H$ constitute a perfect matching in $F_\lambda$. 

By tracing the traversal process in the proof of 
Lemma~\ref{lem:matching-to-extension} with the definition of the perfect
matching $M_H$ above, we conclude that $\ham{M_H}=H$.
\end{proof}

Thus, $F_\lambda$ has a perfect matching if and only if $\lambda$ has a 
good extension. Moreover, from the proofs of 
Lemma~\ref{lem:matching-to-extension} and Lemma~\ref{lem:extension-to-matching}
we observe that the transformations $M\mapsto\ham M$ and 
$H\mapsto M_H$ are computable in linear time. We also observe that for
every anchored Hamiltonian cycle $H$ in $G$ there
are exactly $2^{n/2-1}$ perfect matchings $M$ in $F_\lambda$ with 
$\ham M=H$; these $M$ are all obtainable from each other by 
transposing ports at zero or more vertices $\ell\in L\setminus\{s\}$. 

\section{Another Hamiltonian cycle in bipartite Pfaffian graphs}

\label{sect:ahc}

This section studies the problem of finding another Hamiltonian cycle 
when given as input (i)~a~bipartite Pfaffian graph $G$ and 
(ii)~a~Hamiltonian cycle $H$ in $G$. Recall from 
Lemma~\ref{lem:pfaffian-orientation-from-hamiltonian-cycle}
that we can in linear time construct a Pfaffian orientation $\oa G$ from 
this input. In what follows we thus tacitly assume that such a $\oa G$ is 
available and fixed together with an arbitrary anchor edge $e\in E(H)$.

\subsection{Linear-time solvability in minimum degree three}
\label{sect:linear-time-solvability}

Our first objective in this section is our main theorem, which we
restate below for convenience. 

\thmahc*

We now proceed to prove Theorem~\ref{thm:ahc}. 
Recall and assume the setting of Section~\ref{sect:finding-a-good-coloring}. 
Observe that from the 
given input, we can in linear total time 
(a)~find a Pfaffian orientation $\oa G$ using $H$,
(b)~compute the orientation $\oa G_e$, 
(c)~compute the coloring $\chi_H$, 
(d)~compute the vertex bipartition $(L,R)$ of~$G$, 
(e)~restrict $\chi_H$ to $L$ to obtain the coloring $\lambda$, 
(f)~construct the graph $F_\lambda$, 
as well as
(g) construct the perfect matching $M_H$ in $F_\lambda$.

Using $M_H$ and $F_\lambda$, introduce the directed graph
$D_{\lambda,H}$ with the vertex set $V(D_{\lambda,H})=L$ and the arc set 
defined for all distinct $\ell,\ell'\in L$ by the rule 
$(\ell,\ell')\in E(D_{\lambda,H})$ if and only if
there exist $p,p'\in\{0,1\}$, $r\in R$, and $\rho\in\{0,1\}$ such 
that
\begin{equation}
\label{eq:d-arc}
\{[\ell,p],[r,\rho]\}\in E(F_\lambda)\setminus M_H
\quad\text{and}\quad
\{[\ell',p'],[r,\rho]\}\in M_H\,.
\end{equation}
That is, an arc $(\ell,\ell')\in E(D_{\lambda,H})$ indicates that 
(disregarding ports $p$ and $p'$) 
we can walk from $\ell$ to $\ell'$ in $F_\lambda$ by 
traversing first an edge not in $M_H$, followed by an edge in $M_H$.
We stress that the traversal \eqref{eq:d-arc} {\em preserves} 
the parity $\rho$ for consecutive edges, whereas the traversal in the 
proof of Lemma~\ref{lem:matching-to-extension} {\em changes} parity 
for consecutive edges; the latter also uses edges only in $M_H$. 

We recall that $G$ has minimum degree at least three; this enables us
to find a Hamiltonian cycle other than $H$ in $G$ with the help of a
directed cycle in $D_{\lambda,H}$ revealed in the following lemma. 

\begin{lemma}[Existence of an $s$-avoiding directed cycle in $D_{\lambda,H}$]
\label{lem:s-avoiding-cycle-in-d}
Suppose that every vertex of $G$ has degree at least three. Then, 
the directed graph $D_{\lambda,H}$ contains at least one directed
cycle that avoids the vertex $s$. 
\end{lemma}
\begin{proof}
It suffices to show that all vertices of $D_{\lambda,H}$ have out-degree
at least one and that the vertex $s$ is a source; that is, $s$ 
has in-degree zero. Towards this end, since every vertex of $G$ has degree 
at least three, for every $\ell\in L$ there exist distinct $r,r',r''\in R$ 
and three edges $\{\ell,r\},\{\ell,r'\},\{\ell,r''\}\in E(G)$. Furthermore, 
since $\ell$ gives rise to the vertices $[\ell,0]$ and $[\ell,1]$ in $F_\ell$, 
the edges of $M_H$ project to at most two of these three edges; without loss
of generality we may assume that the edge $\{\ell,r\}$ is not in the 
projection of $M_H$. For $p=1$ and the unique $\rho\in\{0,1\}$ 
such that \eqref{eq:f-parity} holds we thus have by~\eqref{eq:f-edge} 
that $\{[\ell,p],[r,\rho]\}\in E(F_\lambda)\setminus M_H$. Since $M_H$ is a
perfect matching in $F_\lambda$, there exist $\ell'\in L$ and $p'\in\{0,1\}$
such that $\{[\ell',p'],[r,\rho]\}\in M_H$; we must have $\ell'\neq\ell$ since
$\{\ell,r\}$ is not in the projection of $M_H$. 
Thus, we have $(\ell,\ell')\in E(D_{\lambda,H})$ by \eqref{eq:d-arc}.
Since $\ell\in L$ was arbitrary, we conclude that every vertex of 
$D_{\lambda,H}$ has out-degree at least one. 

It remains to show that $s$ has in-degree zero in $D_{\lambda,H}$. 
To reach a contradiction, suppose that there exists 
an arc $(\ell,s)\in E(D_{\lambda,H})$. From \eqref{eq:d-arc} we 
thus have that there exist $p,p',\rho\in\{0,1\}$ and $r\in R$ with
$\{[\ell,p],[r,\rho]\}\in E(F_\lambda)\setminus M_H$
and
$\{[s,p'],[r,\rho]\}\in M_H$. 
Let us split into two cases based on the value of $p'$ and obtain a
contradiction in both cases.

In the first case, suppose that $p'=0$. 
Recall that in the graph $F_\lambda$ we have
that the vertex $[s,0]$ is adjacent only to the vertex $[t,0]$.
Since $M_H$ is a perfect matching, we must thus have $\{[s,0],[t,0]\}\in M_H$,
and hence $r=t$ and $\rho=0$. 
From \eqref{eq:f-parity} we thus have that 
$0=\rho\equiv\lambda(\ell)+\iv{(\ell,r)\in \oa G_e}
=\chi_H(\ell)+\iv{(\ell,t)\in \oa G_e}\pmod 2$.
Now recall that $\chi_H(t)=1$. In particular, we have
$\chi_H(t)\not\equiv\chi_H(\ell)+\iv{(\ell,t)\in \oa G_e}\pmod 2$,
implying by \eqref{eq:oa-induced} that $(t,\ell)\in E(\oa G_e^{\chi_H})$.
But then since $\ell\neq s$ we have that the arc $(t,\ell)$ together 
with the directed $(\ell,t)$-subpath of the directed Hamiltonian 
path $\oa H\setminus(t,s)$ 
yields a directed cycle in $\oa G_e^{\chi_H}\setminus(t,s)$, 
a contradiction to acyclicity in Lemma~\ref{lem:acyclic-hamiltonicity}.

In the second case, suppose that $p'=1$. Then, in the construction 
of $M_H$ in Lemma~\ref{lem:extension-to-matching} we observe that
we must have $\ell_0=s$, $r_0=r$, and $\chi_H(r)=\chi_H(r_0)=\rho_0=\rho$.
In particular, we have $(s,r)\in E(\oa H)$. 
From \eqref{eq:f-parity} and $\lambda(\ell)=\chi_H(\ell)$
we thus conclude that
\[
\chi_H(r)=\rho\equiv \lambda(\ell)+\iv{(\ell,r)\in\oa G_e}\equiv\chi_H(\ell)+\iv{(\ell,r)\in\oa G_e} \pmod 2\,, 
\]
implying by \eqref{eq:oa-induced}
that $(\ell,r)\in\oa G_e^{\chi_H}$. But then since $\ell\neq s$ we have
that the arc $(\ell,r)$ together with the directed $(r,\ell)$-subpath of 
the directed Hamiltonian path $\oa H\setminus(t,s)$ 
yields a directed cycle in $\oa G_e^{\chi_H}\setminus(t,s)$, 
a contradiction to acyclicity in Lemma~\ref{lem:acyclic-hamiltonicity}.
\end{proof}

By the previous lemma we thus know that $D_{\lambda,H}$ contains
an $s$-avoiding directed cycle $\oa Q$ with
$V(\oa Q)=\{\ell_0,\ell_1,\ldots,\ell_{k-1}\}\subseteq L$ 
and $(\ell_j,\ell_{(j+1)\bmod k})\in E(\oa Q)$ for each $j=0,1,\ldots,k-1$ 
and $k\geq 2$. We will use $\oa Q$ to construct from $M_H$ another Hamiltonian 
cycle $H'\neq H$ in $G$. Figure~\ref{fig:cycle-switching} illustrates 
the construction.

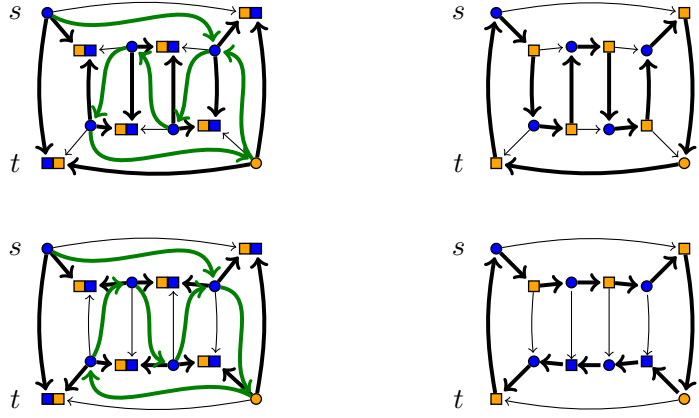
\begin{figure}[!ht]
\[
\begin{tikzpicture}[scale=0.5,shorten >=1pt, auto, node distance=1cm, ultra thick]
\node[rectangle](R) at (-4,2) {$s$};\node[rectangle](R) at (-4,-2) {$t$};\overlayerI
\end{tikzpicture}
\qquad \qquad \qquad
\begin{tikzpicture}[scale=0.5,shorten >=1pt, auto, node distance=1cm, ultra thick]
\node[rectangle](R) at (-4,2) {$s$};\node[rectangle](R) at (-4,-2) {$t$};\solI
\end{tikzpicture}
\]
\[
\begin{tikzpicture}[scale=0.5,shorten >=1pt, auto, node distance=1cm, ultra thick]
\node[rectangle](R) at (-4,2) {$s$};\node[rectangle](R) at (-4,-2) {$t$};\overlayerII
\end{tikzpicture}
\qquad \qquad \qquad
\begin{tikzpicture}[scale=0.5,shorten >=1pt, auto, node distance=1cm, ultra thick]
\node[rectangle](R) at (-4,2) {$s$};\node[rectangle](R) at (-4,-2) {$t$};\solII
\end{tikzpicture}
\]

 \caption{Obtaining another Hamiltonian cycle using a directed cycle in $D_{\lambda,H}$.
 Left: Two perfect matchings $M_H$ in $F_\lambda$ drawn as oriented overlays 
 of $G$ (cf.~Figure~\ref{fig:f-lambda}), with further overlaying
 drawn in green and constituting the arcs of~$D_{\lambda,H}$. 
 Right: The corresponding two orientations $\oa G_e^{\chi_H}$ and colorings
 $\chi_H$. Top and bottom: In both cases we have that $D_{\lambda,H}$ contains 
 a unique $s$-avoiding directed cycle $\oa Q$. Using $\oa Q$ in each case 
 we can switch between the top and bottom Hamiltonian cycles in $G$. 
 Note in particular that the two vertex colorings agree in $L$ but differ 
 in $R$.}
 \label{fig:cycle-switching}
\end{figure}

From \eqref{eq:d-arc} applied to each arc of $\oa Q$ in turn we conclude that 
for $j=0,1,\ldots,k-1$ there exist $r_j\in R$ and $\rho_j,p_j\in\{0,1\}$
with 
\begin{equation}
\label{eq:port-mixed-cycle}
\{[\ell_j,p_j],[r_j,\rho_j]\}\in E(F_\lambda)\setminus M_H
\quad\text{and}\quad
\{[\ell_{(j+1)\bmod k},p_j'],[r_j,\rho_j]\}\in M_H\,.
\end{equation}
We observe that the vertices $[r_j,\rho_j]$ for $j=0,1,\ldots,k-1$ 
are distinct because $M_H$ is a perfect matching and 
$\ell_{(j+1)\bmod k}$ for $j=0,1,\ldots,k-1$ are distinct.
In spite of this, the edges~\eqref{eq:port-mixed-cycle} 
for $j=0,1,2,\ldots,k-1$ need {\em not} form a cycle in $F_\lambda$ 
because we can have $p_j'\neq p_{(j+1)\bmod k}$. Here is where the fact 
that $\oa Q$ is $s$-avoiding pays off. Let $M=M_H$ and
recall that the vertices $[\ell,0]$ and $[\ell,1]$ for 
each $\ell\in L\setminus\{s\}$ have identical vertex neighborhoods 
in $F_\lambda$. Thus, whenever we have $p_j'\neq p_{(j+1)\bmod k}$, 
we can modify $M$ by transposing the vertices 
$[\ell_{(j+1)\bmod k},0]$ and $[\ell_{(j+1)\bmod k},1]$ 
in the edges of $M$; by the identical vertex neighborhoods property, 
the resulting $M$ will still be a perfect matching in $F_\lambda$. 
Moreover, we have $H=\ham{M_H}=\ham M$; indeed, the traversal construction in 
Lemma~\ref{lem:matching-to-extension} is insensitive to the specific values 
of the (opposite) ports. For $j=0,1,\ldots,k-1$ we thus now have 
\begin{equation}
\label{eq:port-uniform-cycle}
\{[\ell_j,p_j],[r_j,\rho_j]\}\in E(F_\lambda)\setminus M
\quad\text{and}\quad
\{[\ell_{(j+1)\bmod k},p_{(j+1)\bmod k}],[r_j,\rho_j]\}\in M\,;
\end{equation}
that is, these edges now form a $2k$-vertex cycle in $F_\lambda$. 
Let us write $A$ for this cycle in $F_\lambda$. 

Observe in particular from \eqref{eq:port-uniform-cycle}
that the edges of $A$ alternate between edges in $M$ and edges 
not in $M$. Thus, we have that the symmetric difference 
$M'=(M\setminus E(A))\cup(E(A)\setminus M)$ is a perfect matching in 
$F_\lambda$. Furthermore, $M'$ and $M$ project to a different set of
edges of $G$; indeed, from \eqref{eq:port-uniform-cycle} we have
that $r_j$ changes adjacency from $\ell_{(j+1)\bmod k}$ in $\ham{M}$
to $\ell_j$ in $\ham{M'}$ for $j=0,1,\ldots,k-1$. 
It follows that $H'=\ham{M'}\neq\ham{M}=H$. 
Thus, we have constructed a Hamiltonian cycle $H'$ in $G$ that is 
different from $H$. Moreover, this construction is computable in 
deterministic linear time. 
This completes the proof of Theorem~\ref{thm:ahc}. $\qedwhite$

\subsection{Logarithmic-space solvability in minimum degree three}

We next describe how we can implement the ideas of the previous 
section in an algorithm that uses little space. We consider as input 
a bipartite Pfaffian graph $G$ of minimum degree three, 
a consistently oriented Hamiltonian cycle $\oa H$ in $G$, and 
an arc $(t,s)\in \oa H$.
More precisely, we assume both graphs $G$ and $\oa H$ are given in the 
input as a list of adjacency lists for each vertex. 
We seek to output a list of edges of another Hamiltonian cycle $H'\neq H$ 
with $\{s,t\}\in E(H')$. We assume that the vertices of $G$ are represented
as $O(\log n)$-bit integers in the input, where $n$ is the number of vertices
in $G$.

\thmahclog*

\begin{proof}
Suppose that $G$ has $n$ vertices. 
We describe an algorithm that uses space that is only logarithmic in $n$; 
however, this algorithm no longer runs in linear time but merely 
in polynomial time.
Recalling the proof of Theorem~\ref{thm:ahc} and the 
alternating cycle $A$ in $F_\lambda$, the algorithm outputs 
(i) all edges in $H$ but not in $A$, and (ii) all edges in $A$ but not on $H$;
the edges (i) and (ii) together form another Hamiltonian cycle $H'$ 
containing the edge $e=\{s,t\}$. The algorithm relies on the following 
logarithmic-space subroutines to accomplish the listing (i) and (ii). 

First, from the given input $\oa H$ we can compute in space logarithmic 
in $n$ the following numerical identifier $\operatorname{id}(v)$ for any 
given vertex $v$. Namely, we set $\operatorname{id}(v)$ to equal the number 
of arcs along the consistently oriented $\oa H$ from $s$ to $v$. From
a given $v$ we can compute $\operatorname{id}(v)$ by keeping track of one 
vertex $w$ (where we currently are) and a counter $c$ (how many edges we 
have traversed along $\oa H$). Starting with $w=s$ and $c=0$, as long 
as $w\neq v$, we traverse arcs of $\oa H$, setting $w$ to the next vertex 
after $w$ on $\oa H$ and increasing $c$ by one, and repeat until we reach $v$,
at which point we return $c=\operatorname{id}(v)$. Also observe that
this identifier subroutine enables us to determine whether a given vertex is
in the set $L$ (even identifier) or the set $R$ (odd identifier) in 
the bipartition $(L,R)$ of $G$ with $s\in L$.

Second, from Lemma~\ref{lem:pfaffian-orientation-from-hamiltonian-cycle} 
applied to $G$ and $H$ with $e=\{s,t\}$ we observe that the Pfaffian 
orientation $\oa G^\chi$, where $\chi$ is the proper coloring of the vertices 
of $G$ with $\chi(s)=0$, has the property that each edge $\{u,v\}\in E(G)$ 
is oriented from $u$ to $v$ in $\oa G^\chi$ if and only if 
$\operatorname{id}(u)<\operatorname{id}(v)$. Thus, using the subroutine
for the vertex identifiers, we can compute in logarithmic space in $n$ 
the orientation of any given edge $\{u,v\}\in E(G)$ in $\oa G^\chi$.

Third, we develop a subroutine for accessing the arcs of a subgraph
$D'_{\lambda,H}$ of $D_{\lambda,H}$ with $V(D'_{\lambda,H})=V(D_{\lambda,H})=L$
and with the property that there is exactly one out-arc from each vertex.
That is, we describe a subroutine that given a vertex $v\in L$ computes
in logarithmic space the end-vertex $u\in L$ of an arc $(v,u)$ in 
$D_{\lambda,H}$. To accomplish this, we compute $\operatorname{id}(v)$, 
find the first vertex $w$ adjacent to $v$ in $G$ from the adjacency list 
for $v$, such that $\{v,w\}\not \in E(H)$. Next, we compute 
$\operatorname{id}(w)$. If $\operatorname{id}(v)<\operatorname{id}(w)$, 
we locate the vertex $u$ as the one immediately preceding $w$ along $\oa H$;
otherwise, that is, when $\operatorname{id}(v)>\operatorname{id}(w)$, we 
locate the vertex $u$ as the one immediately succeeding $w$ along $\oa H$. 
By the structure of $\chi$ and $\oa G^\chi$, this will ensure that 
the parity at $w$ is the same in $F_\lambda$ for the edges projecting 
to $\{v,w\}$ and $\{u,w\}$ in $G$.

Fourth, we can find a vertex on the unique cycle in $D'_{\lambda,H}$ 
by starting from $s$ and walking along the arcs of $D'_{\lambda,H}$ 
for $n$ steps. The subroutine again only needs two additional variables, the current 
vertex and a counter keeping track of how many steps we have taken; each step
is taken with the subroutine in the previous paragraph.
Once we have found a vertex on the cycle, we can easily enumerate the 
vertices in $A$ along a consistent orientation $\oa A$ by traversing 
the arcs on the cycle in $D'_{\lambda,H}$ (again using the subroutine 
from the previous paragraph) until we get back to the starting vertex
on the cycle. This also requires storing only two pointers (vertices).

Finally, the listing (i) (the listing (ii)) can be done by making 
one revolution over $\oa H$ (over $\oa A$) and for each arc traversed using 
the subroutines for one revolution over $\oa A$ (over $\oa H$) to check that 
the underlying edge is not in $A$ (not in $H$). Both enumerations are thus
computable in logarithmic space. 
\end{proof}

\subsection{Thomason's lollipop method in cubic bipartite Pfaffian graphs}
\label{sect: lollipop}

In this section we prove that Thomason's lollipop method runs in a linear number of steps 
in cubic bipartite Pfaffian graphs. 
Let us first set up some preliminaries and then describe 
Thomason's lollipop method.

Let $G$ be a Hamiltonian cubic graph and let $H$ be a Hamiltonian cycle in $G$.
Select an edge $e=\{s,t\}$ in the Hamiltonian cycle $H$. Let us
call $e$ the {\em anchor} edge.
A \emph{lollipop} is a connected graph with one vertex of degree one, 
one vertex of degree three, and all other vertices of degree two. 
All lollipops considered in what follows are subgraphs of $G$
such that $t$ is the unique degree-one vertex and $e$ is 
the edge incident to $t$ on the lollipop. 

The lollipop method is best described as operating on a family of 
Hamiltonian paths in $G$. We say that a Hamiltonian path 
in $G$ that starts at the vertex $t$ and continues via the anchor edge 
$e$ is an $e$-\emph{anchored} Hamiltonian path. Now recall that $G$ is 
cubic, so the vertex $t$ is adjacent to $s$ (via the anchor edge $e$)
and to two other vertices $a$ and $b$. 
The lollipop method transforms a given $e$-anchored Hamiltonian path $P_e$ 
that ends at either $a$ or $b$ into an $e$-anchored 
Hamiltonian path $P_e'\neq P_e$ that ends at either $a$ or $b$.
Observe in particular that both Hamiltonian paths $P_e$ and $P_e'$ can be
completed into $e$-anchored Hamiltonian cycles by adding the missing edge 
$\{a,t\}$ or $\{b,t\}$ into the respective path. 

The transformation from $P_e$ to $P_e'$ is via a sequence
of \emph{lollipop steps}. A lollipop step consists of adding one edge to 
an $e$-anchored Hamiltonian path and removing another one, so that 
another $e$-anchored Hamiltonian path is formed. More precisely, 
let $Q$ be an $e$-anchored Hamiltonian path ending at some vertex $u$.
Since $G$ is cubic, $u$ is adjacent to two other vertices, $x$ and $y$,
such that the edges $\{u,x\}$ and $\{u,y\}$ of $G$ are not in $Q$. 
Assume that $x\neq t$. 
Add the edge $\{u,x\}$ into $Q$ to obtain a lollipop $\Omega$ where the
unique degree-three vertex is $x$. Now observe that among the three
adjacent vertices to $x$ there is a unique vertex $v\notin\{u,t\}$ 
such that both $\{v,x\}\in E(\Omega)$ and removing $\{v,x\}$ from $\Omega$
leaves an $e$-anchored Hamiltonian path $Q'$ ending at $v$.
The transformation from $Q$ to $Q'$ now constitutes one lollipop step.
Observe also that lollipop steps are {\em reversible}; that is, 
we can go back to $Q$ from $Q'$ by performing a lollipop step starting
from $Q'$.

The \emph{lollipop state graph} $\mathcal{L}(G,s,t)$ has as its vertices the 
$e$-anchored Hamiltonian paths in $G$ and two vertices are joined by 
an edge if and only if it is possible to transform between the 
$e$-anchored Hamiltonian paths by one lollipop step.
We observe immediately that $\mathcal{L}(G,s,t)$ 
has no isolated vertices---indeed, from any vertex $Q$ we can arrive at
another vertex $Q'\neq Q$ by a lollipop step---and the degree-one vertices
are exactly the $e$-anchored Hamiltonian paths $Q$ that end at a vertex
$u$ adjacent to $t$ in $G$; that is, $u\in\{a,b\}$; moreover, all other
vertices have degree two. Thus, we can transform 
from $P_e$ to $P_e'\neq P_e$ by tracing a path in $\mathcal{L}(G,s,t)$ 
from~$P_e$ to~$P_e'$.

We now proceed to prove an upper bound on the maximum length of a 
path in $\mathcal{L}(G,s,t)$ on a cubic bipartite Pfaffian graph $G$. 
An example of the lollipop method applied to a cubic bipartite planar graph 
using the terminology in the subsequent proof is given 
in Figure~\ref{fig:lollipop}.

\begin{figure}[!ht]
\[
\begin{tikzpicture}[scale=0.5,shorten >=1pt, auto, node distance=1cm, ultra thick]
\node[rectangle](R) at (0,4) {$H$};\node[rectangle](R) at (-4.5,2.3) {$s$};\node[rectangle](R) at (-4.5,-2.3) {$t$};\Hlp
\end{tikzpicture}
\qquad \qquad 
\begin{tikzpicture}[scale=0.5,shorten >=1pt, auto, node distance=1cm, ultra thick]
\node[rectangle](R) at (0,4) {$Q_0$};\node[rectangle](R) at (-4.5,2.3) {$s$};\node[rectangle](R) at (-4.5,-2.3) {$t$};\PlpO
\end{tikzpicture}
\qquad \qquad 
\begin{tikzpicture}[scale=0.5,shorten >=1pt, auto, node distance=1cm, ultra thick]
\node[rectangle](R) at (0,4) {$Q_1$};\node[rectangle](R) at (-4.5,2.3) {$s$};\node[rectangle](R) at (-4.5,-2.3) {$t$};\PlpI
\end{tikzpicture}
\]
\[
\begin{tikzpicture}[scale=0.5,shorten >=1pt, auto, node distance=1cm, ultra thick]
\node[rectangle](R) at (0,4) {$Q_2$};\node[rectangle](R) at (-4.5,2.3) {$s$};\node[rectangle](R) at (-4.5,-2.3) {$t$};\PlpII
\end{tikzpicture}
\qquad \qquad 
\begin{tikzpicture}[scale=0.5,shorten >=1pt, auto, node distance=1cm, ultra thick]
\node[rectangle](R) at (0,4) {$Q_3$};\node[rectangle](R) at (-4.5,2.3) {$s$};\node[rectangle](R) at (-4.5,-2.3) {$t$};\PlpIII
\end{tikzpicture}
\qquad \qquad 
\begin{tikzpicture}[scale=0.5,shorten >=1pt, auto, node distance=1cm, ultra thick]
\node[rectangle](R) at (0,4) {$Q_4$};\node[rectangle](R) at (-4.5,2.3) {$s$};\node[rectangle](R) at (-4.5,-2.3) {$t$};\PlpIV
\end{tikzpicture}
\]
\[
\begin{tikzpicture}[scale=0.5,shorten >=1pt, auto, node distance=1cm, ultra thick]
\node[rectangle](R) at (0,4) {$Q_5$};\node[rectangle](R) at (-4.5,2.3) {$s$};\node[rectangle](R) at (-4.5,-2.3) {$t$};\PlpV
\end{tikzpicture}
\qquad \qquad 
\begin{tikzpicture}[scale=0.5,shorten >=1pt, auto, node distance=1cm, ultra thick]
\node[rectangle](R) at (0,4) {$Q_6$};\node[rectangle](R) at (-4.5,2.3) {$s$};\node[rectangle](R) at (-4.5,-2.3) {$t$};\PlpVI
\end{tikzpicture}
\qquad \qquad 
\begin{tikzpicture}[scale=0.5,shorten >=1pt, auto, node distance=1cm, ultra thick]
\node[rectangle](R) at (0,4) {$Q_7$};\node[rectangle](R) at (-4.5,2.3) {$s$};\node[rectangle](R) at (-4.5,-2.3) {$t$};\PlpVII
\end{tikzpicture}
\]
\[
\begin{tikzpicture}[scale=0.5,shorten >=1pt, auto, node distance=1cm, ultra thick]
\node[rectangle](R) at (0,4) {$Q_8$};\node[rectangle](R) at (-4.5,2.3) {$s$};\node[rectangle](R) at (-4.5,-2.3) {$t$};\PlpVIII
\end{tikzpicture}
\qquad \qquad 
\begin{tikzpicture}[scale=0.5,shorten >=1pt, auto, node distance=1cm, ultra thick]
\node[rectangle](R) at (0,4) {$Q_9$};\node[rectangle](R) at (-4.5,2.3) {$s$};\node[rectangle](R) at (-4.5,-2.3) {$t$};\PlpIX
\end{tikzpicture}
\qquad \qquad 
\begin{tikzpicture}[scale=0.5,shorten >=1pt, auto, node distance=1cm, ultra thick]
\node[rectangle](R) at (0,4) {$H'$};\node[rectangle](R) at (-4.5,2.3) {$s$};\node[rectangle](R) at (-4.5,-2.3) {$t$};\Hlpprime
\end{tikzpicture}
\]
 \caption{An example of the sequence of steps of Thomason's lollipop method in an $n$-vertex cubic bipartite planar graph viewed as a sequence of arc reversals in the directed graph $D_{\lambda,H}$.
The given input $G$ and $H$ together with $e=\{s,t\}$ is displayed in the
top left; the black arcs are oriented as in $\oa G_e$ obtained from Lemma~\ref{lem:pfaffian-orientation-from-hamiltonian-cycle} on input $H$.
We display the initial Hamiltonian cycle $H$ (top left)
and the final Hamiltonian cycle $H'$ (bottom right) obtained by the method, 
as well as the intermediate $e$-anchored 
Hamiltonian paths $Q_0,Q_1,\ldots,Q_9$ obtained in consecutive lollipop steps; 
the end-vertex of each $Q_i$ is highlighted with red. 
The green arcs in $Q_0$ are the arcs of $D_{\lambda,H}$. 
Observe that each lollipop step from $Q_i$ to $Q_{i+1}$ can be understood as
reversing the light-green arc in $Q_i$; the method terminates
when the end-vertices of $Q_0$ and $Q_{i+1}$ agree. By the structure of
$D_{\lambda,H}$, we must have $i\leq n$; cf.~Theorem~\ref{thm:tho}.}
\label{fig:lollipop}
\end{figure}
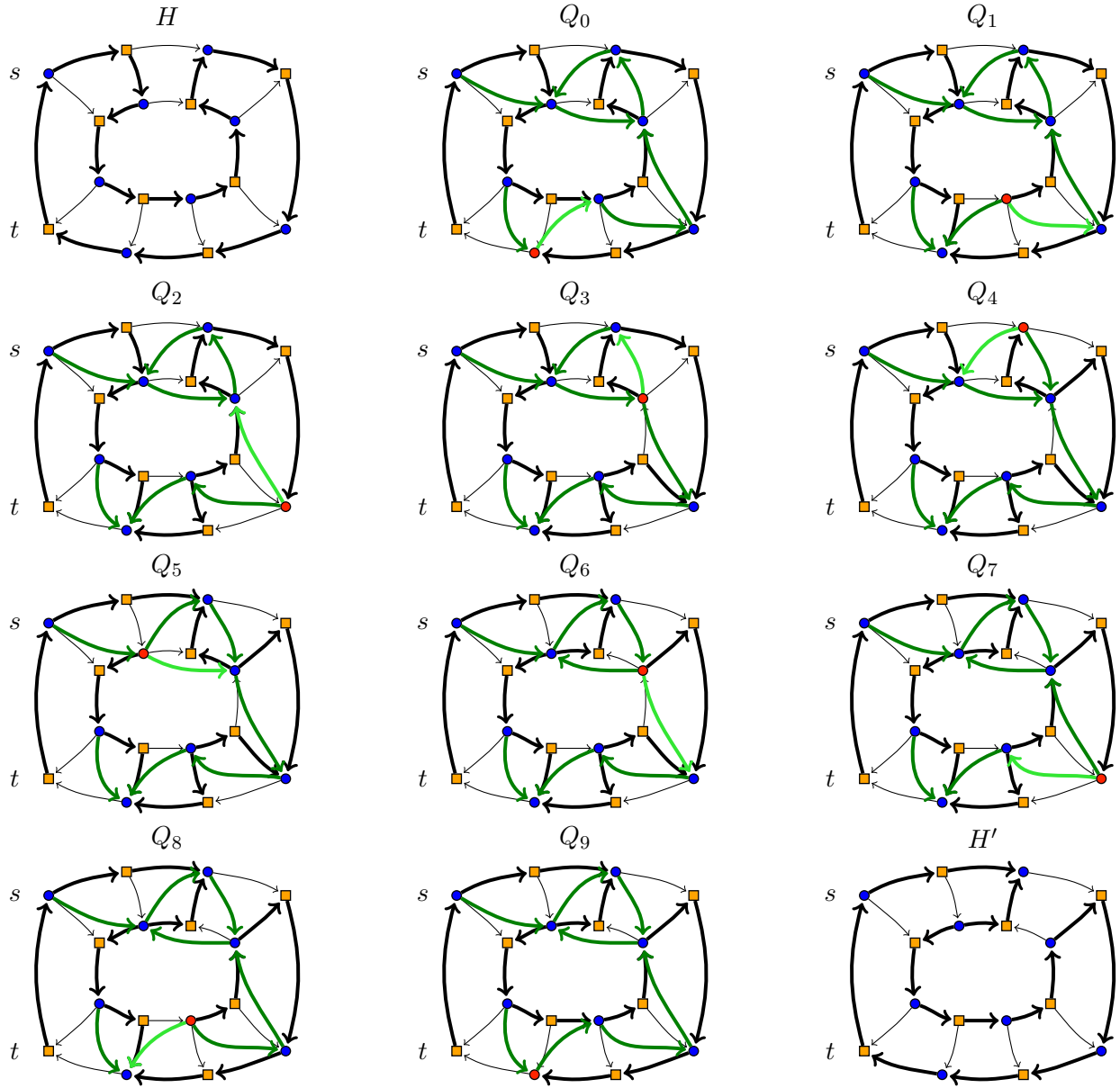

\thmtho*

\begin{proof}
To analyze the lollipop method in cubic bipartite Pfaffian graphs, 
let a cubic bipartite Pfaffian graph $G$, a Hamiltonian cycle $H$ in $G$, 
and $e=\{s,t\}\in E(H)$ be given as input. This input enables us to work
in the setting of Section~\ref{sect:linear-time-solvability}; let 
the vertex bipartition $(L,R)$ of $G$, the coloring $\lambda$ of $L$,
the graph $F_\lambda$, the perfect matching $M_H$ in $F_\lambda$,
and the directed graph $D_{\lambda,H}$ be constructed accordingly. 
Recall that $s\in L$ and $t\in R$.  

The lollipop method starts by removing the edge $\{t,u\}$ 
with $u\neq s$ from $H$ to obtain $e$-anchored Hamiltonian path $P_e$. 
Let $Q_0,Q_1,\ldots,Q_h$ be the sequence of $e$-anchored Hamiltonian 
paths traversed by consecutive lollipop steps in $\mathcal{L}(G,s,t)$ 
with $P_e=Q_0$ and $Q_h=P_e'$. We will show that $P_e'$ ends at $u$ and 
thus we can obtain a Hamiltonian cycle $H'\neq H$ by inserting the 
edge $\{t,u\}$ into $P_e'$. Moreover and crucially, we will show 
that $h\leq n$.

Our analysis of the lollipop method is based on the directed graph 
$D_{\lambda,H}$. 
We recommend consulting Figure~\ref{fig:lollipop} for intuition at this point.
Recall from the proof of Lemma~\ref{lem:s-avoiding-cycle-in-d} that,
in the directed graph $D_{\lambda,H}$, the vertex $s$ has in-degree zero
and every vertex has out-degree at least one. In particular, by traversing
out-arcs from the vertex $u$ in $D_{\lambda,H}$, and traversing the eventual 
directed cycle encountered, as well as traversing backwards to $u$ from the 
directed cycle, in precise terms we observe that there
exist vertices $w_0,w_1,\ldots,w_{d-1}$ with 
$(w_j,w_{(j+1)\bmod d})\in E(D_{\lambda,H})$ 
for $j=0,1,\ldots,d-1$
as well as vertices $w_0',w_1',\ldots,w_{d'}'$ 
with $w_0=w_{d'}'$, $w_0'=u$, 
$\{w_0,w_1,\ldots,w_{d-1}\}\cap\{w_0',w_1',\ldots,w_{d'-1}'\}=\emptyset$,
and 
$(w_j',w_{j+1}')\in E(D_{\lambda,H})$ for $j=0,1,\ldots,d'-1$.
That is, the sequence $w_0',w_1',\ldots,w_{d'}'$ forms
a directed path starting at the vertex $u=w_0'$ and ending at the vertex
$w_{d'}=w_0$, which is on the directed cycle formed by the 
vertices $w_0,w_1,\ldots,w_{d-1}$ in $D_{\lambda,H}$; the directed cycle and
the directed path intersect exactly at the vertex $w_{d'}=w_0$. 
In particular $d+d'\leq |L|=n/2$.

It will be convenient to introduce the following sequence of vertices
visited on the traversal of $D_{\lambda,H}$ from $u$. 
For $i=0,1,\ldots,2d'+d$, define
\begin{equation}
\label{eq:lollipop-vi}
v_i=
\begin{cases}
w_i'         & \text{for $i=0,1,\ldots,d'-1$}\,;\\
w_{i-d'}     & \text{for $i=d',d'+1,\ldots,d'+d-1$}\,;\\
w_{2d'+d-i}' & \text{for $i=d'+d,d'+d+1,\ldots,2d'+d$}\,.
\end{cases}
\end{equation}
We have $(v_i,v_{i+1})\in E(D_{\lambda,H})$ for $i=0,1,\ldots,d'+d-1$;
these arcs are precisely the arcs traversed forward. 
We have $(v_{i+1},v_i)\in E(D_{\lambda,H})$ for $i=d'+d,d'+d+1,\ldots,2d'+d-1$;
these arcs are precisely the arcs traversed backward.
For an arc $(\ell,\ell')\in E(D_{\lambda,H})$, let us write 
$r(\ell,\ell')$, $\rho(\ell,\ell')$, and $p'(\ell,\ell')$, respectively,
for the unique $r\in R$, $\rho\in\{0,1\}$, and $p'\in\{0,1\}$
such that \eqref{eq:d-arc} holds. Also, let us write $p(\ell,\ell')$
for the minimum $p\in\{0,1\}$ such that that \eqref{eq:d-arc} holds.

Let $M^-$ be a matching with $n-1$ edges in $F_\lambda$ such that 
the vertex $[t,1]$ is left unmatched by $M^-$;
we call such matchings {\em almost perfect}---indeed, any perfect matching
$M$ in $F_\lambda$ has $n$ edges. Also observe that the other vertex left 
unmatched by $M^-$ is $[\ell,p]$ for some $\ell\in L$ and $p\in\{0,1\}$.
Recall the parity-and-port-changing traversal of $M$ in 
the proof of Lemma~\ref{lem:matching-to-extension} resulting
in the Hamiltonian cycle $\ham M$. 
Define a similar parity-and-port-changing traversal of $M^-$ by starting 
at the vertex $[\ell,\bar p]$ and observe by a similar argument as in 
the proof of Lemma~\ref{lem:matching-to-extension} that this traversal 
defines an $e$-anchored Hamiltonian path $\hp{M^-}$ from the vertex 
$\ell$ to the vertex $t$ in $G$; in particular, observe that 
$\hp{M^-}$ is $e$-anchored since by the structure of $F_\lambda$
the almost perfect $M^-$ must contain the edge $\{[s,0],[t,0]\}$. 

We now proceed to characterize the $e$-anchored Hamiltonian paths 
$Q_0,Q_1,\ldots,Q_h$ using corresponding almost perfect matchings 
$M_0^-,M_1^-,\ldots,M_h^-$, and conclude that $h=2d'+d\leq n$
in the process. 
For $i=0,1,\ldots,h$, let us write $u_i$ for the end-vertex of $Q_i$ 
other than $t$. Recalling that $Q_0=P_e$ is constructed by deleting the
edge $\{u,t\}$ from the Hamiltonian cycle $H$, let $p\in\{0,1\}$ be the 
port and $f\in E(F_\lambda)$ the edge with $f=\{[u,p],[t,1]\}\in M_H$. 
Take $M_0^-=M_H\setminus \{f\}$. 
In particular, we have $Q_0=P_e=\hp{M_0^-}$ and $u_0=v_0=u$. Let $p_0=p$;
we will fix values $p_i\in\{0,1\}$ for $i=1,2,\ldots,h$ as we progress 
in what follows.

We split the analysis into two ranges based on the parameter $i$. The
first range corresponds to the forward-traversal of arcs 
in $D_{\lambda,H}$. For $i=0,1,\ldots,d'+d-1$, 
we say an almost perfect matching $M^-$ has {\em property} $i$ if 
\begin{enumerate}
\item[(i)]
$[v_i,p_i]$ is left unmatched by $M^-$; and
\item[(ii)]
we have 
$\{[v_j,p(v_j,v_{j+1})],[r(v_j,v_{j+1}),\rho(v_j,v_{j+1})]\}\in M^-$ and\\
$\{[r(v_j,v_{j+1}),\rho(v_j,v_{j+1})],[v_{j+1},p'(v_j,v_{j+1})]\}\notin M^-$
for all $0\leq j\leq i-1$; and
\item[(iii)]
we have 
$\{[v_j,p(v_j,v_{j+1})],[r(v_j,v_{j+1}),\rho(v_j,v_{j+1})]\}\notin M^-$ and\\
$\{[r(v_j,v_{j+1}),\rho(v_j,v_{j+1})],[v_{j+1},p'(v_j,v_{j+1})]\}\in M^-$
for all $i\leq j\leq d'+d$.
\end{enumerate}
The second range corresponds to the backward-traversal of arcs 
in $D_{\lambda,H}$. For $i=d'+d,d'+d+1,\ldots,2d'+d$, we 
we say an almost perfect matching $M^-$ has {\em property} $i$ if 
\begin{enumerate}
\item[(i')]
$[v_i,p_i]$ is left unmatched by $M^-$; and
\item[(ii')]
we have 
$\{[v_j,p(v_j,v_{j+1})],[r(v_j,v_{j+1}),\rho(v_j,v_{j+1})]\}\in M^-$ and\\
$\{[r(v_j,v_{j+1}),\rho(v_j,v_{j+1})],[v_{j+1},p'(v_j,v_{j+1})]\}\notin M^-$\\
for all $d'\leq j\leq d'+d-1$ as well as for all $0\leq j\leq 2d'+d-1-i$; and
\item[(iii')]
we have 
$\{[v_j,p(v_j,v_{j+1})],[r(v_j,v_{j+1}),\rho(v_j,v_{j+1})]\}\notin M^-$ and\\
$\{[r(v_j,v_{j+1}),\rho(v_j,v_{j+1})],[v_{j+1},p'(v_j,v_{j+1})]\}\in M^-$
for all $2d'+d-i\leq j\leq d'-1$.
\end{enumerate}
From previous observations and \eqref{eq:d-arc} we have that $M^-_0$ satisfies 
property $0$. 

Let us now analyse the lollipop step mapping $Q_i$ to $Q_{i+1}$
one value $i=0,1,\ldots,d'+d-1$ at a time.
Suppose that there is an almost perfect matching $M_i^-$ of $F_\lambda$ 
that satisfies property $i$ and that $Q_i=\hp{M_i^-}$. In particular,
we have $u_i=v_i$ by (i) and $Q_i=\hp{M_i^-}$. We claim that the vertex
$r(v_i,v_{i+1})$ is the unique degree-three vertex in the lollipop formed
by the lollipop step transforming $Q_i$ to $Q_{i+1}$. Observe by (iii) that 
$\{[r(v_i,v_{i+1}),\rho(v_i,v_{i+1})],[v_{i+1},p'(v_i,v_{i+1})]\}\in M_i^-$,
implying that $\{r(v_i,v_{i+1}),v_{i+1}\}$ is an edge in $Q_i=\hp{M_i^-}$. 
Recalling that $M_i^-$ is almost perfect, all vertices in $R\times\{0,1\}$ are
matched, so $r(v_i,v_{i+1})\in R$ is in fact adjacent to another vertex 
$\omega_i\neq v_{i+1}$ along an edge in $Q_i=\hp{M_i^-}$. 
By (iii) and \eqref{eq:d-arc} we have $\{v_i,r(v_i,v_{i+1})\}$ is an edge
in $G$ but not in $Q_i=\hp{M_i^-}$, and $Q_i$ ends at $v_i$. 
Thus, $r(v_i,v_{i+1})$ is the unique degree-three vertex in the lollipop.
Next, the lollipop step proceeds to delete an edge adjacent to the degree-three 
vertex $r(v_i,v_{i+1})$ in the lollipop. This edge is 
$\{v_{i+1},r(v_i,v_{i+1})\}$ by the previous analysis. It follows
that $Q_{i+1}$ is obtained from $Q_i$ by deleting $\{v_{i+1},r(v_i,v_{i+1})\}$
and inserting $\{v_i,r(v_i,v_{i+1})\}$. Thus, $Q_{i+1}$ ends
at $u_{i+1}=v_{i+1}$. Define $M_{i+1}^-$ by
starting with $M_i^-$ and deleting the edge 
$\{[r(v_i,v_{i+1}),\rho(v_i,v_{i+1})],[v_{i+1},p'(v_i,v_{i+1})]\}$
as well as inserting the edge
$\{[v_i,p(v_i,v_{i+1})],[r(v_i,v_{i+1}),\rho(v_i,v_{i+1})]\}$.
Fix $p_{i+1}=p'(v_i,v_{i+1})$.
From (i), (ii), and (iii) we have that $M_{i+1}^-$ is an almost perfect
matching that satisfies property $i+1$. Furthermore, $Q_{i+1}=\hp{M_{i+1}^-}$.

Analysis of the lollipop step mapping $Q_i$ to $Q_{i+1}$
for $i=d'+d,d'+d+1,\ldots,2d'+d-1$ is now similar, but relying on
properties (i'), (ii'), (iii') instead. From the existence of an
almost perfect matching $M_i^-$ of $F_\lambda$ that satisfies 
property $i$ and $Q_i=\hp{M_i^-}$, by a similar analysis 
we conclude that there exists an almost perfect matching
$M_{i+1}^-$ of $F_\lambda$ that satisfies 
property $i+1$ and $Q_{i+1}=\hp{M_{i+1}^-}$. 
Since $v_{2d'+d}=u$ and $u$ is adjacent to $t$ in $G$, from (i') we 
conclude in particular that $P_e'=Q_{2d'+d}$ and thus $h=2d'+d$. 
Since $2d'+d\leq n$, we have shown 
that the lollipop method terminates in at most $n$ lollipop steps.
\end{proof}

We note that the algorithm implicit in the proof not only uses at most a linear number of lollipop steps, but also can be implemented with the guidance of $D_{\lambda,H}$ to run in linear time.

\subsection{Graph-theoretic structural corollaries}

\label{sect:structural}

We now restate and prove Corollaries~\ref{cor:ahc} and~\ref{cor:chc}.

\corahc*

\begin{proof}
The proof of Theorem~\ref{thm:ahc} in the previous section shows how 
to generate one other Hamiltonian cycle $H'\neq H$ such that 
the vertex colorings $\chi_{H'}$ and $\chi_H$ 
differ in at least one vertex in $R$ but agree for all vertices in $L$.
We can change the roles of $L$ and $R$, $s$ and $t$, and the two colors in the construction, and 
compute another Hamiltonian cycle $H''$ from $H$ with $\chi_{H''}$ 
different from $\chi_H$ in $L$ but the same in $R$. 
Clearly, as the vertex coloring is unique for a Hamiltonian cycle, 
$H'$ and $H''$ must be different Hamiltonian cycles that are 
also different from $H$. Also, taking either of $H'$ or $H''$ as 
the source, again flipping the roles of $L$ and $R$, we can generate 
a fourth Hamiltonian cycle $H'''$ whose coloring $\chi_{H'''}$ is different 
from the colorings of $H$, $H'$, and $H''$. This way we have generated at least 
three new distinct Hamiltonian cycles $H',H'',$ and $H'''$ from $H$.
\end{proof}

\corchc*

\begin{proof}
Given a Hamiltonian cycle $H$, we can use Theorem~\ref{thm:ahc} to obtain 
another Hamiltonian cycle $H'$ that in particular must have an edge 
$f\in E(H')\setminus E(H)$. From Corollary~\ref{cor:ahc} we have 
thus that there exist at least four distinct Hamiltonian cycles in $G$
with $f$ as the anchor edge. This means we have generated at least five 
distinct Hamiltonian cycles in $G$. 
From Bos\'ak's theorem~\cite{Bosak1967} we know every cubic bipartite graph 
must have an even number of distinct Hamiltonian cycles, which shows that there 
must be at least six distinct Hamiltonian cycles in every cubic bipartite 
Pfaffian graph. 
\end{proof}

\section{Faster algorithms for TSP and counting Hamiltonian cycles}
\label{sect:fhc}
In this section we describe and analyze our new algorithms for TSP and counting Hamiltonian cycles restricted to bipartite Pfaffian graphs, in particular bipartite planar ones.

\subsection{Algorithmic preliminaries and conventions}
Our algorithms in this section differ from earlier ones on planar graphs 
in that we begin by computing and fixing a Pfaffian orientation $\oa G$ of the 
input graph $G$. This can be done in linear time in the case of 
planar graphs, and in $O(n^3)$ time with the algorithm 
in~\cite{Robertson1999} for any bipartite Pfaffian graph. 
We also fix an arbitrary edge $e\in E(G)$ with $e=\{s,t\}$ for distinct 
$s,t\in V(G)$ and compute the orientation $\oa G_e$ of $G$ as well as 
the vertex bipartition $(L,R)$ of $G$ with $s\in L$ and $t\in R$. 

We will first consider the anchored problems of finding the shortest traveling 
salesperson tour through $e$, and counting the Hamiltonian cycles 
through $e$. To get the full solution, we can consider each edge in turn 
and in the case of counting, divide the total amount by $n$ in the end. 
This only incurs a polynomial overhead on the running time.

Apart from this initial step of fixing an edge and a Pfaffian orientation, 
our algorithm follows and uses well-known material for dynamic programming 
over graph decompositions and is only repeated here for the reader's 
convenience. The only new parts are what we count, how to represent 
partial solutions of what we count, and how to update them. 

\subsection{Graph decompositions}
We consider some of the most familiar graph decompositions, originally proposed in Robertson and Seymour~\cite{Robertson1986, Robertson1991}. We will not explicitly consider tree width as a parameter in our algorithms but define it anyway along the way to explain a path decomposition.

\medskip
\noindent
{\em Tree and path decompositions.}
A \emph{tree decomposition} of an undirected graph $G$ is a tree $T_\td$
in which each vertex $x\in V(T_\td)$ is associated with a set of 
vertices $B_x\subseteq V(G)$, called a \emph{bag}, 
such that $\cup_{x\in V(T_\td)} B_x=V(G)$ with the properties
\begin{enumerate}
\item 
for every $\{u,v\}\in E(G)$, 
there exists an $x\in V(T_\td)$ such that both $u,v\in B_x$, and
\item 
if $v\in B_x$ and $v\in B_y$, then $v\in B_z$ for all $z\in V(T_\td)$ on 
the path joining $x$ and $y$ in $T_\td$.
\end{enumerate}

The width of a tree decomposition is measured in the size of its largest bag. 
More precisely, we say that the {\em width} of the decomposition is equal to 
$\max_{x\in V(T_\td)} |B_x|-1$. A \emph{path decomposition} is a tree 
decomposition where the tree $T_\td$ is a path. We also say that the 
{\em tree width} $\mbox{tw}(G)$, and the {\em path width} $\mbox{pw}(G)$, 
of the 
graph $G$ equals the minimum width of any such decompositions for the graph.

A \emph{nice} tree decomposition as defined in~\cite{Bodlaender2015} is 
a tree decomposition $T_\td$ with one special bag $r$ called the
{\em root} and in which each bag is one of the following types:
\begin{enumerate}
\item 
Leaf bag: a leaf $x\in V(T_\td)$ with $B_x = \emptyset$.
\item 
Introduce vertex bag: an internal vertex $x\in V(T_\td)$ with 
one child vertex $y$ for which $B_x = B_y \cup \{v\}$ for some 
$v \not \in B_y$. This bag is said to {\em introduce} $v$.
\item 
Introduce edge bag: an internal vertex $x\in V(T_\td)$ labeled with 
an edge $\{u,v\} \in E(G)$ with one child bag $y$ for 
which $u, v \in B_x = B_y$. This bag is said to {\em introduce} $\{u,v\}$.
\item 
Forget bag: an internal vertex $x\in V(T_\td)$ with one child 
bag $y$ for which $B_x = B_y \setminus \{v\}$ for some $v \in B_y$. 
This bag is said to {\em forget} $v$.
\item 
Join bag: an internal vertex $x\in V(T_\td)$ with two child 
vertices $p$ and $q$ with $B_x = B_p = B_q$.
\end{enumerate}
Moreover, every edge $\{u,v\}\in E(G)$ is introduced exactly once. 
Proposition~2.2~in~\cite{Bodlaender2015}~shows how to compute a nice tree 
decomposition from an arbitrary tree decomposition of the same tree 
width in polynomial time. A nice path decomposition is defined similarly 
and has no join bags. It is preferable to put the root vertex at one 
of the two leaf bags.

\medskip
\noindent
{\em Branch decompositions.}
A \emph{branch decomposition} of $G$ is an unrooted binary tree 
$T_\bd$ in which edges are associated with subsets of $V(G)$, 
called \emph{middle sets}. Each leaf in $T_\bd$ is associated with a unique
edge of $G$. With each edge $x\in E(T_\bd)$ we associate the
{\em middle set} $B_x$ consisting of every vertex $v\in V(G)$ that 
appears in both subtrees obtained by removing $x$ from $T_\bd$. That is, 
for each vertex $v\in V(G)$ both of the subtrees have at least 
one leaf associated with a graph edge in $G$ incident to $v$.
The {\em width} of the branch decomposition is $\max_{x\in E(T_\bd)} |B_x|-1$. 
The {\em branch width} $\mbox{bw}(G)$ of a graph is the minimum width 
of any branch decomposition of $G$. An optimal branch decomposition 
in a planar graph can be found in polynomial time, 
see Gu and Tamaki~\cite{GuT2008}.

\subsection{Vertex states in a graph decomposition}

The correspondence between perfect matchings and good extensions 
(Lemma~\ref{lem:matching-to-extension} and 
Lemma~\ref{lem:extension-to-matching}) makes it possible to count and 
detect optimally edge-weighted Hamiltonian cycles by considering 
all colorings in a dynamic programming fashion over graph decompositions. 
That is, recalling Section~\ref{sect:finding-a-good-coloring}, 
we consider all ways to {\em simultaneously} produce 
a coloring $\lambda:L\rightarrow \{0,1\}$ and as well as a 
perfect matching $M$ in the corresponding bipartite graph $F_\lambda$. 
However, there will be no notion of ports of the vertices in $L$ in 
our dynamic programming. Indeed, we will just make sure by our vertex 
states that every vertex in $L$ is mapped twice. Since edges of the graph 
are considered in a fixed order across the dynamic programming over 
the graph decomposition, we may interpret the matching in $F_\lambda$ we 
build to consistently use port $0$ for the first edge we choose
incident to any $\ell\in L\setminus\{s\}$, and port $1$ for the
second edge we choose incident to that $\ell$. For $\ell=s$, to meet
~\ref{eq:f-s}, we treat a chosen edge $\{s,u\}$ with $u\neq t$ 
to use port $1$, and the second edge $\{s,t\}$ to use port $0$. 
To make sure that we indeed pick $\{s,t\}$ as one of the two
edges incident to $s$, we can use an extra global state bit.

One only needs four states per vertex to keep track of partial solutions. 
For each vertex in the original graph $G$, that now represents two 
vertices in the graph $F_\lambda$ under a partial coloring of the vertices 
in $L$, we have the states 
\begin{enumerate}
\item[$(\sigma_{00})$] 
  we have not matched this vertex,
\item[$(\sigma_{01})$] 
  we have matched this vertex once, and we have colored it $0$ (blue),
\item[$(\sigma_{10})$] 
  we have matched this vertex once, and we have colored it $1$ (orange),
\item[$(\sigma_{11})$] 
  we have matched this vertex twice.
\end{enumerate}
We do however need to handle vertices in $L$ and $R$ differently. 
If we have a vertex $\ell\in L$ that is matched once with color $c$, 
we have to make sure that the next matching edge uses the same color $c$ 
for the vertex $\ell$. That is, we have already assigned a color to 
this vertex with $\lambda$ and we need to stick to it. But for 
a vertex $r\in R$ that is matched once with color $c$, we have to make 
sure the next matching edge uses the {\em opposite} color $\bar c=1-c$,
as we seek two edges that match differently colored vertices in $R$ 
in our perfect matching in $F_\lambda$. We also only use the states $\sigma_{00},\sigma_{01},$ and $\sigma_{11}$ for the vertex $s$ to make sure it gets the color $0$ (blue).

\subsection{Dynamic programming over a path decomposition}

We will now see how we can solve minimum TSP over a path decomposition 
in Theorem~\ref{thm:TSP-pw}. 

\thmtsppw*

For a variable $A$, let $A \mineq B$ be shorthand for $A:=B$ if $A$ was 
not previously assigned a value, and $A:=\min(A,B)$ otherwise. 
Similarly, let $A \acceq B$ be shorthand for $A:=B$ if $A$ was not 
previously assigned a value, and $A:=A+B$ otherwise.

As input we are given an edge-weighted graph $G$ with an 
edge weight function $w:E(G)\rightarrow \mathbb{Q}_{>0}$, and a path 
decomposition $T_\pd$ of width $\operatorname{pw}(G)$.  
We will compute the optimum TSP tour from below up to the root node 
$r$ of $T_\pd$. We say a node $y$ is \emph{below} another node $x\neq r$ 
if the path from $r$ to $y$ in $T_\pd$ goes through $x$.
We define for each vertex $x\in V(T_\pd)$ a 
map $M_x:\{\sigma_{00},\sigma_{01},\sigma_{10},\sigma_{11}\}^{B_x}\rightarrow \mathbb{Q}_{>0}$.
The mapping captures for each set of states $S$ on the vertices in 
the bag $B_x$, the smallest possible sum of edge weights for any matching 
such that 
\begin{enumerate}
\item 
all vertices that are not in the bag $B_x$, but are in some bag for some 
node $y$ below $x$, have been matched to the state $\sigma_{11}$,
\item 
all vertices in the bag $B_x$ have been matched to the corresponding state in $S$, and
\item 
all remaining vertices are unmatched and have state $\sigma_{00}$.
\end{enumerate}

We evaluate the nodes in the path decomposition one at a time, 
starting with the leaf node on the other side of the root. The last node 
processed is the root (the leaf node on the other side of the path). 
If $\emptyset \in M_r$, we return $M_r(\emptyset)$ as the length of 
the shortest traveling salesperson tour through the edge $e$, otherwise 
we return that there is not any tour through $e$. We update the node's 
states as follows according to type of node.
Let $y$ be the current node and $x$ its predecessor (if any):
\begin{enumerate}
\item
{\em Leaf node.}
If $y$ is not the root, we set $M_y(\emptyset):=0$. If $y$ is the root, we set $M_y=M_x$ with $x$ the forget bag adjacent to
this leaf node $y$, 
\item
{\em Introduce vertex node.}
$B_y=B_x\cup\{v\}$.
For every state $S\in M_x$, we add the state $S'=S+\{v\leftarrow \sigma_{00}\}$ to $M_y$ with $M_y(S')\mineq M_x(S)$.
\item
{\em Introduce edge node.}
$\{\ell,r\}\in B_y=B_x$, with $\ell\in L$ and $r\in R$.
For every state $S\in M_x$, we first add the state $S \in M_y$ 
with $M_y(S) \mineq M_x(S)$. Next, we add the state $S'$ to $M_y$ for 
every matching pair in Table~\ref{table:edge-states} where we have replaced 
the states for $\ell$ and $r$ in $S$ with the ones given in the tables. 
We set $M_y(S') \mineq M_x(S)+w(\{\ell,r\})$ to account for the weight 
of the added edge. Note that there are different state transitions 
depending on whether or not $(\ell,r)\in E(\oa G_e)$.

\begin{table}[]
 \caption{Mapping from vertex states for the endpoints of an edge $\{\ell,r\}$ to the new states, if any. Left: The states' mapping when $(\ell,r)\notin E(\oa G_e)$. Right: The states' mapping when $(\ell,r)\in E(\oa G_e)$.} 
  \label{table:edge-states}

\begin{tabular}{l || c | c | c | c}
 $\ell\backslash r$ & $\sigma_{00}$ & $\sigma_{01}$ & $\sigma_{10}$ & $\sigma_{11}$ \\
 \hline
 \hline
 $\sigma_{00}$  & $(\sigma_{01},\sigma_{01})$ & $(\sigma_{10},\sigma_{11})$ & $(\sigma_{01},\sigma_{11})$  & - \\
 & $(\sigma_{10},\sigma_{10})$ & & & \\
 \hline
 $\sigma_{01}$  & $(\sigma_{11},\sigma_{01})$  & -  & $(\sigma_{11},\sigma_{11})$  & -  \\
 \hline
 $\sigma_{10}$ & $(\sigma_{11},\sigma_{10})$  & $(\sigma_{11},\sigma_{11})$ & -  & -  \\
 \hline
 $\sigma_{11}$ & - & - & - & - 
\end{tabular}
\qquad
\begin{tabular}{l || c | c | c | c}
 $\ell\backslash r$ & $\sigma_{00}$ & $\sigma_{01}$ & $\sigma_{10}$ & $\sigma_{11}$ \\
 \hline
 \hline
 $\sigma_{00}$  & $(\sigma_{01},\sigma_{10})$ & $(\sigma_{01},\sigma_{11})$ & $(\sigma_{10},\sigma_{11})$  & - \\
 & $(\sigma_{10},\sigma_{01})$ & & & \\
 \hline
 $\sigma_{01}$  & $(\sigma_{11},\sigma_{10})$  & $(\sigma_{11},\sigma_{11})$  & -  & -  \\
 \hline
 $\sigma_{10}$ & $(\sigma_{11},\sigma_{01})$  & - & $(\sigma_{11},\sigma_{11})$   & -  \\
 \hline
 $\sigma_{11}$ & - & - & - & - 
\end{tabular}
\end{table}

\item
{\em Forget vertex node.}
$B_y=B_x\setminus\{v\}$.
For each state $S\in M_x$ such that $S(v)=\sigma_{11}$, add $S'=S-\{v\}$ 
to $M_y$ with $M_y(S')\mineq M_x(S)$.
\end{enumerate}
This concludes the description of the algorithm and 
the proof of Theorem~\ref{thm:TSP-pw}. 

We next turn to the proof of Theorem~\ref{thm:chc-pw}.

\thmchcpw*

To modify the algorithm to count Hamiltonian cycles in 
Theorem~\ref{thm:chc-pw}, we let $M$ map to the ring $\mathbb{Z}[\xi]$ 
instead of $\mathbb{Q}_{>0}$ with $\xi$ a formal polynomial indeterminate.  
We set the initial state to $M(\emptyset)=1$ and replace 
$M(S')\mineq M(S)+w(\{\ell,r\})$ in the introduce-edge nodes with
$M(S')\acceq \xi M(S)$, and also replace the equality operator 
$\mineq$ with $\acceq$ in all steps. Finally, we look at the 
coefficient of $\xi^n$ in $M_r(\emptyset)$ 
to see how many Hamiltonian cycles go through the edge $e$.
Summing over all edges $e$ and dividing by $n$ gives us the final count.

Note that since there are at most $4^{\operatorname{pw}(G)+1}$ states 
to consider at each bag, the running time follows.
This completes the proof of Theorem~\ref{thm:chc-pw}.

\subsection{Dynamic programming over a branch decomposition}

We will now see how we can solve minimum TSP over a branch decomposition 
to obtain Theorem~\ref{thm:TSP-bw}.

\thmtspbw*

Given an edge weighted graph $G$ with weights $w:E\rightarrow \mathbb{Q}_{>0}$ 
and a branch decomposition $T_\bd$ of width $\operatorname{bw}(G)$ as input, 
we first insert a new root node $r$ in $T_\bd$ at an arbitrary edge 
$\{u,v\}\in E(T_\bd)$ by subdividing it as $\{u,r\}$ and $\{r,v\}$. 
We will compute the optimum TSP tour from the leaves up to the root $r$ 
in $T_\bd$.
We say an edge $y\in E(T_\bd)$ is \emph{below} another edge $x\in E(T_\bd)$
with $x\neq r$ if the path joining $r$ and $y$ in $T_\bd$ goes through $x$.
We define for each edge $x\in E(T_\bd)$ 
a map $M_x:\{\sigma_{00},\sigma_{01},\sigma_{10},\sigma_{11}\}^{B_x}\rightarrow \mathbb{Q}_{>0}$.
The mapping captures for each set of states $S$ on the vertices in 
the middle set $B_x$, the smallest possible sum of edge weights for 
any matching such that 
\begin{enumerate}
\item 
all vertices that are not in the middle set $B_x$, but are in some 
middle set for some node $y$ below $x$, have been matched to the state $\sigma_{11}$,
\item 
all vertices in the middle set $B_x$ have been matched to the 
corresponding state in $S$, and
\item 
all remaining vertices are unmatched and have state $\sigma_{00}$.
\end{enumerate}

For any internal node $v\neq r$ in $T_\bd$ we have three incident 
edges $x$, $y$, and $z$; let $z$ be the unique edge closest 
to $r$ in $T_\bd$. We define the four pairwise disjoint sets
\begin{enumerate}
\item the Left set: $L_v=(B_x \cap B_z) \setminus B_y$, 
\item the Right set: $R_v=(B_y\cap B_z) \setminus B_x$,
\item the Forget set: $F_v=(B_x\cap B_y) \setminus B_z$, and
\item the Intersection set: $I_v=B_x\cap B_y \cap B_z$.
\end{enumerate}

By identifying the sets 
$B_x=L_v\cup F_v \cup I_v$, 
$B_y=R_v\cup F_v \cup I_v$,
and 
$B_z=L_v \cup R_v \cup I_v$, we get a bound of $\operatorname{bw}(G)$ 
on the size of each of these three unions. In particular
\begin{equation}
\label{eq:middle-set-bound}
|I_v|+|L_v|+|R_v|+|F_v|\leq \frac{3}{2}\operatorname{bw}(G).
\end{equation}

\begin{table}[]
 \caption{Matching states from $B_x$ and $B_y$ in the forget set marked with an ``x''. Left: For vertices in $L$. Right: For vertices in $R$.} 
  \label{table:states-in-forget-sets}
\begin{tabular}{l || c | c | c | c}
 $L$ & $\sigma_{00}$ & $\sigma_{01}$ & $\sigma_{10}$ & $\sigma_{11}$ \\
 \hline
 \hline
 $\sigma_{00}$  &  &  &  &  x\\
 \hline
 $\sigma_{01}$  &  & x  &   &  \\
 \hline
 $\sigma_{10}$ &  &  & x  &  \\
 \hline
 $\sigma_{11}$ & x &  &  &  
\end{tabular}
\qquad \qquad
\begin{tabular}{l || c | c | c | c}
  $R$ & $\sigma_{00}$ & $\sigma_{01}$ & $\sigma_{10}$ & $\sigma_{11}$ \\
 \hline
 \hline
 $\sigma_{00}$  &  &  &  & x  \\
 \hline
 $\sigma_{01}$  &  &  &  x &  \\
 \hline
 $\sigma_{10}$ &  & x &  &  \\
 \hline
 $\sigma_{11}$ & x &  &  & 
\end{tabular}
\end{table}

We compute the mapping $M$ for all edges in the branch decomposition tree 
in order from furthest to nearest the root $r$.
For any leaf node $v$ and its only incident edge $z$ representing 
an edge $\{\ell,r\}\in E(G)$, we 
set $M_z(\ell\leftarrow \sigma_{00},r\leftarrow \sigma_{00})=0$ 
and either
\begin{enumerate}
\item 
$M_z(\ell\leftarrow \sigma_{01},r\leftarrow \sigma_{01})=M_z(\ell\leftarrow \sigma_{10},r\leftarrow \sigma_{10})=w(\{\ell,r\})$ if $(\ell,r)\notin E(\oa G_e)$, or
\item 
$M_z(\ell\leftarrow \sigma_{01},r\leftarrow \sigma_{10})=M_z(\ell\leftarrow \sigma_{10},r\leftarrow \sigma_{01})=w(\{\ell,r\})$ if $(\ell,r)\in E(\oa G_e)$.
\end{enumerate}

For any internal node $w$ with three edges $x,y,$ and $z$, with $z$ closest 
to $r$, we do the following:
For every state $S_x\in M_x$ and every state $S_y\in M_y$ such that 
vertices in $I_v$ have the same state in $S_x$ and $S_y$, and vertices 
in $F_v$ have \emph{matching} states in $S_x$ and $S_y$ according to 
Table~\ref{table:states-in-forget-sets}, we define a new state $S'$ by
\begin{enumerate}
\item 
taking $I_v$ vertices' state values from both $S_x$ and $S_y$ (as they agree), 
\item 
taking $L_v$ vertices' state values from $S_x$, and
\item 
taking $R_v$ vertices' state values from $S_y$. 
\end{enumerate}

We set $M_z(S')\mineq M_x(S_x)+M_y(S_y)$.

Finally, for the two edges $x$ and $y$ incident to the root $r$, we list 
all pairs of states with $S_x\in M_x$ and $S_y\in M_y$ such that all 
vertices match according to Table~\ref{table:states-in-forget-sets}. 
The minimum of $M_x(S_x)+M_y(S_y)$ over all such pairs is the length of 
the shortest Hamiltonian cycle in $G$.

This completes the description of the algorithm in Theorem~\ref{thm:TSP-bw}. 
We observe that the running time is bounded by 
$4^{|I_v|+|L_v|+|R_v|+|F_v|}$ at every internal node $v$. This 
is at most $8^{\operatorname{bw}(G)}$ by the bound 
in \eqref{eq:middle-set-bound}. The computation needed at the root vertex 
is at most $4^{\operatorname{bw}(G)}\operatorname{poly}(n)$ if we iterate 
over all states $S_x\in M_x$ and compute what the unique $S_y\in M_y$ is 
that matches $S_x$ and look it up in an efficient dictionary structure 
for $M_y$. This completes the proof of Theorem~\ref{thm:TSP-bw}.

We next turn to the proof of Theorem~\ref{thm:chc-bw}.

\thmchcbw*

To modify the algorithm to count Hamiltonian cycles in 
Theorem~\ref{thm:chc-bw}, we again let $M$ map to $\mathbb{Z}[\xi]$ with
$\xi$ a formal polynomial indeterminate.
We set the initial states at the leaves to $M(\emptyset)=1$ and either
\begin{enumerate}
\item $M_z(\ell\leftarrow \sigma_{01},r\leftarrow \sigma_{01})=M_z(\ell\leftarrow \sigma_{10},r\leftarrow \sigma_{10})=\xi$ if $(\ell,r)\notin E(\oa G_e)$, or
\item $M_z(\ell\leftarrow \sigma_{01},r\leftarrow \sigma_{10})=M_z(\ell\leftarrow \sigma_{10},r\leftarrow \sigma_{01})=\xi$ if $(\ell,r)\in E(\oa G_e)$.
\end{enumerate}

We also replace $M_z(S')\mineq M_x(S_x)+M_y(S_y)$ for $M_z(S')\acceq M_x(S_x)M_y(S_y)$ at the internal nodes.
Finally, for the root $r$ we compute the sum $T$ over all pairs of 
states $S_x\in M_x$ and $S_y\in M_y$ such that vertices match according 
to to Table~\ref{table:states-in-forget-sets} of $M_x(S_x)M_y(S_y)$. 
The coefficient of $\xi^n$ in $T$ is the number of Hamiltonian cycles 
through $e$.
Summing this contribution over all edges $e$ and dividing by $n$ gives us 
the final count.

Contrary to the TSP case, there is a known technique to speed up the 
counting computation at every internal node proposed by Dorn~\cite{Dorn2006}. 
As we in this case want to compute a sum--product formula for each triple of 
states in $I_v,L_v$, and $R_v$ at an internal node $v$, we can invoke 
fast matrix multiplication to get the counts faster than through 
a na\"ive enumeration. For each state $\iota \in I_v$, we define two 
matrices $A_{v,\iota}$ that is a $4^{|L_v|}\times 4^{|F_v|}$ matrix with 
rows representing states in $L_v$ and columns states in $F_v$, 
and $B_{v,\iota}$ that is a  $4^{|F_v|}\times 4^{|R_v|}$ matrix with 
columns representing states in $R_v$ and rows representing states in $F_v$. 
However, the rows are \emph{permuted} so that matching states 
according to Table~\ref{table:states-in-forget-sets} are paired up. 
That is, the forget set at column index $i$ in $A_{v,\iota}$ and the 
forget set at row $i$ in $B_{v,\iota}$ are such that each pair of vertex 
states match. This enables us to compute the matrix product 
$C_{v,\iota}=A_{v,\iota}B_{v,\iota}$ to get the 
state $\iota \cup \lambda \cup \rho$ by mapping it to the 
entry $C_{v,\iota}[\lambda,\rho]$. Dorn~\cite{Dorn2006} proved 
that the worst computation times occurs when 
$|L_v|=|R_v|=|F_v|=\frac{\operatorname{bw}(G)}{2}$, which gives us 
the running time in Theorem~\ref{thm:chc-bw}.

Finally, we turn to Corollary~\ref{cor:TSP-planar} 
and Corollary~\ref{cor:chc-planar}, restated below for convenience.

\cortspplanar*

\corchcplanar*

These follow from Theorem~\ref{thm:TSP-bw} and~\ref{thm:chc-bw}, 
respectively, after applying a bound on the branch width for planar 
graphs by Fomin and Thilikos~\cite{FominT2006}. 
They prove $\operatorname{bw}(G)<2.122\sqrt{n}$ in a planar graph, and 
as mentioned above, such a branch decomposition can be found in 
polynomial time. This concludes the description of our algorithmic results.

\section*{Acknowledgment}

\noindent
We thank the anonymous reviewers of an earlier version of this manuscript 
for their comments, in particular for pointing out~\cite{GorskySW2023}. \ifdefined\A \else AB is supported by the VILLUM Foundation, Grant 16582. JN is supported by the project CRACKNP that has received funding from the European Research Council (ERC) under the European Union's Horizon 2020 research and innovation programme (grant agreement No 853234).\fi

\newpage
\appendix

\begin{center}
\textsc{\large Appendix}
\end{center}

\section{NP-completeness in planar graphs}
\label{sect: NP-hard}
In this section we prove that the degree requirement is necessary in Theorem~\ref{thm:ahc} and~\ref{thm:ahc-log}, i.e., when we drop the minimum 
degree constraint, just detecting if the graph has another Hamiltonian cycle becomes NP-hard.

\begin{theorem}[Hardness of detecting another Hamiltonian cycle]
\label{thm:hard}
Given a bipartite Pfaffian graph $G$ and a Hamiltonian cycle $H$ in $G$
as input, it is NP-complete to decide whether $G$ has a Hamiltonian 
cycle $H'\neq H$.
\end{theorem}

The problem is clearly in NP as one can guess a Hamiltonian cycle in polynomial time and 
check that it is a Hamiltonian cycle and that it differs from the given one.
Our hardness--reduction borrows several parts of the original proof by Garey, Johnson, and Tarjan~\cite{GareyJT1976} that \textsc{Hamiltonian Cycle} is NP-complete in planar graphs. In particular we use the xor-gadget and the cross-over construction from these reductions, see Figure~\ref{fig:xor} and Figure~\ref{fig:uncrossing}. We do however reduce from another problem than \textsc{SAT} as we find it easier to argue one-to-one correspondence between a Hamiltonian cycle and a solution to our NP-hard problem.
We reduce from the \textsc{Exact Set Cover} problem, that given an $n$-element universe $U$ and a family $\mathcal{F}$ of subsets of $U$ asks if there is a subset $\mathcal S\subseteq \mathcal{F}$ such that $\dot{\cup}_{S\in \mathcal S} S=U$. Karp~\cite{Karp72} showed that \textsc{Exact Set Cover} is NP-complete.
Given an instance $(U,\mathcal F)$ to  \textsc{Exact Set Cover}, we build a slightly larger instance $(U,\mathcal F')$, with $\mathcal F'=\mathcal F \bigcup \{U\}$, that is, we add a subset covering all of the universe by itself. With foresight, we will encode $(U,\mathcal F')$ in a graph so that every solution to the  \textsc{Exact Set Cover} problem corresponds to a unique Hamiltonian cycle in the graph. We will then provide the Hamiltonian cycle corresponding to the solution $\mathcal S=\{U\}$ and ask for another Hamiltonian cycle. Any other Hamiltonian cycle will then by definition correspond to a solution to the original   \textsc{Exact Set Cover} instance $(U,\mathcal F)$.

A common component in our gadgets is the forced edge, drawn as a thick black edge in the figures. It is encoded by a path on three edges with two intermediate additional degree-two vertices.
We will encode the condition for each $u\in U$ that exactly one subset in $F$ covers it with a gadget called avoid-one. The gadget encodes a set of \emph{choices} $c_1,\ldots,c_k$. It is depicted in Figure~\ref{fig:avoid-one}.

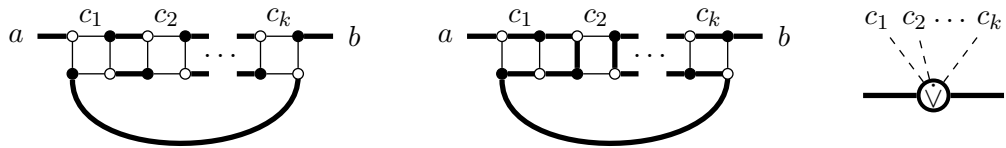
\begin{figure}[!ht]
\[
\begin{tikzpicture}[scale=0.5,shorten >=1pt, auto, node distance=1cm, ultra thick]
\avoidone
\end{tikzpicture}
\qquad 
\begin{tikzpicture}[scale=0.5,shorten >=1pt, auto, node distance=1cm, ultra thick]
\avoidoneex
\end{tikzpicture}
\qquad 
\begin{tikzpicture}[scale=0.5,shorten >=1pt, auto, node distance=1cm, ultra thick]
\avoidoneschematic
\end{tikzpicture}
\]
 \caption{The avoid-one gadget. Left: The gadget itself. Thick edges must be used by a Hamiltonian cycle. This is accomplished by replacing each thick edge with a path on three edges and two additional degree-two vertices. The key property of the gadget is that in any Hamiltonian cycle in our final construction, exactly one of the $k$ choice edges $c_1,c_2,\ldots, c_k$ is not part of the cycle. Middle: Any Hamiltonian cycle enters and exits only at the terminals $a$ and $b$ avoiding exactly one of the choices, here $c_2$. It will follow from the rest of the construction that the Hamiltonian cycle could not enter and exit also via a choice edge, since that part would then form a cycle of its own. Right: A schematic version of the gadget used in subsequent figures.} 
  \label{fig:avoid-one}
\end{figure}

We also need the xor-gadget that given a pair of edge-connected vertex pairs $a_1,a_2$ and $b_1,b_2$ ensures that a Hamiltonian cycle either goes through $a_1,a_2$ or $b_1,b_2$ but not both, see Figure~\ref{fig:xor}.

\begin{figure}[!ht]
\[
\begin{tikzpicture}[scale=0.5,shorten >=1pt, auto, node distance=1cm, ultra thick]
\xor
\end{tikzpicture}
\qquad \qquad \qquad
\begin{tikzpicture}[scale=0.5,shorten >=1pt, auto, node distance=1cm, ultra thick]
\xorschematic
\end{tikzpicture}
\]
 \caption{The xor-gadget. Left: The gadget itself. Thick edges must be used by a Hamiltonian cycle. The key property of the gadget is that in any Hamiltonian cycle in our final construction, either the cycle enters and leaves through $a_1$ and $a_2$ or through $b_1$ and $b_2$, but not both. Right: A schematic version of the gadget used in subsequent figures.} 
  \label{fig:xor}
\end{figure}
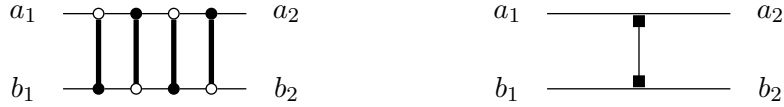

Our instance graph will have one strain of subset choice gadgets at the top, and one strain of avoid-one gadgets for each element in the universe $U$ at the bottom. The two strains are connected with each other at both the left and right end. Between the two strains run several xor-gadgets that ensure consistence of the subset selection and the universe cover. 
See Figure~\ref{fig:instance}.
Note that the xor-gadgets may cross each other but any such crossing can be resolved by a well-known uncrossing gadget from~\cite{GareyJT1976}, see Figure~\ref{fig:uncrossing}. This completes our construction.

\begin{figure}[!ht]
\[
\begin{tikzpicture}[scale=0.5,shorten >=1pt, auto, node distance=1cm, ultra thick]
\instance
\end{tikzpicture}
\]
 \caption{An instance corresponding to the \textsc{Exact Set Cover} instance $U=\{u_1,u_2,\ldots,u_6\}$ and $\mathcal{F}=\{S_1,S_2,\ldots,S_6\}$. Only two subsets $S_3=\{u_1,u_4,u_5\}$ and $S_4=\{u_2,u_3,u_4\}$ are drawn for clarity.} 
  \label{fig:instance}
\end{figure}
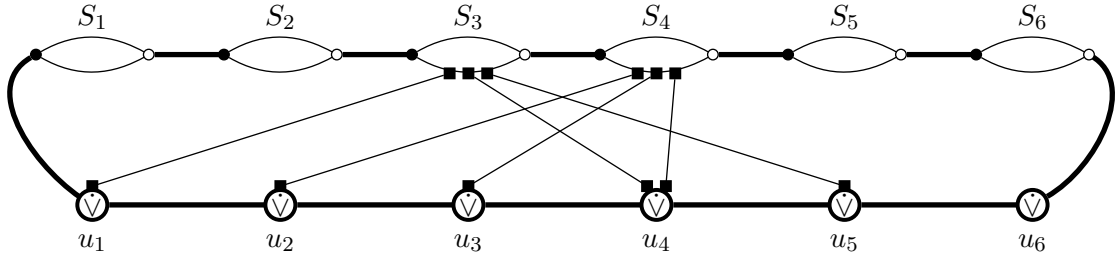

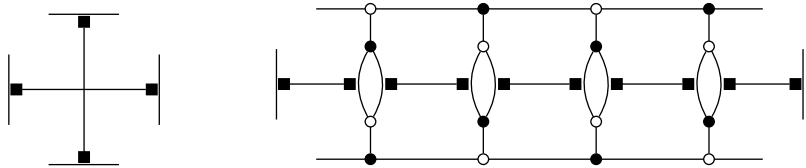
\begin{figure}[!ht]
\[
\begin{tikzpicture}[scale=0.5,shorten >=1pt, auto, node distance=1cm, ultra thick]
\xorcrossing 
\end{tikzpicture}
\qquad \qquad
\begin{tikzpicture}[scale=0.5,shorten >=1pt, auto, node distance=1cm, ultra thick]
\xoruncrossing
\end{tikzpicture}
\]
 \caption{Crossing and uncrossing. Left: Two crossing xor-gadgets. Right: Uncrossing via the cross-over construction from~\cite{GareyJT1976}.}
\label{fig:uncrossing}
\end{figure}

We now proceed with a proof Theorem~\ref{thm:hard}.
We will first argue that the set of solutions to the \textsc{Exact Cover} instance $(U,\mathcal F')$ corresponds one-to-one to the Hamiltonian cycles in the constructed graph. First consider any solution $\mathcal{S}\subseteq \mathcal{F}'$ to the \textsc{Exact Cover} instance $(U,\mathcal F')$. We can construct a Hamiltonian cycle through the constructed graph by taking the upper path (c.f. Figure~\ref{fig:instance}) on each subset choice gadget for each subset $S_i\in \mathcal{F}\setminus \mathcal{S}$, and the lower path for each subset choice gadget $S_i\in \mathcal{S}$, including covering all vertices in the associated xor-gadgets. Since this is a solution to the \textsc{Exact Cover}  problem, this part of the cycle will block exactly one of the choice edges in each avoid-one gadget with the corresponding xor-gadget, and we can take the unique path from $a$ to $b$ in the gadget that passes every vertex of the gadget and avoids using the one choice edge which is blocked from the subset choice gadget string side.

In the other direction, any Hamiltonian cycle must cover all choice edges in an avoid-one gadget except precisely one, since the xor-gadgets assure us that we cannot have paths entering and exiting at both sides. That means in particular that it is impossible for the Hamiltonian cycle to leave an avoid-one gadget in a choice edge to continue in a subset choice gadget (and vice versa). The path has to return to the other end of the choice edge. The omitted choice edge in each avoid-one gadget must be covered from the subset choice gadget string side, which forces the Hamiltonian cycle to take the lower path exactly for the subsets needed to cover $U$ in a solution to the \textsc{Exact Cover} instance $(U,\mathcal F')$.

The proof now follows by providing the unique Hamiltonian cycle that represents the planted solution $\mathcal S=\{U\}$ and ask for another one. By the above one-to-one relationship, that other Hamiltonian cycle encodes a solution to $(U,\mathcal F)$, which is NP-hard to find.
This completes the proof of Theorem~\ref{thm:hard}. \qedwhite

\newpage
\bibliographystyle{alphaurl}
\bibliography{pfaffian}

\end{document}